\documentclass[12pt,reqno]{amsart}
\usepackage{amsmath,amsfonts,amssymb,amscd,amsthm,amsbsy,epsf}
\usepackage[dvips]{graphicx}
\usepackage[usenames]{xcolor}
\usepackage{comment}
\usepackage[backend=bibtex, maxnames=3, minnames=1]{biblatex}
\usepackage[english]{babel}
\usepackage{marginnote}
\footskip=18pt \numberwithin{equation}{section}
\newtheorem{theorem}{Theorem}
\newtheorem{meta-thm}[theorem]{Meta-Theorem}

\newtheorem{proposition}[theorem]{Proposition}

\newtheorem{remark}[theorem]{Remark}
\newtheorem{definition}[theorem]{Definition}

\DeclareMathAlphabet{\mathcalligra}{T1}{calligra}{m}{n}

\newcommand\beq[1]{ \begin{equation}\label{#1} }
	\newcommand{\eeq}{ \end{equation} }
\newcommand{\beqno}{ \[ }
\newcommand{\eeqno}{ \] }
\newcommand\beqa[1]{ \begin{eqnarray} \label{#1}}
	
	\newcommand{\eeqa}{ \end{eqnarray} }

\newcommand{\beqano}{ \begin{eqnarray*} }
	
	\newcommand{\eeqano}{ \end{eqnarray*} }

\newcommand\equ[1]{{\rm (\ref{#1})}}

\def\P{{\mathcal P}}
\def\E{{\mathcal E}}

\def\G{{\mathcal G}}
\def\H{{\mathcal H}}

\def\P{{\mathcal P}}

\def\integer{{\mathbb Z}}

\def\real{{\mathbb R}}
\def\torus{{\mathbb T}}

\begin{document}

	\title[Stability and bifurcations of resonances in ring's dynamics]
	{Stability and bifurcations of resonances in ring's dynamics}
	
	\author[A. Celletti]{Alessandra Celletti}
	\address{
		Department of Mathematics, University of Roma Tor Vergata, Via
		della Ricerca Scientifica 1, 00133 Roma (Italy)}
	\email{celletti@mat.uniroma2.it}
	\author[I. De Blasi]{Irene De Blasi}
	\address{
		Department of Mathematics, University of Torino, via Carlo Alberto 10, 10123 Torino (Italy)}
	\email{irene.deblasi@unito.it}
	\author[S. Di Ruzza]{Sara Di Ruzza}
	\address{
Department of Mathematics and Informatics, University of Palermo, Via Archirafi 34, 90123 Palermo (Italy)}
	\email{sara.diruzza@unipa.it}
	
	\baselineskip=18pt              
	
	\begin{abstract}
We use perturbation theory and bifurcation theory to analyze the dynamical behavior of resonances, associated to a model describing a particle moving within a ring around a celestial object. The central body is modeled as a homogeneous triaxial ellipsoid, rotating about its shortest physical axis at a constant angular velocity. It is assumed that the massless ring particle moves within the equatorial plane of the ellipsoid. The dynamics of the particle is studied using epicyclic variables, that lead to a straightforward definition of corotation and Lindblad resonances. These resonances are associated to a Hamiltonian function with two degrees of freedom, for which we compute appropriate expansions for the normal form and the resonant Hamiltonian. Initially, the normal form is verified to be non--degenerate, thereby guaranteeing the existence of invariant KAM tori, providing the stability of the resonances, through their confinement in phase space. Subsequently, two exemplary test cases are examined: a nearly spherical ellipsoid and a highly aspherical ellipsoid. Furthermore, this study concentrates on three principal resonances: corotation, $1:2$, and $1:3$, for which we present  results concerning their dynamical behavior obtained analyzing the Hamiltonian formulation of the model and the resonant normal form. Specifically, we examine the phase space structure, the amplitude of libration around the resonances, and the occurrence of bifurcations. Remarkably, in none of the two studied test cases the $1:3$ resonance presents evidence of bifurcations for relevant values of the  eccentricity. Our dynamical study thus supports a greater probability of selecting the $1:3$ resonance in comparison to the other resonances. 
	\end{abstract}

\maketitle


	
	\section{Introduction}\label{sec:introduction}

The analytical study of planetary rings is a traditional subject, which has been investigated in great detail in several works; among the others, we quote \cite{Borderies1987,Goldreich1982,Murray1994,Sicardy2006,Tiscareno2013} and we refer the reader to the references therein. 

The recent discoveries of rings surrounding small bodies within the Solar system, for example Chariklo, Haumea and Quaoar (see \cite{Braga2014,Morgado2023,Ortiz2017,Pereira2023,Winter2019}), significantly expanded the field of research in rings dynamics (see, e.g., \cite{W2023,Marzari2020,Sicardy2019,Sicardy2020} ) and raised interesting questions from the dynamical systems point of view.

There is a fundamental difference between the dynamics of rings around planets and that around small celestial bodies. 
While planets have an almost spherical shape, small bodies of the Solar system, like dwarf planets or Centaurs, might have a marked irregular shape, whose non--axisymmetric gravity field generates a variety of dynamical behaviors. Furthermore, although space missions have yielded comprehensive insights into the structure of planets, the shapes of certain small celestial bodies continue to warrant exploration, and fundamental quantities such as the approximate dimensions of these bodies remain largely undisclosed. Consequently, the  dynamics of rings around small bodies needs the exploration of various models to accurately approximate their shapes. 
As frequently found in the literature, a small body can be approximated as a triaxial ellipsoid or alternatively as a sphere with relatively minor topographical features; the latter case is often called the \sl mass anomaly \rm model. A prolate ellipsoid and a mass anomaly model have been studied in \cite{Madeira2022} and \cite{Ribeiro2023}, using Poincar\'e surfaces of section with the aim to locate the resonances or to estimate the size of regular and chaotic regions.

The model considered in this work is that of a homogeneous body with the shape of a triaxial ellipsoid; we assume that the ellipsoid rotates with constant angular velocity around its shortest physical axis. To account for different shapes, from planets to small bodies, we consider two sample cases: an almost spherical body (with small values of oblateness and elongation) and a highly aspherical body (with high values of oblateness and elongation). We analyze the dynamics of a single particle, that we assume to belong to a ring system around the triaxial body. We constrain the particle to move on the equatorial plane. With these assumptions, we derive the potential in terms of the polar coordinates, namely, the radius and the true longitude. Afterwards, we split the potential into the symmetric and non--symmetric parts, and we introduce the two (classical) frequencies, called synodic and epicyclic frequencies. Corotation and Lindblad  resonances are defined in terms of these two frequencies. In the corotation case the orbital mean motion of the particle is equal to the frequency of rotation of the central body. A Lindblad resonance of order $p:q$ (with $p$, $q$ integers) has the following meaning: within a reference frame rotating with the particle precession rate, the particle makes $p$ orbits, while the body completes $q$ rotations. Next, we introduce a system of coordinates which is well suited to study the dynamics, namely the epicyclic coordinates. We provide a 2 DOF Hamiltonian function, which is conveniently written in terms of action-angle variables associated to the epicyclic coordinates; in particular, the angle variables have frequency equal to the epicyclic and synodic frequencies. 

Our first task consists in checking Kolmogorov's non--degeneracy condition of the normal form (namely, the integrable part of the Hamiltonian under different approximations), to ensure the existence of KAM invariant tori (for more details, see \cite{ArnoldKAM,kolmogorov1954conservation,  moser1962invariant}); in a 2 DOF system, like the one under study, such tori provide stability results on fixed energy levels, since they yield a confinement of the resonant motion, ensuring its stability for infinite times. 

Using the averaging method, we proceed to introduce the resonant Hamiltonian, defined as the sum of the normal form (i.e., the Hamiltonian part that does not depend on the angles) and the resonant part (which depends on the resonant -- either Lindblad or corotation -- combination of the angles). We consider the two sample cases of almost spherical and highly aspherical body, and we compute the resonant Hamiltonian for three main case studies: corotation or $1:1$ resonance (which turns out to be located very close to the central body), $1:3$ resonance (which is far from the central body) and $1:2$ resonance (whose location is between corotation and $1:3$ resonance). For each resonance and for a constant value of the particle's eccentricity, orbits are shown in the phase space, and the amplitude of libration of the resonance is quantified, with consideration for the oblateness and elongation of the central body. The findings indicate that these shape factors significantly impact the amplitude of libration, thereby establishing a distinct disparity between celestial bodies such as planets, which are nearly spherical, and smaller bodies, which are highly aspherical. Notably, this behavior is predominantly influenced by the elongation factor.

The dynamical study proceeds with the computation of the equilibrium points, the linear stability analysis and the occurrence of bifurcations. We found that $1:1$ and $1:2$ resonances undergo  pitchfork or saddle-node bifurcations as the eccentricity varies. Quite surprisingly, the $1:3$ resonance does not present bifurcations within the considered range of eccentricity, thus making such resonance more likely than the other resonances. 

It is important to note that the model currently under review serves as merely an approximation of the actual physical scenario. In a more precise approach, alternative models might have been employed to more accurately depict the shape of the central body, or the gravitational effects of potential satellites of the central body could have been included. Regardless, our findings substantiate the conclusion that the central body's potential governs a distinct behavior of the resonances, independently of any additional factors.

\vskip.1in 

This work is organized as follows. In Section~\ref{sec:potential}, we describe the derivation of the potential function. In Section~\ref{sec:epicyclic} we introduce epicyclic variables, define the resonances, provide the Hamiltonian function and its series expansion, discuss the non--degeneracy condition of the resulting Hamiltonian. In Section~\ref{sec:test}, we present the results for the two test cases of an almost spherical and a highly aspherical body, and the three resonances given by $1:1$, $1:2$ and $1:3$ resonances.

\section{Potential of a homogeneous triaxial ellipsoid}\label{sec:potential}
	
Let $\E$ be a body of mass $M_P$, whose shape is assumed to be given by an ellipsoid with semi--axes $a>b>c>0$.
Let $O$ be the center of mass of the ellipsoid $\E$. Assume that $\E$ rotates around the shortest physical axis with constant angular velocity $\Omega_P=2\pi/T_{rot}$,
where $T_{rot}$ is the rotation period. Following \cite{Sicardy2020},
assume that the mass distribution is symmetric with respect to the equatorial plane and that it has a symmetry plane containing the rotation axis.
	
Consider a massless particle $\P$ moving in the equatorial plane and let $(r,L)$ be its polar coordinates in an inertial frame with center in $O$, say $(O,x,y,z)$ (see Figure~\ref{fig1}); $r$ denotes its distance from $O$ and $L$ is the true longitude.	
	\begin{figure}[htbp]
		\centering
		\includegraphics[scale=0.5]{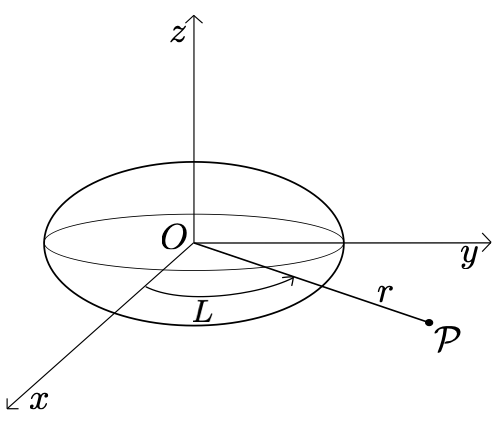}
		\caption{A particle $\P$ moving in the equatorial plane of the ellipsoid with center $O$ in an inertial frame $(O,x,y,z)$; the coordinates of $\P$ are denoted as $(r,L)$.}
		\label{fig1}
	\end{figure}
In a quasi-inertial frame rotating with angular velocity $\Omega_P$, we introduce
	the angle $\theta$ defined as
	$$
	\theta=L-\Omega_P t\ .
	$$
Following \cite{Balmino} and \cite{SicardyNature}, it is possible to derive the three-dimensional gravitational potential generated by the above rotating body. 

\begin{proposition}\label{pro:potential}
The gravitational potential generated by $\mathcal E$, moving around it minor axis with angular velocity $\Omega_P$, is given by 
\beq{U2}
	U(r,\theta,\phi)=-{{\G M_P}\over r}\ \sum_{p=0}^\infty \cos(2p\theta)\ \sum_{\ell=p}^{\infty} \left({R\over r}\right)^{2\ell} C_{2\ell,2p}
	P_{2\ell,2p}(\sin\phi)\ ,
	\eeq
where $\phi$ denotes the latitude of the particle, $\G$ is the gravitational constant, $R$ is the reference length, introduced
	to make the coefficients dimensionless, that we can define as\footnote{Notice that in \cite{SicardyNature}, the reference length is defined as
		$$
		{3\over R^2}={1\over a^2}+{1\over b^2}+{1\over c^2}\ .
		$$
	}
\beq{refradius}
R=\Big({3\over {a^{-2}+b^{-2}+c^{-2}}}\Big)^{1\over 2}\ ,
\eeq
$P_{\ell,p}$ is the associated Legendre polynomial defined as
$$
	P_{\ell,p}(x)=(-1)^p\ {{(1-x^2)^{p\over 2}}\over {2^\ell \ell!}}\ {{d^{\ell+p}}\over {dx^{\ell+p}}}((x^2-1)^\ell)\ ,
$$
while the coefficients $C_{2\ell,2p}$ are given by the following expression (compare with \cite{Boyce}):
	\beq{coeffC}
	C_{2\ell,2p}={3\over R^{2\ell}}\ {{\ell!(2\ell-2p)!}\over {2^{2p}(2\ell+3)(2\ell+1)!}}(2-\delta_{0,p})\
	\sum_{k=0}^{[{{\ell-p}\over 2}]}\
	{{(a^2-b^2)^{p+2k}(c^2-{1\over 2}(a^2+b^2))^{(\ell-p-2 k)}}\over {16^k(\ell-p-2k)!(p+k)!k!}}\ .
	\eeq
\end{proposition}

The expansion \equ{U2} of the gravitational potential is a classical topic of Celestial Mechanics; the proof of Proposition \ref{pro:potential} can be found, e.g., following the results contained in \cite{Balmino} and \cite{SicardyNature}.

    \begin{definition}
We introduce the quantities 
\beq{oblat}
\mathcal O_b =\frac{a^2+b^2-2c^2}{4R^2} \, ,
\eeq
	and
	\beq{elong}
	\E_\ell=\frac{a^2-b^2}{2R^2}
	\eeq
that measure, respectively, the oblateness and the elongation of the body.
\end{definition}

We notice that the coefficients $C_{2\ell,2p}$ depend on oblateness and elongation.

When the particle moves in the equatorial plane, then $\phi=0$; hence, the potential acting on $\P$ can be written as 
\beq{U1}
U(r,\theta)=\sum_{p=0}^\infty U_{2p}(r)\ \cos(2p\theta)\ ,
\eeq
which contains only even terms, since the ellipsoid is symmetric and the potential is invariant
under a rotation of angle $\pi$. 

From Eq.~\equ{U2} with $\phi=0$, we obtain the following expression for the potential and, hence, for the coefficient $U_{2p}(r)$ in Eq.~\equ{U1}: 
\beq{U3}
U(r,\theta)=-{{\G M_P}\over r}\
\sum_{p=0}^\infty \cos(2p\theta)\ \sum_{\ell=p}^\infty \left({R\over
r}\right)^{2\ell} C_{2\ell,2p} P_{2\ell,2p}(0)\ ,
\eeq
where, using the binomial formula for $(x^2-1)^\ell$ (see \cite{SicardyNature}), the associated
Legendre polynomials take the form
\beq{coeffP}
P_{2\ell,2p}(0)=(-1)^{\ell-p}\ {{(2\ell+2p)!}\over {2^{2\ell}(\ell+p)!(\ell-p)!}}\ ,\qquad \ell \ge p\ . 
\eeq 


\section{A Hamiltonian model in epicyclic formulation}\label{sec:epicyclic}

In this Section, we introduce the epicyclic formulation of the equations of motion for a ring's particle lying in the equatorial plane, specifying how the resonances are defined in such framework; we also introduce the \emph{resonant Hamiltonian}, which will allow us to study the particle's dynamics in the vicinity of the resonances. 

\begin{proposition}
Let the momenta $p_r$ and $p_L$, conjugated respectively to $r$ and $L$, be 
$$
p_r=\dot r\ ,\qquad p_L=r^2\dot L\ .
$$
Then, the Hamiltonian describing the dynamics of the particle is given by 
\begin{equation}\label{eq:Hinit}
\H(p_r,p_L,r,L-\Omega_P t)={{p_r^2}\over 2}+{{p_L^2}\over {2r^2}}+U(r,L-\Omega_P t)\ .
\end{equation}
\end{proposition}

The proof of $\mathcal H(p_r,p_L,r,L-\Omega_P t)$ in Eq.~\eqref{eq:Hinit} follows from classical arguments of Hamiltonian and Keplerian dynamics, recalling that $L$ is an angular coordinate and that the kinetic energy associated to it is given by $p_L/2r^2$.

Recalling that $\theta=L-\Omega_P t$ and defining the momentum $p_\theta\equiv p_L$, we obtain the Hamiltonian\footnote{In fact, we have 
	$$
	\dot\theta={{\partial\H}\over {\partial p_\theta}}={{p_\theta}\over {r^2}}-\Omega_P={{p_L}\over {r^2}}-\Omega_P=\dot L-\Omega_P\ . 
	$$}
\beq{H_pr_ptheta_r_theta}
\H(p_r,p_\theta,r,\theta)={{p_r^2}\over 2}+{{p_\theta^2}\over {2r^2}}-\Omega_P p_\theta+U(r,\theta)\ .
\eeq

The gravitational potential \equ{U3} can be written as the sum of an axisymmetric part $U_s(r)$ which is obtained taking the term with $p=0$ in Eq.~\equ{U3}, plus a non--axisymmetric part $U_{ns}(r,\theta)$. The axisymmetric part is given by 
\beq{U0}
U_s(r)=-{{\G M_P}\over r} \sum_{\ell=0}^\infty \Big({R\over
r} \Big)^{2\ell} C_{2\ell,0} P_{2\ell,0}(0) \ .
\eeq
We remark that $U_s(r)$ can be written as the sum of the Keplerian potential $U_{Kep}(r)=-\G M_P/r$ (corresponding to $\ell=0$) and a potential $\widetilde{U}_s=\widetilde{U}_s(r)$ depending only on the radius. 
From Eq.~\equ{U3}, the non--axisymmetric part $U_{ns}(r,\theta)$ is given by 
\beq{ns_potential}
U_{ns}(r,\theta)=-{{\G M_P}\over r}\
\sum_{\ell=1}^\infty \Big({R\over
r}\Big)^{2\ell}\ \sum_{p=1}^\ell \cos(2p\theta) C_{2\ell,2p} P_{2\ell,2p}(0)\ .
\eeq

We now introduce two key concepts to describe the motion of the particle around the rotating body $\mathcal E$: first of all, we  define the principal \emph{frequencies} of the particle's motion, which will be used to characterize the  \emph{resonances}, representing the main subject of this work.
\begin{definition}
Given the axisymmetric potential $U_s$ as in Eq.~\eqref{U0}, we define
the mean motion $n$ and the epicyclic frequency $\kappa$ as (see \cite{Chandrasekhar1942}))
\begin{equation}\label{eq:n2U0}
n(r) = \Big({1\over r}\ {{dU_s(r)}\over {dr}}\Big)^{1\over 2}\ ,
\end{equation}
and 
\begin{equation}\label{k2U0}
\kappa(r) = \Big(\frac{3}{r}{{dU_s(r)}\over {dr}} + {{d^2U_s(r)}\over {dr^2}}\Big)^{1\over 2}\ .
\end{equation} 
We define the synodic frequency of the particle as $n-\Omega_P$.
\end{definition}

The synodic frequency denotes the frequency of the particle's passages through a fixed point relative to the central body, while the epicyclic frequency describes the frequency of passages at pericenter (see also \cite{Goldreich1982,MurrayDermottbook}. Notice that the square of the quantity in Eq.~\equ{k2U0} can be written also as 
$$
\kappa^2(r) =  \frac{1}{r^3} \frac{d(r^4 n^2(r))}{dr}\ .       
$$

Expanding $U_s$ as in Eq.~\eqref{U0}, we can write the squares of the frequencies as 
\beqa{nkappa_2}
n^2(r) & = & {{\G M_P}\over r^3}\ \Big(1+ \sum_{\ell=1}^\infty (2 \ell +1) \Big({R\over r}\Big)^{2\ell} C_{2\ell,0} P_{2\ell,0}(0) \Big)\nonumber\\
& = & {{\G M_P}\over r^3}\ \Big(1-{3\over 2} C_{2,0} \Big({R\over r}\Big)^2+{{15}\over 8}C_{4,0} \Big ({R\over r} \Big )^4+...\Big)\nonumber\\
\kappa^2(r) & = &{{\G M_P}\over r^3}\ \Big(1-\sum_{\ell=1}^\infty (2 \ell +1)\Big ({R\over r}\Big )^{2\ell} C_{2\ell,0} P_{2\ell,0}(0)\Big)\nonumber\\
& = & {{\G M_P}\over r^3}\ \Big(1+{3\over 2}C_{2,0} \Big ({R\over r}\Big)^2-{{45}\over 8}C_{4,0} \Big ({R\over r}\Big)^4+...\Big) \, .
\eeqa
We notice that, to the lowest order (i.e., in the spherical case), $n$ and $\kappa$ reduce to the Keplerian mean motion, namely $n^2(r)=\kappa^2(r)=\G M_P/r^3$.

\begin{definition} \label{def:LinRes}
We introduce a Lindblad resonance for $j,m\in\integer\backslash\{0\}$, whenever the following condition is satisfied: 
\beq{linres}
j \kappa(r) = m \big( n(r) - \Omega_P \big) \ .
\eeq
For given $j$, $m$, $\Omega_P$, we define the Lindblad resonant radius as the quantity $r_{lin}(\Omega_P; m,j)$ which solves the equation \equ{linres}
(see also Section~\ref{sec:rradius} for the computation of the resonant radius).
\end{definition}

\begin{definition}\label{def:CorRes}
We define a corotation resonance whenever, in Eq.~\eqref{linres}, $j=0$ (that is, $n(r)=\Omega_P$). In such case, the correspondent solution is called corotation radius and is denoted with $r_{cor}(\Omega_P)$. 
\end{definition}
\begin{remark}
Using $\kappa=n-\dot\omega$, for a Lindblad resonance \equ{linres} we have 
\beqno
\frac{n -\dot\omega }{\Omega_P -\dot\omega} = \frac{m}{m-j} \, ,
\eeqno
meaning that in a frame rotating at the particle precession rate $\dot\omega$, the particle completes $m$ revolutions, while the central body completes $m-j$ rotations. For this reason, we will also refer to this kind of resonance as a $m:(m-j)$ Lindblad resonance. 
\end{remark}
Since our final aim is to study the local dynamics of the ring's particle in the resonant regime, it is convenient to introduce a new set of coordinates, centered at the resonant radius, which highlights the role of the synodic and epicyclic frequencies.

\begin{proposition}
    Let $r_*$ denote the value of the radius at the resonance, either $r_{lin}$ or $r_{cor}$, and define $\kappa_*= \kappa(r_*)$ and $n_* = n(r_*)$. Then, there exists a set of canonical action--angle coordinates $(J, I, \varphi, \theta)$, such that the particle's Hamiltonian, locally around $r_*$, can be expressed as 
\begin{equation}\label{eq:HIJ}
\begin{aligned}
\H(J,I,\varphi,\theta)  &= |\kappa_*| J + \big(n_*-\Omega_P \big)I 
\\ & + \sum_{|i|=0}^{\infty}\sum_{j=0}^{\infty} \big( \alpha_{2i,j}(J, I)\cos(2 i \theta +j \varphi)+ \beta_{2i,j}(J, I)\sin(2 i \theta +j \varphi)\big) \, ,
\end{aligned}
\end{equation}
where $\alpha_{2i,j}$ and $\beta_{2i, j}$ are suitable polynomials in $J, I$. 
\end{proposition}

\begin{proof}
In a neighborhood of the resonance, we define the distance $\rho$ from the resonant radius $r_*$ as 
$$
\rho=r-r_* \ .
$$
Next, we introduce the angular momentum $p_*$ at the resonant radius per unit mass in an orbit with axisymmetrical potential as 
\beq{pstar}
p_*=n_* \ r_*^2\ , 
\eeq
where $n_* = n(r_*)$. 
Let us introduce the momenta $p_{\rho}$ (conjugated to $\rho$) and $I$ (as the displacement from $p_*$) as 
$$
p_{\rho}= p_r \ ,\qquad 
I=p_\theta-p_*\ .
$$
Expressing the Hamiltonian~\equ{H_pr_ptheta_r_theta} in the translated variables $(\rho,p_\rho)$, $(\theta,I)$, one gets 
\beq{Hnew2}
	\mathcal H(p_\rho, I, \rho, \theta)=\frac{p_\rho^2}{2}+\frac{(I+p_*)^2}{2(r_*+\rho)^2}- (I+p_*)\Omega_P+U_s(\rho+r_*)+U_{ns}(\rho+r_*,\theta)\ . 
\eeq
Let us now take the Taylor expansion of $(I+p_*)^2/(2(\rho+r_*)^2)$ and of the axisymmetric potential $U_s(\rho+r_*)$ around $I=0$ and $\rho=0$; ignoring  constant terms and grouping together terms of the same order, the Hamiltonian~\equ{Hnew2} takes the form 
\begin{equation}\label{acca}
	\begin{aligned}
	\mathcal H(p_\rho, I, \rho, \theta)&=\frac{p_\rho^2}{2}+\left[\frac{p_*}{r_*^2}-\Omega_P\right]I+\frac{I^2}{2r_*^2}+\left[\frac{dU_s}{dr}(r_*)-\frac{p_*^2}{r_*^3}\right]\rho-2\frac{p_*}{r_*^3}I \rho-\frac{1}{r_*^3}I^2\rho\\
	&+\frac{1}{2}\left[\frac{d^2U_s}{dr^2}(r_*)+3\frac{p_*^2}{r_*^4}\right]\rho^2+\frac{3p_*}{r_*^4}I \rho^2+\frac{3}{2r_*^4}I^2\rho^2+U_{ns}(\rho+r_*, \theta) + F(I,\rho)\ , 
	\end{aligned}
\end{equation}
where $F(I,\rho)$ is a polynomial in $\rho$ and $I$ of order at least three in $\rho$ and with terms of order $0, 1$ and $2$ in $I$. \\
Using Eq.~\equ{pstar}, recalling Eq.~\equ{eq:n2U0}, and omitting constant terms, the Hamiltonian~\equ{acca} can be written as
\begin{equation}
	\begin{aligned}
		\mathcal H(p_\rho, I, \rho, \theta)=& \,\,\, \frac{p_\rho^2}{2}+\left(n_*-\Omega_P\right)I+\frac{I^2}{2r_*^2}-2\frac{n_*}{r_*}I \rho-\frac{1}{r_*^3}I^2\rho \\
		& +\frac{1}{2}\left(\frac{d^2U_s}{dr^2}(r_*)+\frac{3}{r_*}\frac{dU_s}{dr}(r_*)\right)\rho^2+\frac{3n_*}{r_*^2}I \rho^2+\frac{3}{2r_*^4}I^2\rho^2 \\
		& +U_{ns}(\rho+r_*, \theta) + F(I,\rho)\ .\nonumber
	\end{aligned}
\end{equation}
Keeping the explicit form of just the terms of second order and 
putting together all terms of total degree at least three, one arrives at the expression 
\beqa{HE}
\H(p_{\rho},I,\rho,\theta) & = & {{p_{\rho}^2}\over 2} + (n_*-\Omega_P) I+ \frac{1}{2} \Bigg(\frac{d^2 U_s}{d r^2}(r_*)+\frac{3}{r_*} \frac{d U_s}{dr}(r_*) \Bigg) \ \rho^2 \nonumber \\
&& -\frac{2 n_*}{r_*} I\, \rho   + \frac{I^2}{2 r_*^2} + F_0(I,\rho) + F_1(\rho, \theta)  \ ,
\eeqa
where $F_0$ is given by 
\beq{F0_expansion}
F_0(I,\rho) = \sum_{i=0}^2 \ \sum_{j=3-i}^{\infty} \ c_{i,j} \, I^i \rho^j\ ,
\eeq
which is a polynomial in $I$ and $\rho$ with total degree greater or equal than 3, where $I$ varies from zero to second degree and $c_{i,j}$ are suitable coefficients depending on $C_{i,j}$ (in Appendix~\ref{app:appA}, we report the $c_{i,j}$ coefficients for $i=0,1,2$ and for $j$ up to 4). The function $F_1$ is a polynomial in $\rho$ of any degree and trigonometric in $\cos(2\, p\, \theta)$; it is the contribution of the non--axysimmetric part given by $U_{ns}$, namely, $F_1(\rho,\theta)=U_{ns}(\rho+r_*, \theta)$, and it can be written as 
\beq{Vji}
F_1(\rho,\theta)= \sum_{j=0}^{\infty} \sum_{i=1}^{\infty} V_{j,i}\ \rho^j \cos(2 i \theta)
\eeq
with $V_{j,i}$ suitable coefficients (in Appendix~\ref{app:appA}, we report the $V_{j,i}$ coefficients for $j$ up to 4 and for $i$ up to 3).

Let $\kappa_*=\kappa(r_*)$; we note that the coefficient of $\rho^2$ in Eq.~\equ{HE} turns out to be (see Eq.~\equ{k2U0})
$$
\frac{d^2 U_s}{d r^2}(r_*)+\frac{3}{r_*} \frac{d U_s}{dr}(r_*)= \kappa_*^2 \ . 
$$
Thus, the Hamiltonian~\equ{HE} becomes
\beq{Hnew}
\H(p_{\rho},I,\rho,\theta) = {{p_{\rho}^2}\over 2}+ \frac{1}{2}  \kappa_*^2\ \rho^2  + (n_*-\Omega_P) I- \frac{2 n_*}{r_*} I\, \rho +  \frac{I^2}{2 r_*^2}+  F_0(I,\rho) + F_1(\rho, \theta)  \, .
\eeq
Let us introduce the epicyclic action--angle variables $(J,\varphi)$ related to $(\rho, p_{\rho})$ by the canonical transformations:  
\beq{epivar}
\left\{
	\begin{aligned}
\rho &= \sqrt{{{2 J}\over |\kappa_*|}}\ \sin\varphi\\
p_\rho &= \sqrt{2|\kappa_*|J}\ \cos\varphi\ . 
\end{aligned}	
	\right.
\eeq
Then, the Hamiltonian~\equ{Hnew} can be written in terms of the epicyclic variables as 
\beq{Hnewnew}
\H(J,I,\varphi,\theta)  = |\kappa_*| J + \big(n_*-\Omega_P \big)I 
+ \widetilde F_s(J,I, \varphi) +  \widetilde F_{ns}(J,\varphi, \theta) \ ,
\eeq
where $\widetilde F_s$ is due to the axisymmetrical part of the potential and can be written as
\begin{equation}\label{F_tilde}
	\begin{aligned}
\widetilde F_s(J,I, \varphi) &=  -2\frac{n_*}{r_*} I\sqrt{\frac{2 J}{|\kappa_*|}}\sin\varphi +\frac{I^2}{2r_*^2}+F_0\left(I,\sqrt{\frac{2 J}{|\kappa_*|}}\sin\varphi\right)\\
&= -2\frac{n_*}{r_*} I\sqrt{\frac{2 J}{|\kappa_*|}}\sin\varphi +\frac{I^2}{2r_*^2} + \sum_{i= 0}^2 \ \sum_{j=3-i}^{\infty} d_{i,j}  \, I^i \ J^{j/2} \ \sin^j \varphi \ ,
\end{aligned}
\end{equation}
where 
\beq{dij}
d_{i,j} = c_{i,j}\left( \frac{2}{|\kappa_*|} \right)^{j/2}  \, ,
\eeq
 while  $\widetilde F_{ns}$ is due to the non--axisymmetrical part of the potential and can be written as
\beqa{F1_tilde_product}
		\widetilde F_{ns}(J, \varphi, \theta) & = & U_{ns}\left(\sqrt{\frac{2 J}{|k_*|}}\sin\varphi+r_*, \theta\right) \nonumber \\
		 & = &\sum_{j=0}^{\infty} \sum_{i=1}^{\infty} \widetilde V_{j,i}\ J^{j/2} \cos(2 i \theta) \sin^j \varphi\, , 
\eeqa
where
$$
\widetilde V_{j,i}= V_{j,i} \left( \frac{2}{|\kappa_*|} \right)^{j/2} \ .
$$
Let us rewrite the sum of Eq.~\eqref{F_tilde} and of Eq.~\eqref{F1_tilde_product} as a linear combination of the angles $\theta$, $\varphi$:
\beq{F_tilde_plus_F1_tilde}
\widetilde F_s(J,I, \varphi)+\widetilde F_{ns}(J, \varphi, \theta) = \sum_{|i|=0}^{\infty}\sum_{j=0}^{\infty} \big( \alpha_{2i,j}\cos(2 i \theta +j \varphi)+ \beta_{2i,j}\sin(2 i \theta +j \varphi)\big) \, ,
\eeq
where $\alpha_{2i,j},\beta_{2i,j}$ are polynomial functions in $I$, $J$ depending on $\G$, $M_P$ and on the coefficients $P_{2\ell,2p},C_{2\ell,2p}$ (in Appendix~\ref{app:appA}, we report the $\alpha_{2i,j},\beta_{2i,j}$ coefficients for $|i|$ up to 3 and  $j$ from 0 to 4). Once again, we remark that the function~\equ{F_tilde_plus_F1_tilde} depends on multiples of $2\theta$, due to the symmetry of the central body $\E$, which is assumed to be an ellipsoid. Casting together \equ{Hnewnew}, \equ{F1_tilde_product}, \equ{F_tilde_plus_F1_tilde}, we obtain \equ{eq:HIJ}.
\end{proof}

As consequence of the definition of the epicyclic action--angle variables $(J,\varphi)$ given in Eqs.~\eqref{epivar}, it is possible to relate the action $J$ to the eccentricity of the orbit of the particle revolving about the central body. Thus, we have the following result. 

\begin{proposition}\label{prop:Jecc}
	Let $r_*$ denote the value of the radius at the resonance and define $\kappa_*= \kappa(r_*)$. Then, locally around $r_*$, the action variable $J$ can be expressed as 
	\beqno
	J = \frac{1}{2}\,|\kappa_*|\, r_*^2 \, e^2 \, ,
	\eeqno
	where $e$ is the eccentricity of the orbit of the particle $\P$  around $\E$.
	\end{proposition}	
	\begin{proof}
From Kepler's relations, taking as semimajor axis $a=r_*$, we have $r=r_*(1-e\, \cos u)$, where $u$ is the eccentric anomaly. Recalling that $\rho = r -r_*$ and using the first of Eqs.~\eqref{epivar}, we obtain  
$$
\sqrt{\frac{2 J}{|\kappa_*|}} \sin \varphi = -r_* \, e \, \cos u \, .
$$ 
From the last relation, we can set $u = \varphi + \pi/2$ and $\sqrt{\frac{2 J}{|\kappa_*|}} = r_* e$, which implies 
\beqano
J = \frac{1}{2}\,|\kappa_*|\, r_*^2 \, e^2 \,.
\eeqano
	\end{proof}	

\subsection{KAM non--degeneracy of the axisymmetric part}
 
In this Section, we check that the Hamiltonian~\equ{Hnewnew}, under different approximations, satisfies the so-called \sl Kolmogorov non--degeneracy condition, \rm {amounting to require that the Hessian matrix of the integrable part has determinant different from zero. The formal definition is given below.}

\begin{definition}\label{def:ND}
Given an $n$ DOF nearly--integrable Hamiltonian  $\H_0(\underline{Q},\underline{\eta})=h_0(\underline{Q})+ {f_0(\underline{Q},\underline{\eta})}$ with action--angle variables $(\underline{Q},\underline{\eta})\in\real^n\times\torus^n$, we say that $\H_0$ satisfies Kolmogorov non--degeneracy condition, if the Hessian matrix of the integrable part has determinant different from zero: 
$$
\det(\partial^2_{\underline Q} h_0(\underline{Q}))\not=0\ . 
$$
\end{definition}

Kolmogorov non--degeneracy condition is essential for the application of Kolmogorov-Arnold-Moser (KAM) theory, which ensures the persistence of invariant tori under small perturbations of the integrable part. The phase space associated to the Hamiltonian~\equ{Hnewnew} (shown in Section~\ref{sec:test}) is 4D; hence, on a (3D) fixed energy level, the 2D KAM rotational invariant tori provide stability for infinite times (see, e.g., \cite{Celletti1990II,Celletti1990I,CC2007} for stability results through trapping KAM tori in model problems of Celestial Mechanics). 

\vskip.1in 

{
In Proposition \ref{prop:nondeg} below, we prove that the Hamiltonian \equ{Hnewnew}, upon suitable truncations and averaging, is non--degenerate in the sense of Kolmogorov.
\begin{definition}\label{def:expansion}
	With reference to the Hamiltonian \equ{Hnewnew}, we define the axisymmetric Hamiltonian as 	
$$
	\H^{axi}(J,I,\varphi)  = |\kappa_*| J + \big(n_*-\Omega_P \big)I 
	+ \widetilde F_s(J,I, \varphi)\ ,
$$
which can be expanded in powers of $I$ and $\sqrt{J}$ as  
\beq{Haxiexp}
\H^{axi}(J,I,\varphi) =  |\kappa_*| J + \big(n_*-\Omega_P \big)I  -2\frac{n_*}{r_*} I\sqrt{\frac{2 J}{|\kappa_*|}}\sin\varphi +\frac{I^2}{2r_*^2} + \sum_{i= 0}^2 \ \sum_{j=3-i}^{\infty} d_{i,j}  \, I^i \ J^{j/2} \ \sin^j \varphi
\eeq
with $d_{i,j}$ as in \equ{dij}.
\end{definition}
}

{
For our purposes, we will identify the Hamiltonian \equ{Hnewnew} with $\mathcal H_0(\underline Q, \underline \eta)$ in Definition \ref{def:ND} by setting $\underline Q=(J, I)$, $\underline \eta=(\varphi, \theta)$ and 
\beqa{Haxi}
h_0(J, I)&=&|\kappa_*| J + \big(n_*-\Omega_P \big)I +\langle \widetilde F_s(J,I, \varphi) \rangle\nonumber\\
f_0(J, I, \varphi, \theta)&=& \widetilde F_s(J,I, \varphi)-\langle \widetilde F_s(J,I, \varphi) \rangle+\widetilde F_{ns}(J, \varphi, \theta)\ , 
\eeqa
where $\langle\cdot\rangle$ denotes the average with respect to $\varphi$. We notice that $h_0$ coincides with the average of $\H^{axi}$ in Definition \ref{def:expansion}: 
\beq{h0Haxi}
h_0(J, I)=\langle \H^{axi}(J,I,\varphi) \rangle\ .
\eeq
Then, we set $\mathcal H_0$ as 
\beq{H0new}
\mathcal H_0(J,I,\varphi,\theta)=h_0(J, I)+f_0(J, I, \varphi, \theta)\ . 
\eeq
}

{
	\begin{proposition}\label{prop:nondeg}
Consider the Hamiltonian $\mathcal H_0$ as in \equ{H0new} with $h_0$ defined in \equ{h0Haxi} expanded in the action variable $J$ up to first or second order (see \equ{Haxiexp}). 
Then, the Kolmogorov non--degeneracy condition of Definition \ref{def:ND} is satisfied in the following cases. 
		\begin{itemize}
\item[(i)] Consider the Hamiltonian $\H^{axi}$ expanded up to first order in $J$: 
			\beqa{H_rho2}
			\H^{axi}(J,I,\varphi) & = &|\kappa_*| J + \big(n_*-\Omega_P \big)I  -2\frac{n_*}{r_*}\sqrt{\frac{2 }{|\kappa_*|}} I J^{1/2}\sin\varphi \nonumber \\ && +\frac{I^2}{2r_*^2} +\frac{6 \, n_*}{|\kappa_*| r_*^2 }\, I\, J \sin^2 \varphi + O_{5/2} \, ,
			\eeqa
			where $O_{5/2}$ are terms of order greater or equal than $5/2$ in the variable $IJ$. Denoting by $h_0(J, I)$ its average with respect to the angle as in \equ{h0Haxi}, then the Hamiltonian $\mathcal H_0$ is non--degenerate provided $n_*$ is different from zero.
\item[(ii)] Consider the Hamiltonian $\H^{axi}$ expanded up to second order in $J$: 
			\beqa{H_rho4}
			\H^{axi}(J,I,\varphi) & = &|\kappa_*| J + \big(n_*-\Omega_P \big)I  -2\frac{n_*}{r_*}\sqrt{\frac{2 }{|\kappa_*|}} I J^{1/2}\sin\varphi+\frac{I^2}{2r_*^2}  \nonumber \\ && +\frac{6 \, n_*}{|\kappa_*| r_*^2 }\, I\, J \sin^2 \varphi + d_{0,3} J^{3/2} \sin^3 \varphi + d_{0,4} J^2 \sin^4 \varphi + O_{5/2} \, ,
			\eeqa
			where $O_{5/2}$ are terms of order greater or equal than $5/2$ in the variable $IJ$, and $d_{0,3}$ (truncated to $O({R^7\over {r^{11}}})$), $d_{0,4}$ (truncated to $O({R^7\over {r^{12}}})$) are given by
			\beqa{d03}
			d_{0,3} & = & -\frac{2 \sqrt{2}\,  \G M_P \,P_{0,0} \,C_{0,0}  }{\kappa_*^{3/2} \, r_*^4} + \frac{8\sqrt{2} \,  \G M_P R^2 \, P_{2,0}\,C_{2,0}  }{\kappa_*^{3/2} \, r_*^6} \nonumber \\ 
			&& +\frac{50 \sqrt{2}\,  \G M_P R^4 \,P_{4,0}\,C_{4,0}  }{\kappa_*^{3/2} \, r_*^8}  +\frac{140\sqrt{2} \,  \G M_P R^6 \,P_{6,0}\,C_{6,0}  }{\kappa_*^{3/2} \, r_*^{10}} \  ,
			\eeqa
			\beqa{d04}
			d_{0,4} & = & \frac{6\,  \G M_P \,P_{0,0} \,C_{0,0}  }{\kappa_*^2 \, r_*^5} - \frac{30\,  \G M_P R^2 \, P_{2,0}\,C_{2,0}  }{\kappa_*^2 \, r_*^7} \nonumber \\ 
			&& -\frac{230\,  \G M_P R^4 \,P_{4,0}\,C_{4,0}  }{\kappa_*^2 \, r_*^9}  -\frac{770 \,  \G M_P R^6 \,P_{6,0}\,C_{6,0}  }{\kappa_*^2 \, r_*^{11}} \  ,
			\eeqa
			being $C_{i,j}$, $P_{i,j}$ as in Eqs.~\equ{coeffC}, \equ{coeffP}. Denoting by $h_0(J, I)$ its average with respect to the angle as in \equ{h0Haxi}, then the Hamiltonian $\mathcal H_0$ is non--degenerate provided 
			\beq{condii}
			\frac{3 d_{0,4}}{4 r_*^2}-\left(\frac{3 \, n_*}{|\kappa_*| r_*^2 }\right)^2 \neq 0\ .
			\eeq
		\end{itemize}
	\end{proposition}
}

{
	\begin{proof}
		(i) Denoting by  
		\beq{frequencies}
		\omega_1 =|\kappa_*| \, , \quad \omega_2 =  n_*-\Omega_P \, ,
		\eeq
		averaging $\mathcal H^{axi}$ with respect to $\varphi$ and omitting the higher order terms, the average Hamiltonian becomes
		\beq{H_rho2_ave}
		h_0(J,I)  = \omega_1 J + \omega_2 I +\frac{I^2}{2r_*^2} +A\, I J \, ,
		\eeq
		where
		\beq{defA}
		A = \frac{3 \, n_*}{|\kappa_*| r_*^2 } \, .
		\eeq
		Let us compute the determinant of the Hessian of the Hamiltonian~\eqref{H_rho2_ave}:
		\beqno
		\det \Big( \partial^2_{JI} h_0(J,I) \Big) = \det
		\begin{pmatrix}
			0 & A \\
			A & \frac{1}{r_*^2} 
		\end{pmatrix}
		= -A^2 \ ;
		\eeqno
		thus, the Hamiltonian~\eqref{H_rho2_ave} is not degenerate, provided that $A\not=0$, which amounts to require $n_*\not=                                                                                                   0$. \\	
		(ii) Averaging Eq.~\equ{H_rho4} with respect to $\varphi$, recalling Eq.~\eqref{frequencies}, omitting the zero--average terms in $\varphi$ and the terms of higher orders, one gets that the average Hamiltonian $h_0$ is given by 
		\beq{H_rho4_ave}
		h_0(J,I) =  \omega_1 J + \omega_2 I  + \frac{I^2}{2r_*^2} +A\, I J + \frac{3}{8} d_{0,4} \, J^2\ ,
		\eeq
		where $A$ is as in Eq.~\equ{defA}. \\
		The determinant of the Hessian of the Hamiltonian~\eqref{H_rho4_ave} is
		\beq{det}
		\displaystyle
		\det \Big( \partial^2_{JI} h_0(J,I) \Big) = \det
		\begin{pmatrix}
			\frac{3}{4}d_{0,4} & A \\
			A & \frac{1}{r_*^2} 
		\end{pmatrix}
		= \frac{3 d_{0,4}}{4 r_*^2}-A^2   \, ,
		\eeq
		which is different from zero under \equ{condii}, ensuring that the Hamiltonian~\eqref{H_rho4_ave} is not degenerate. \\	
\end{proof}}

\subsection{Resonances in the epicyclic variables}\label{Sec:resepi}
Let us now proceed to the definition of the resonances using the epicyclic formulation: from Eq.~\equ{eq:HIJ} we have that 
$$
\dot\varphi={{\partial\H}\over {\partial J}}\ ,\qquad 
\dot\theta={{\partial\H}\over {\partial I}}
$$
and at lowest order it is 
$$
\dot\varphi=|\kappa_*|\ ,\qquad \dot\theta=n_*-\Omega_P\ ;
$$
Moreover, we observe that the angles $(\theta,\varphi)$ appear as combinations of the form $2i\theta+j\varphi$ with $i,j\in\integer$. Following Definitions~\ref{def:LinRes} and \ref{def:CorRes}, we have that  
\begin{enumerate}
	\item[(i)] a corotation or $1:1$ resonance corresponds to $j=0$ and any $i$;  
	\item[(ii)] a Lindblad resonance of order $m:(m-j)$ occurs when $j\kappa_*=m(n_*-\Omega_P)$ and it corresponds to $2i=-m$ and any $j$.
\end{enumerate}

The resonant combination of the angles associated to a Lindblad resonance as in (ii) is $-m\theta+j\varphi$ with $m=-2i$. In particular, 

\vskip.1in 

\noindent 
(A) $m=-1$, $j=1$ gives $m:(m-j)=1:2$; the resonant angle is $\theta+\varphi$ and its multiples;

\noindent 
(B) $m=-1$, $j=2$ gives $m:(m-j)=1:3$; the resonant angle is $\theta+2\varphi$ and its multiples. 


\vskip.1in

\begin{remark}
	We note that, at the corotation resonance, the expression $n_*=\Omega_P  $ holds; the term of first order in $I$ in Hamiltonian~\equ{eq:HIJ} vanishes; this is consistent with the fact that $\dot \theta = \partial \H /\partial I = 0$, namely $\theta = L - \Omega_P t$ is constant, meaning that the particle $\P$ is revolving with the same angular velocity of the body $\E$. 
\end{remark}

\section{Analysis of three resonances in two test cases}\label{sec:test}

To analyze the differences among the resonances, we consider two test cases:  the first one representing an almost spherical (AS) body, the second one representing the opposite case of a highly aspherical (HA) body. We consider three resonances: corotation (also denoted as $1:1$), $1:2$ and $1:3$ resonances, whose choice is motivated by the fact that they represent, respectively, resonances at close, intermediate and far distances from the primary body (see Section~\ref{sec:rradius}, compare with \cite{CDD25} for applications to real test cases). In Table~\ref{table:test_values}, we report the astronomical values of the parameters of the two test cases, chosen to be comparable with the values of small bodies in the Solar system (like dwarf planets or Centaurs). Beside the physical parameters (mass, semi--axes, reference radius -- computed from Eq.~\equ{refradius}), the elongation and oblateness parameters (computed from Eqs.~\equ{oblat}, \equ{elong}), we provide the rotation period of the central body and the resonant radii (see Section~\ref{sec:rradius}) for corotation, $1:2$ and $1:3$ resonances. We notice that with the values in Table~\ref{table:test_values}, the coefficients $C_{2\ell,2 p}$ defining the potential can be computed through the expressions given in Eq.~\equ{coeffC}. 

We also notice that we have checked for AS and HA the KAM non--degeneracy condition as in Proposition~\ref{prop:nondeg} and we confirm that the Hamiltonian is non--degenerate in all resonances $1:1$, $1:2$, $1:3$ at order 1 and 2 in $J$.
	\begin{center}
\begin{table}[h]

	\begin{tabular}{|c|c|c|}
	\hline
 & AS & HA \\	\hline
		Rotation period $T_{rot}$ (h) & 8 & 8 \\
	\hline
	 Mass $M_P$ (kg) & $10^{21}$ & $10^{21}$\\
	 \hline
	 Semi--axes $a \times b \times c$ (km) &  $1000 \times 980 \times 960$ & $1000 \times 650 \times 400$\\
	 \hline
	 Reference radius $R$ (km)& 979 & 558 \\
	 \hline
	 Elongation parameter & 0.0305705 & 0.927371  \\
	 \hline
	 Oblateness parameter & 0.0206585 & 0.885218  \\
	 \hline
	 Corotation radius (km) & 1124.5 & 1176.82\\
	 \hline
	 $1:2$ resonance radius (km) & 1776.57 & 1771.35 \\
	 \hline
	 $1:3$ resonance radius (km) & 2327.16 & 2316.93 \\
	 \hline
	\end{tabular}

\vskip.1in 

\caption{Parameter values of the two selected test cases of an almost spherical body (AS) and a highly aspherical object (HA); the table provides also the location of the three selected resonances.}\label{table:test_values} 

\end{table}
	\end{center}

\subsection{The resonant radius}\label{sec:rradius}
As a first task, we compute the value of the radius corresponding to a  resonance of order $p:q$. We will use two different methods, described in the following Definition.  

\begin{definition} 
  The Keplerian resonant radius $r_{pq}^{kep}$ is defined as 
    \beq{r_kep}
r_{pq}^{kep} = (\G M_P T_{pq}^2)^{1/3} \, ,
\eeq
where $T_{pq}$ is the period of revolution of the particle $\P$ about $\E$, namely
\beqno
T_{pq} = \frac{q}{p} \, \frac{1}{\Omega_P} \, .
\eeqno
Moreover, from Eq.~\eqref{linres}, the resonant radius $r_{pq}$ is defined as a solution of the equation 
\beq{res_eq}
(p-q) \kappa(r_{pq}) =  p (n(r_{pq}) - \Omega_P) \ ,
\eeq
where $n(r)$, $\kappa(r)$ can be computed through Eqs.~\eqref{nkappa_2}. 
\end{definition}

While the definition of $r_{pq}^{kep}$ takes advantage of Kepler's third law, considering then a spherical approximation of the central body $\mathcal E$, the method used to find $r_{pq}$, although heavier from a computational point of view,  takes into account higher order harmonics of the gravitational potential of the ellipsoid $\E$. Note that $r_{pq}$ is the analogous of the Lindblad radius given in Definition~\ref{def:LinRes} or the corotation radius given in Definition~\ref{def:CorRes}. 
In Figure~\ref{fig:rel_diff_radii}, left panels, we plot  the relative difference $d$ between the radii computed by Kepler's law~\equ{r_kep} and by solving Eq.~\equ{res_eq}, namely, the quantity 
\begin{equation*}
	d= \frac{|r_{pq}^{kep}-r_{pq}|}{r_{pq}^{kep}}\ , 
\end{equation*}
considering several resonances ranging from $1:3$ to corotation. 

In Figure~\ref{fig:rel_diff_radii}, right panels, we plot the values of the resonant radii obtained from Eq.~\equ{res_eq} for several resonances from $1:3$ to corotation, compared with the reference radius $R$ and the biggest semi--axes $a$ in the two test cases (AS and HA). 

We notice that the relative difference of the radii is much smaller in the AS case than in the HA case, as expected since the latter differs more from a spherical body. Moreover, the highest difference is found for the corotation resonance. Interestingly, in both cases the corotation radius is very close to the body (namely, to the semi--axis $a$), while the $1:3$ resonant radius is the farthest. 

\begin{figure}[h]
	\centering
	\includegraphics[scale=0.3]{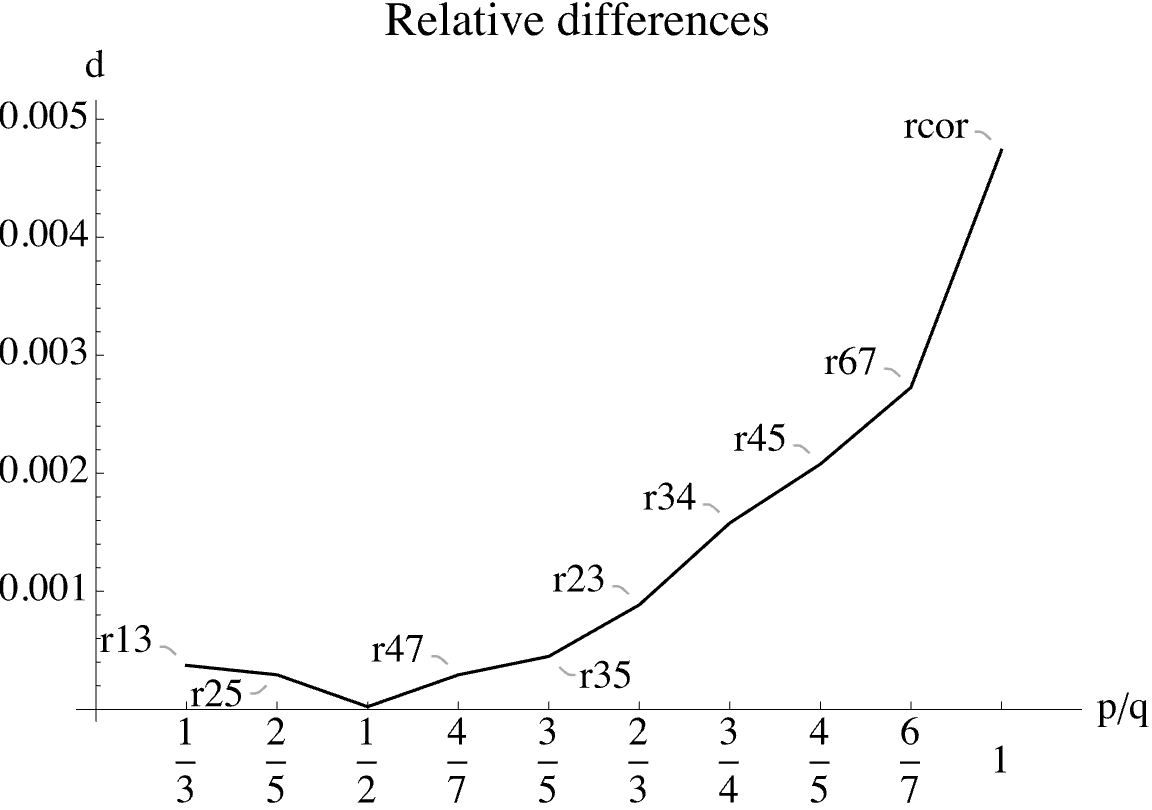}\hglue0.5cm
	\includegraphics[scale=0.3]{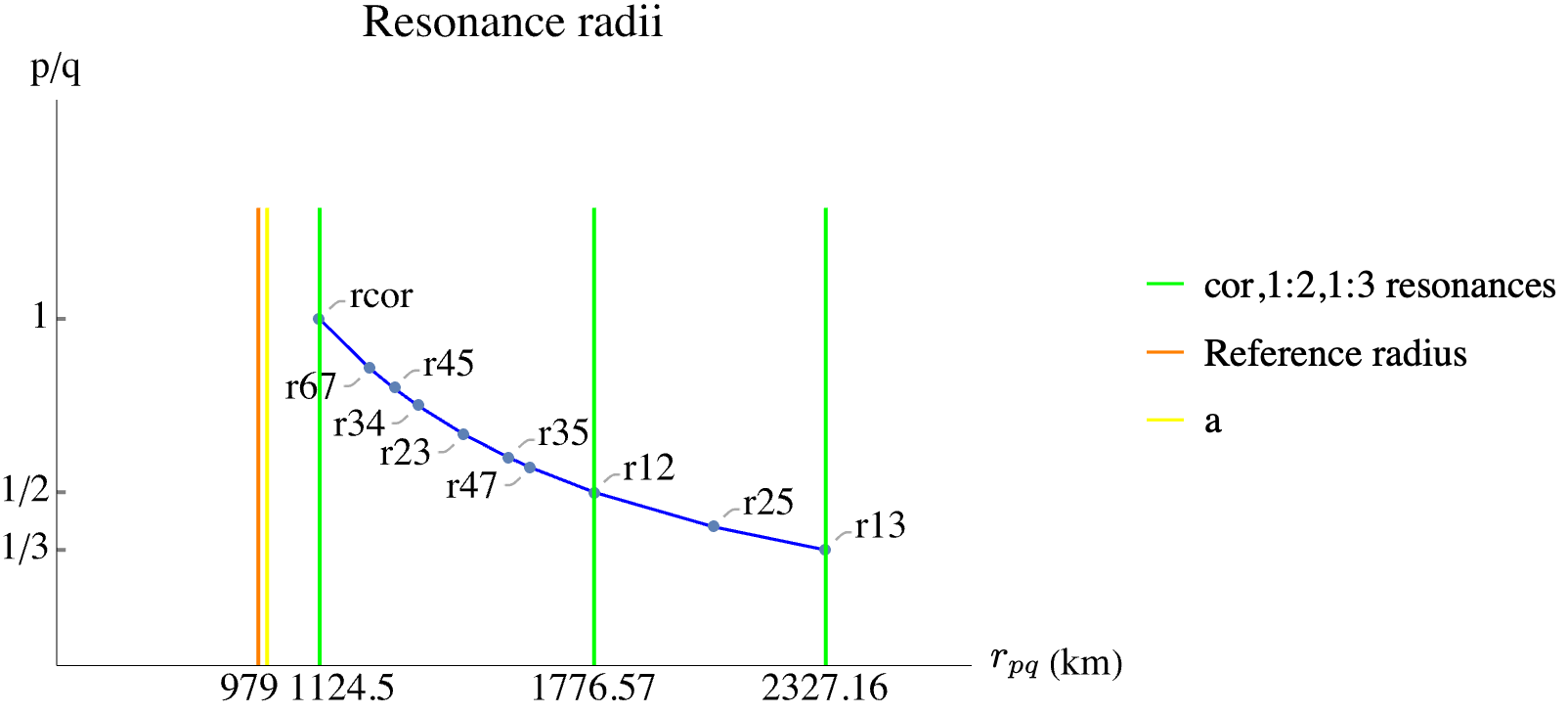}\\
	\vspace{1cm}
	\includegraphics[scale=0.3]{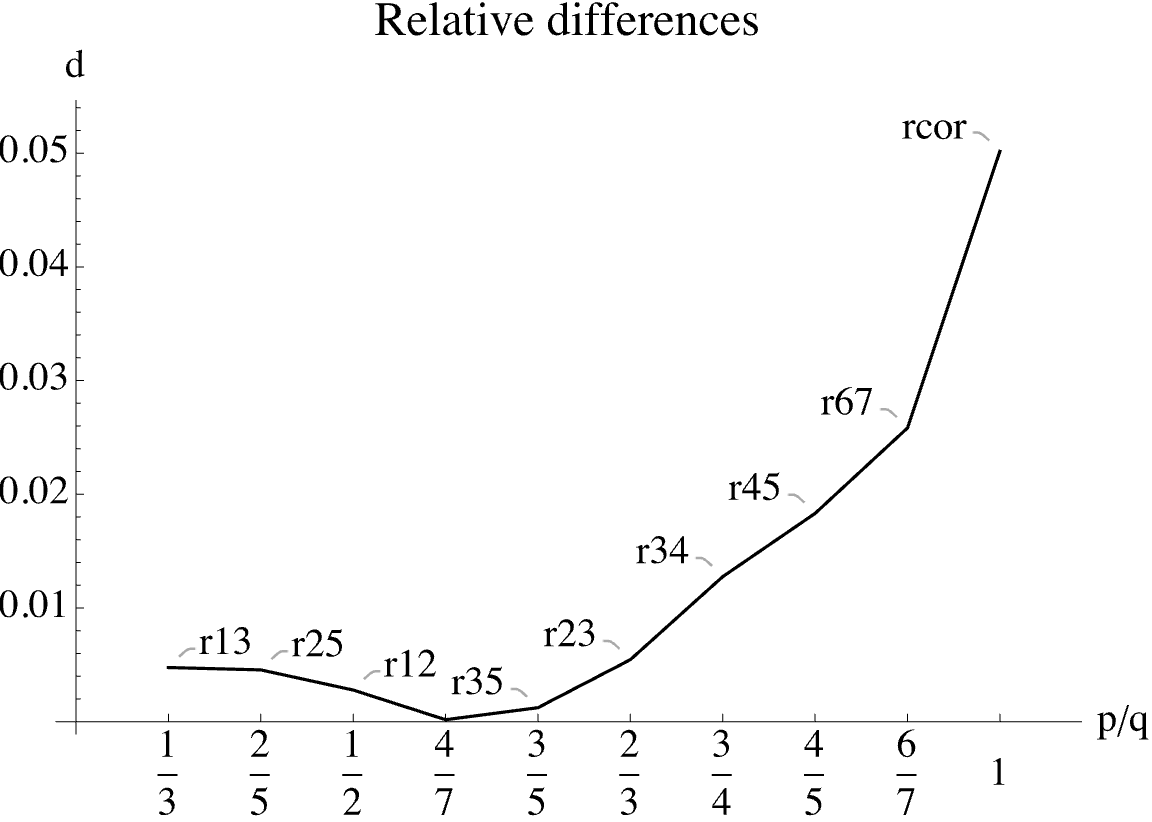}\hglue0.5cm
	\includegraphics[scale=0.3]{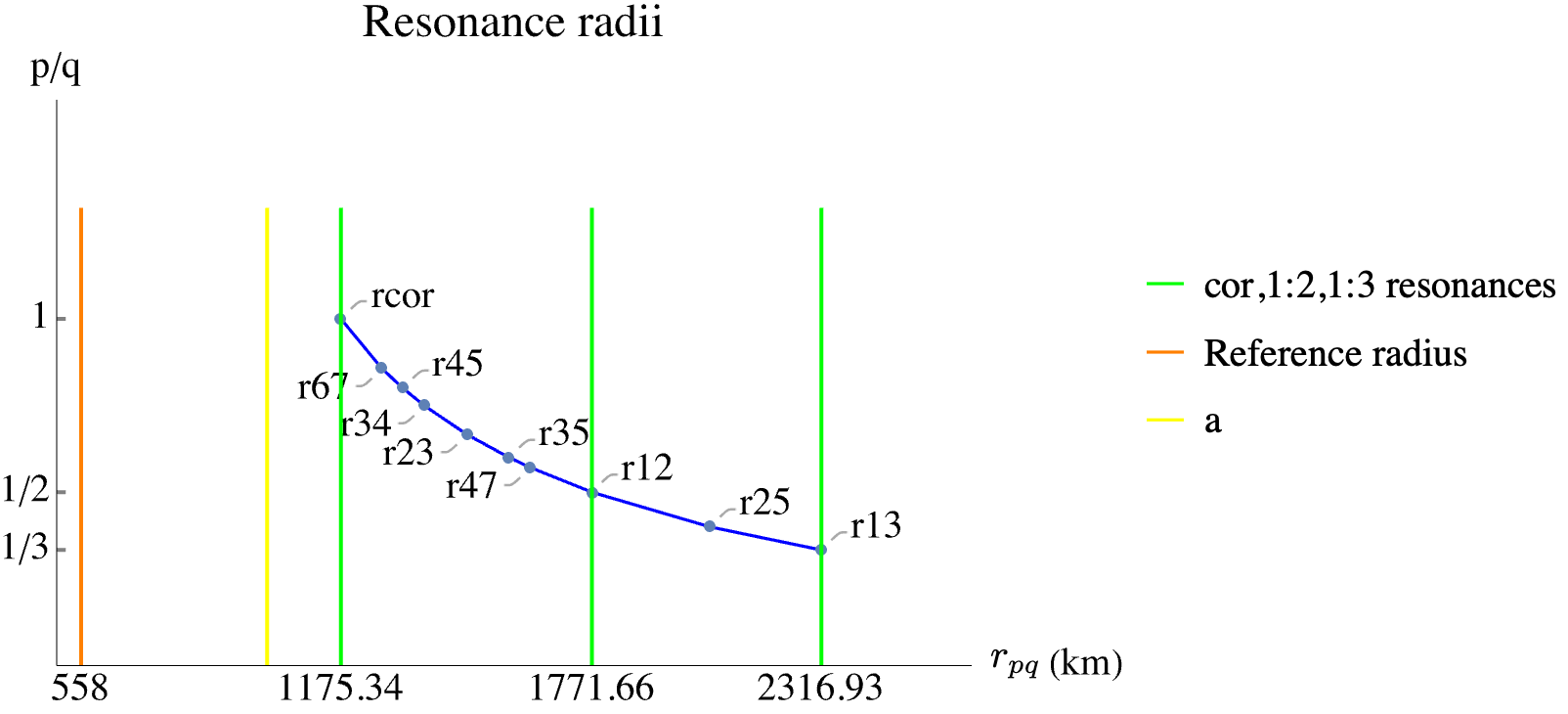}
	\caption{Top: almost spherical case. Bottom: highly aspherical case. Left: relative difference $d$ between the radii computed through the third Kepler's law~\equ{r_kep} and considering the complete resonance equation~\equ{res_eq} as the order of resonance $p/q$ varies. Right: lengths of the resonance radii $r_{pq}$ (in km) using Eq.~\equ{res_eq}, from corotation to $1:3$, compared with the reference radius $R$ (in orange) and the biggest semi--axis $a$ (in yellow).}
	\label{fig:rel_diff_radii}
\end{figure}

\subsection{Amplitudes of libration and occurrence of bifurcations}

As an indication of the strength of the resonances, we compute their amplitude of libration. To this end, we proceed to reduce the Hamiltonian to the lowest terms of the expansion around the resonance; this leads to a one--dimensional pendulum Hamiltonian for which we can compute, using standard  formulae (see, e.g., \cite{lichtenberg2013regular}), the amplitude around the stable equilibrium. We will limit the implementation of this computation to the corotation case. 

A more accurate estimate is obtained following the dynamics induced by the whole Hamiltonian through a numerical integration and in computing the distance between the separatrices around the elliptic equilibria. 

The results described in the following are obtained for the Hamiltonian~\eqref{eq:HIJ} truncated up to $i = 5$ and $j = 8$; this means truncating the axisymmetric potential $U_s(r)$ given in Eq.~\eqref{U0} and the non--axisymmetric part $U_{ns}(r,\theta)$ given by Eq.~\eqref{ns_potential} at order $\ell = 5$ and expanding the function $F_0(I,\rho)$, given in Eq.~\eqref{F0_expansion}, and the function $F_1(\rho, \theta)$, given in Eq.~\eqref{Vji}, in $\rho$ up to order 16. 

This choice has been done taking into account that the expansion in $\rho$ up to order 16 corresponds to an expansion of the Hamiltonian~\eqref{eq:HIJ} in $J$ up to the order 8 and in the eccentricity of the particle's orbit up to the order 16 as proved in Proposition~\ref{prop:Jecc}; in this case, the remainder of the complete Hamiltonian with respect to the truncated Hamiltonian turns out to be of the order of $10^{-6}$ for values of the eccentricities up to 0.5. Moreover, the truncation of the gravitational potential at order $\ell=5$ provides an error of the highest terms still of the order of $10^{-6}$. Notice that such errors vary according to the chosen resonance.

The first step done to study the local behavior of a particle in the vicinity of a resonance is to consider a suitable truncation of the Hamiltonian $\mathcal H(J, I, \varphi, \theta)$ as expressed in Eq.~\eqref{eq:HIJ}, which takes into account only the terms affected by the resonance itself. In general, this is done by taking $r_*$ as the resonant radius, and keeping only the trigonometric terms depending on the resonant angular combination. 
\begin{proposition}
Let us consider a resonance of order $m:(m-j)$, and consider the corresponding resonant radius $r_{m,m-j}$, which is, according to Definitions~\ref{def:LinRes} and \ref{def:CorRes}, solution of the equation~\eqref{linres}.
Let $r_*=r_{m, m-j}$; then, the resonant Hamiltonian associated to the resonance $m:(m-j)$ is given by 
\begin{equation}\label{eq:Ham_res}
\mathcal H_{res}^{m,m-j}(J, I, \varphi, \theta)=Z(I,J)+\sum_{i=1}^N f_i(I, J) cs\left(i(j \varphi-m\theta)\right),  
\end{equation}
where $Z(I, J)$, called normal part, collects all the non--trigonometric terms of $\mathcal H(J, I, \varphi, \theta)$, $f_i$ are suitable polynomials in the actions $I, J$, $cs$ stands for cosine or sine, depending on the parity of $i$ and $j$, $N$ depends on the maximal order of the expansion in $\theta$ and all the terms are evaluated at $r_*$.
\end{proposition}
\begin{proof}
    The expression of $\mathcal H_{res}^{m,m-j}(J, I, \varphi, \theta)$ comes from Eq~\eqref{eq:HIJ}. In the resonant Hamiltonian, only the terms corresponding to the resonant angles associated to the $m:(m-j)$ resonance, that is, multiples of $j \varphi-m\theta$, are taken into account. The polynomials $f_i(J, I)$ are the analogous of $\alpha_{2i,j}, \beta_{2i,j}$ in Eq.~\eqref{F_tilde_plus_F1_tilde}, taking into account the change of notation in the indices. The relation between the parity of the coefficient of $\varphi$ and the presence of either a sine or a cosine in $\mathcal H_{res}^{m,m-j}$ can be deduced from the expressions of  $\alpha_{2i,j}, \beta_{2i,j}$ in Appendix~\ref{app:appA}: whenever such coefficient is even, we have a cosine in the sum; on the other hand, when it is odd a sine appears in Eq.~\eqref{eq:Ham_res}. 
\end{proof}
As we will see in the specific cases of the $1:2$ and $1:3$ resonances, with a suitable change of coordinates it is possible to reduce $\mathcal H_{res}^{m,m-j}$ to a 1 DOF Hamiltonian, depending only on the resonant angle and the conjugated action.
	
\subsubsection{Corotation}\label{subsec:cor}

For the corotation case, we write the Hamiltonian ~\eqref{eq:Ham_res} as the sum of the normal part (depending only on the actions) computed in $r_*= r_{cor}$ and the part depending on the resonant angle $\theta$. Evaluating the coefficients at $r_{cor}$ obtained by solving Eq.~\eqref{res_eq}, we get the Hamiltonian 
\beqa{ham_cor_num_J0_AS}
\H_{res}^{cor}(J,I,\varphi,\theta) & = & f_1  (J) \, I + f_2(J)\, I^2  + f_3(J) \,\cos 2 \, \theta + f_4(J)\, \cos 4 \, \theta \nonumber \\
&& +  f_5(J) \, \cos 6 \, \theta+  f_6(J) \, \cos 8 \, \theta+  f_7(J) \, \cos 10 \, \theta\ ,
\eeqa
where $f_i(J)$, $i=1,\ldots,7$, are polynomials in the $J$ variable of order 8. Since the variable $\varphi$ is cyclic, then $J$ is constant, say $J_0$. Thus, we get a Hamiltonian of the type 
$$ 
\H_{res}^{cor}(J,I,\varphi,\theta) =\H_{res}^{cor}(I,\theta; J_0) 
$$
with $f_i(J) = f_i(J_0)$ constant terms for $i=1,\ldots,7$.

After fixing $J_0$, we can integrate the equations of motion associated to the 1 DOF Hamiltonian~\equ{ham_cor_num_J0_AS}. To fix $J_0$, recall that from Eqs.~\equ{epivar} at corotation one has  
\beq{J_0}
J_0 = \frac{p_{\rho}^2}{2 |\kappa(r_{cor})|}+\frac{|\kappa(r_{cor})|\rho^2}{2}\ ,
\eeq
which depends on $\rho$ and $p_{\rho}$. From standard Kepler's relations (see for example \cite{Celletti2010}), the value of $J_0$ is determined by taking $\rho =e\, r_{cor}$ and $p_\rho=r_{cor}\,e\,n_*$ (where we used $a=r_{cor}$, $n=n_*$). 

Fixing, for example, $e=10^{-3}$, from Eq.~\eqref{J_0} we obtain $J_0$ which, in the two sample cases AS and HA, takes the values  
$$
J_{0,AS} = 2.7590\cdot 10^{-4}\ , \qquad J_{0,HA} = 3.1135\cdot 10^{-4}\, .
$$
At $J_0$ the Hamiltonian~\equ{ham_cor_num_J0_AS} has the form  
\beq{ham_cor_Itheta}
\H_{res}^{cor}(I,\theta)=\alpha_1 \, I + \alpha_2 \, I^2 + \alpha_3 \cos 2\, \theta +\alpha_4  \cos 4\, \theta + \alpha_5 \cos 6\, \theta + \alpha_6 \cos 8\, \theta + \alpha_7 \cos 10\, \theta 
\eeq
with the coefficients $\alpha_j$, $j=1,...,7$, as in Table~\ref{tab:coeffs_cor} for the test cases AS and HA when $e=10^{-3}$.
\begin{table}[h!]
	\begin{tabular}{|c|c|c|}
		\hline
Coefficients & $1:1$ AS & $1:1$ HA  \\
		\hline
		$\alpha_1$ & $6.6\cdot 10^{-10}$ & $8.6 \cdot 10^{-10}$ \\
		\hline
		$\alpha_2$ & $3.9\cdot10^{-7}$ & $3.6\cdot10^{-7}$  \\
		\hline
		$\alpha_3$ & $-5.6\cdot10^{-4}$ & $-8.4\cdot10^{-3}$ \\
		\hline
		$\alpha_4$ & $-5.5\cdot10^{-6}$ & $-1.1\cdot10^{-3}$ \\
		\hline
		$\alpha_5$ & $-8.0\cdot10^{-8}$ & $-2.3\cdot10^{-4}$ \\
		\hline
		$\alpha_6$ & $-1.3\cdot10^{-9}$ &  $-5.1\cdot10^{-5}$\\
		\hline
		$\alpha_7$ & $-2.6\cdot10^{-11}$ & $-1.0\cdot10^{-5}$ \\
		\hline
	\end{tabular}
	
	\vskip.1in 
	
\caption{Coefficients appearing in the corotation resonant Hamiltonian~\eqref{ham_cor_Itheta} for the sample cases AS and HA at $e=10^{-3}$.}\label{tab:coeffs_cor} 	
\end{table}

If we retain only the term in $\cos 2\theta$, we obtain the Hamiltonian 
\beqno
\H(I,\theta) = \alpha_1 I + \alpha_2 I^2 + \alpha_3 \, \cos 2\,\theta \, ;
\eeqno
it is straightforward to see that the semi--amplitude of the resonant island can be computed as (see  \cite{lichtenberg2013regular}) 
\beq{amplitude}
\Delta I = \sqrt{\frac{-2\alpha_3}{\alpha_2}}\ . 
\eeq
Using the values in Table~\ref{tab:coeffs_cor} for the two test cases, we obtain (the units of measure are $\text{km}^2/\text{s}$): 
\beq{amplI}
\Delta I_{AS} = 53.5516\ ,\qquad 
\Delta I_{HA} = 216.419\ .
\eeq
In Figure~\ref{fig:cor_amplitude}, we plot the variation of the semi--amplitude $\Delta I$ as the eccentricity varies according to Eq.~\eqref{amplitude}. 

\begin{figure}[h!]
	\centering
	\includegraphics[scale=0.38]{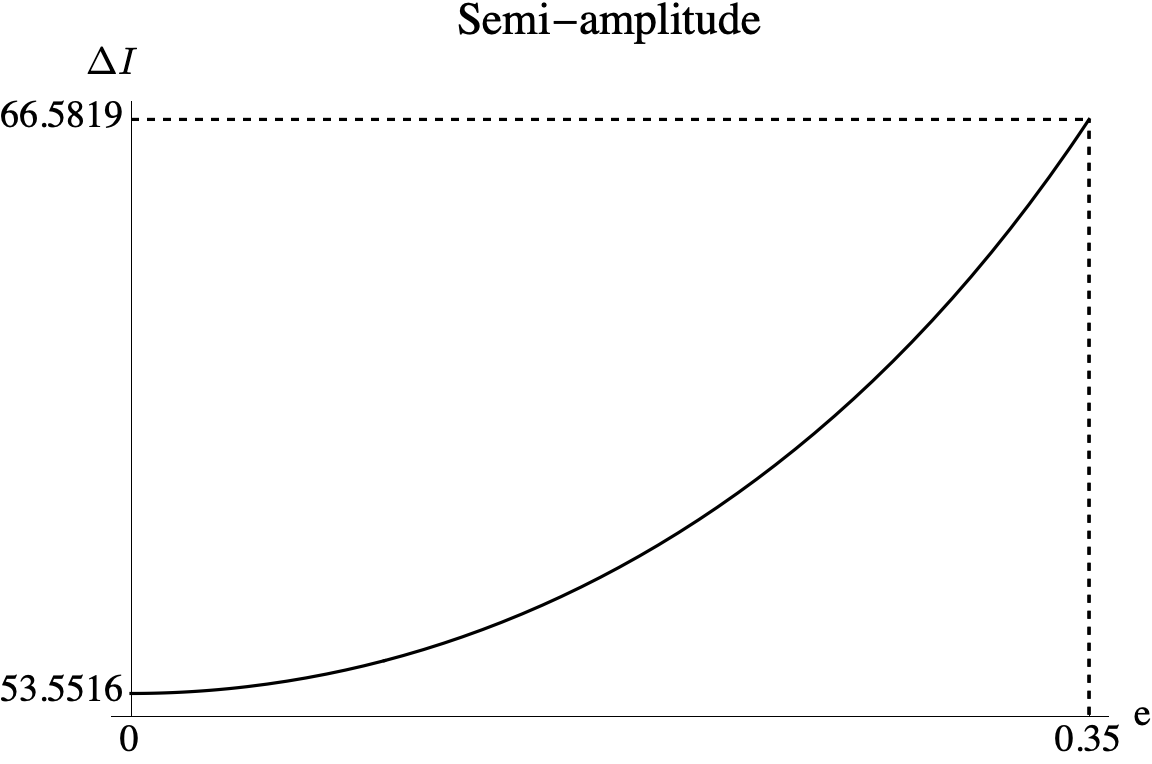}
	\includegraphics[scale=0.38]{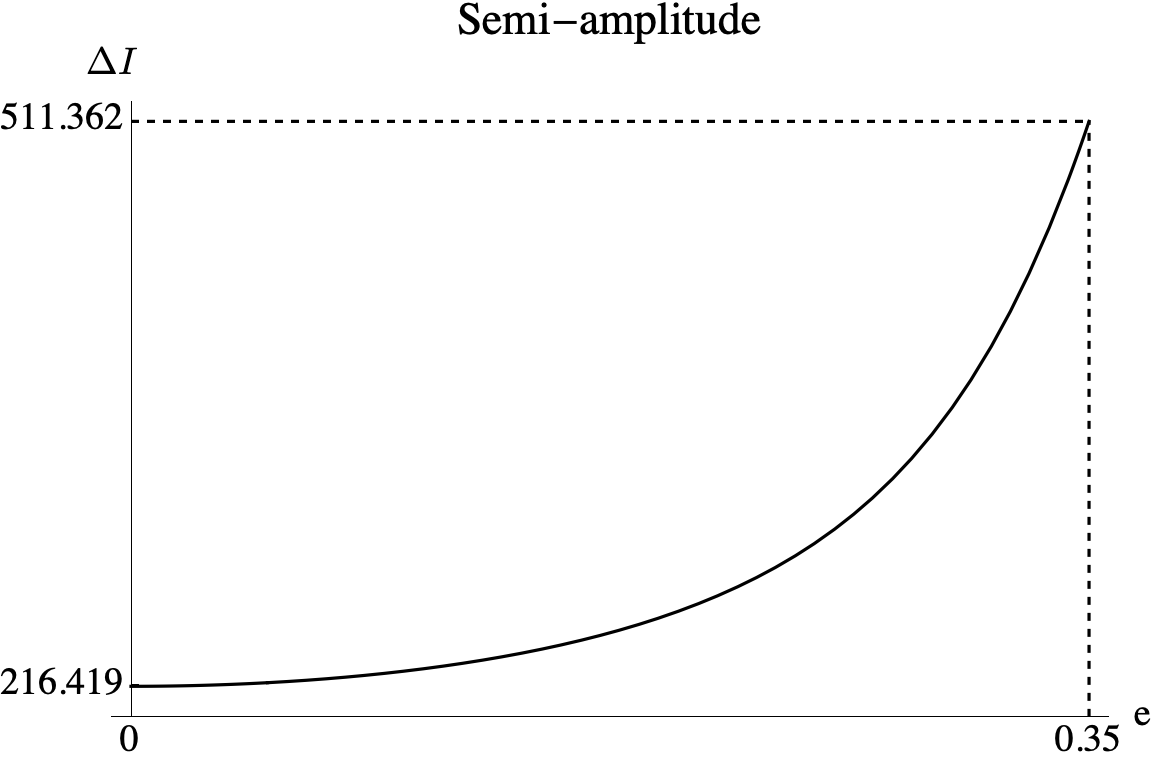}
	\caption{Corotation resonance. Left: almost spherical. Right: highly aspherical. Variation of the semi--amplitude $\Delta I$, according to Eq.~\eqref{amplitude}, as the eccentricity varies up to 0.5.} 
	\label{fig:cor_amplitude}
\end{figure}

A more accurate computation, based on a numerical analysis, is obtained integrating the equations of motions associated to the Hamiltonian~\eqref{ham_cor_Itheta}: 
\beqno
\left\{
	\begin{aligned}
\dot \theta & =  \frac{\partial \H_{res}^{cor}(I,\theta)}{\partial I} = \alpha_1 + 2 \, \alpha_2 \, I \\
\dot I & =  - \frac{\partial \H_{res}^{cor}(I,\theta)}{\partial \theta} =  2\,\alpha_3 \sin 2\, \theta + 4 \alpha_4  \, \sin 4\, \theta + 6\, \alpha_5 \sin 6\, \theta  \\
& \hspace{3.4cm} + \,8\, \alpha_6 \sin 8\, \theta+ 10\, \alpha_7 \sin 10\, \theta\ ;
	\end{aligned}	
	\right.
\eeqno
from the integration, we compute the distance between the separatrices and the elliptic equilibria, obtaining the values
\beqno
\Delta I_{AS} = 53.5554\ ,\qquad 
\Delta I_{HA} = 219.515\ ,
\eeqno
which are in good agreement with the values in~\equ{amplI}. In Figure~\ref{fig:cor_ph_por}, 
we show the phase portraits of the Hamiltonian~\eqref{ham_cor_Itheta}. The equilibria correspond to  $\theta=0$, $\pi/2$, $\pi$,  $3\, \pi/2$; the linear stability of the equilibria in $\theta=0$ and $\theta=\pi$ is of centre type, while the equilibria in $\theta=\pi/2, 3\,\pi/2$ are saddles. In Figure~\ref{fig:cor_ph_por}, we provide the phase portraits for a small eccentricity, say $e=10^{-3}$ (upper panels) and a relatively high eccentricity, say $e=0.35$ (bottom panels). We remark that the amplitudes of libration are greater in the HA case and as the eccentricity increases. For high values of the eccentricity, in the HA case (bottom right panel), we note the occurrence of a saddle--node bifurcation; in the same way, in the bifurcation diagrams given in Figure~\ref{fig:cor_bif}, at around $e=0.328$, we see the appearance of this kind of bifurcation in the HA case (right panel).
 
\begin{figure}[h!]
	\centering
	\includegraphics[scale=0.26]{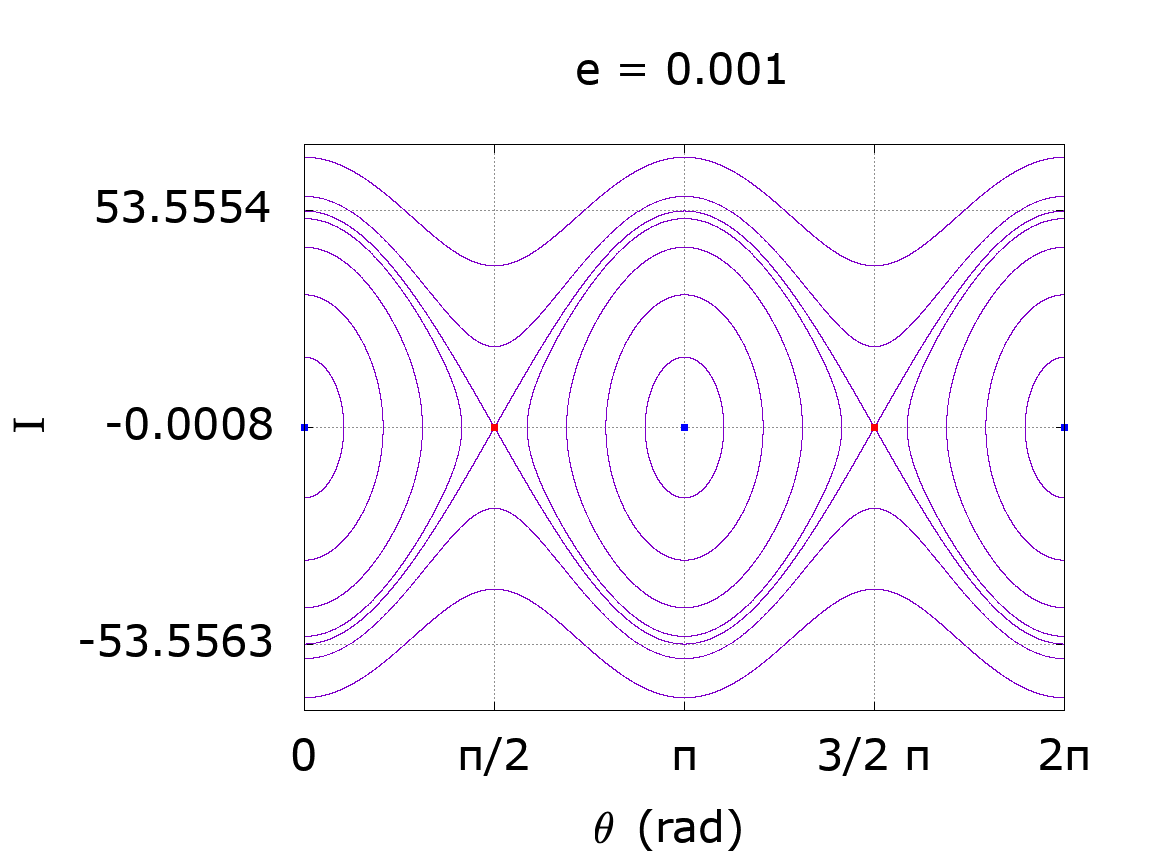}\hspace{-0.3cm}
	\includegraphics[scale=0.26]{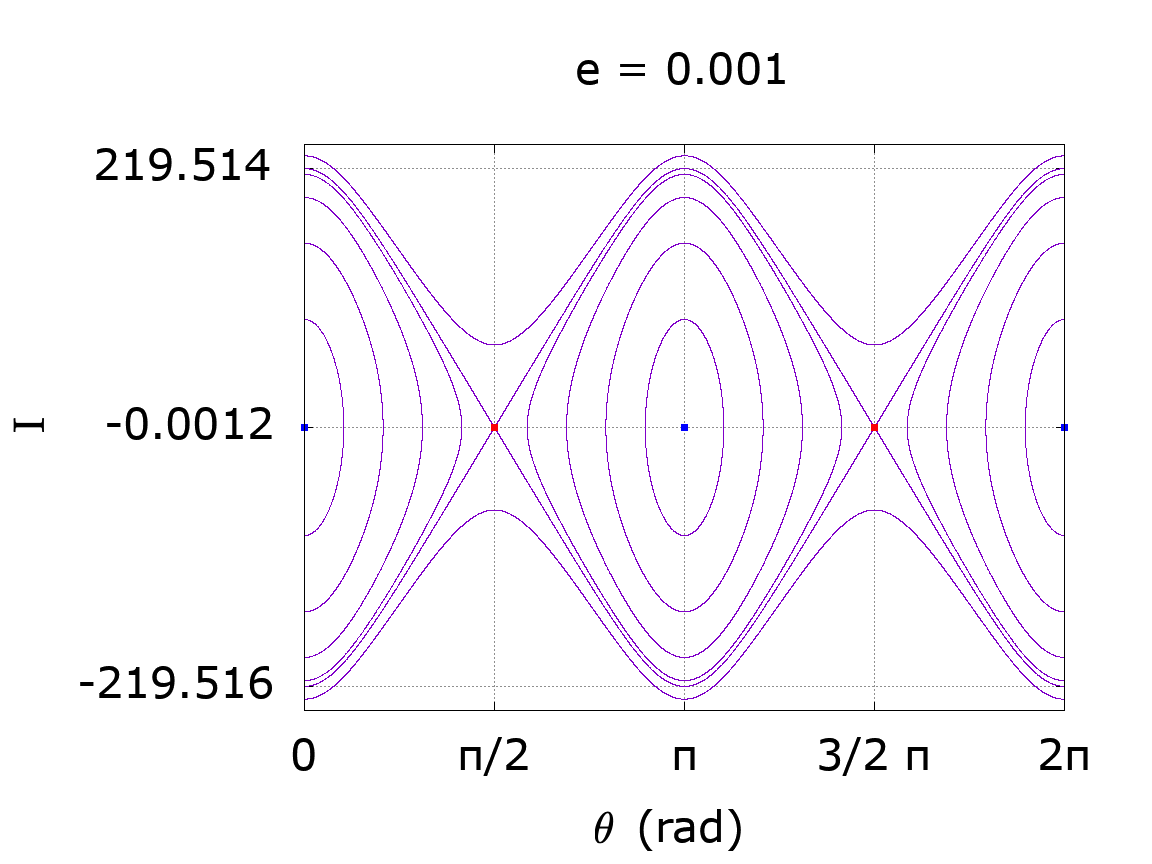}
	\includegraphics[scale=0.26]{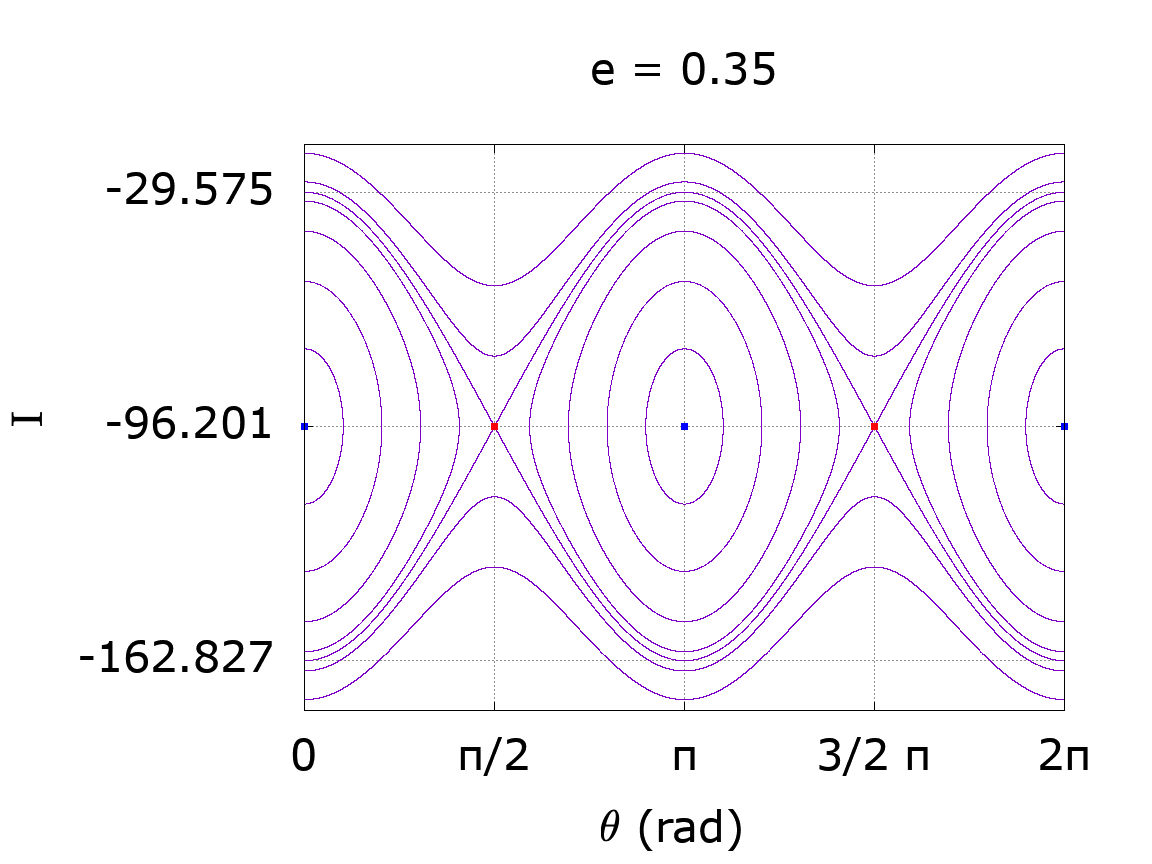}\hspace{-0.3cm}
	\includegraphics[scale=0.26]{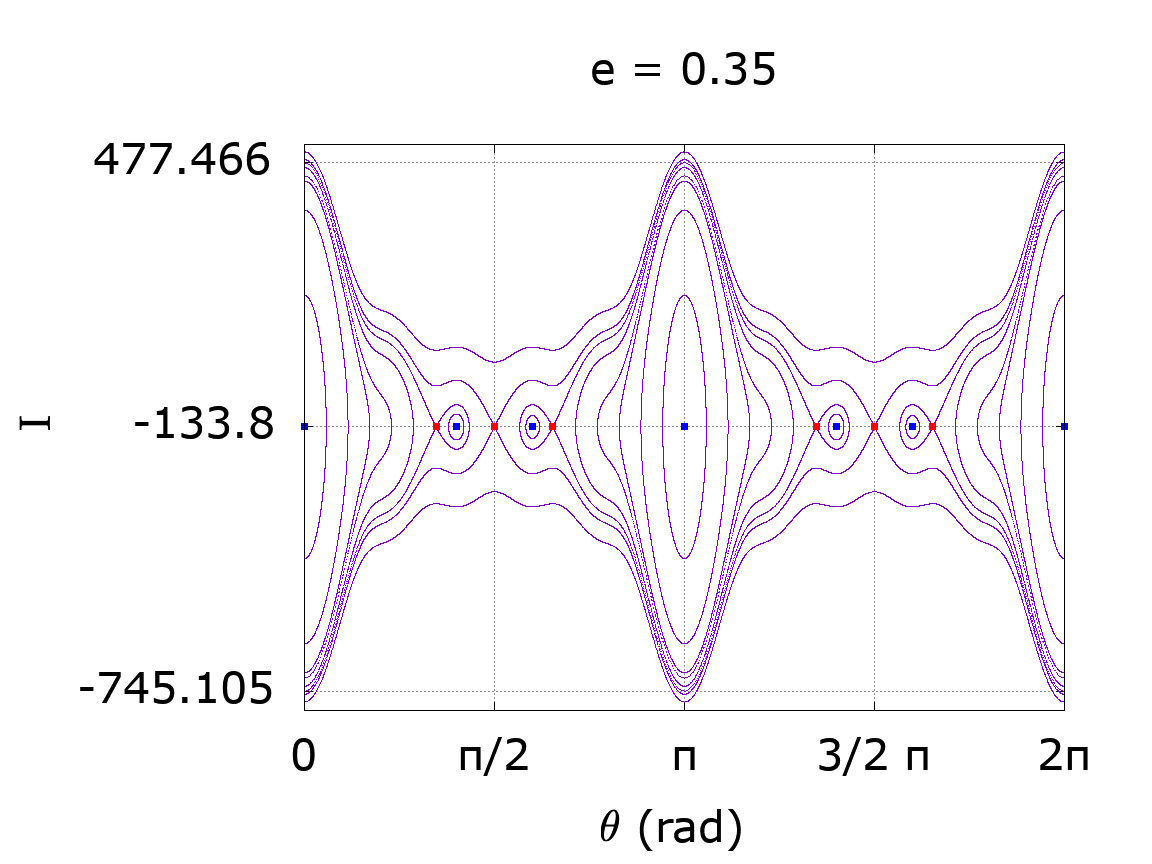}
	\caption{Corotation resonance. Left: almost spherical. Right: highly aspherical. Phase portraits at eccentricity $e=10^{-3}$ (top) and $e=0.35$ (bottom). Red points stand for saddle equilibria, blue points stand for centre equilibria.}
	\label{fig:cor_ph_por}
\end{figure}

\begin{figure}[h!]
	\centering
	\includegraphics[scale=0.37]{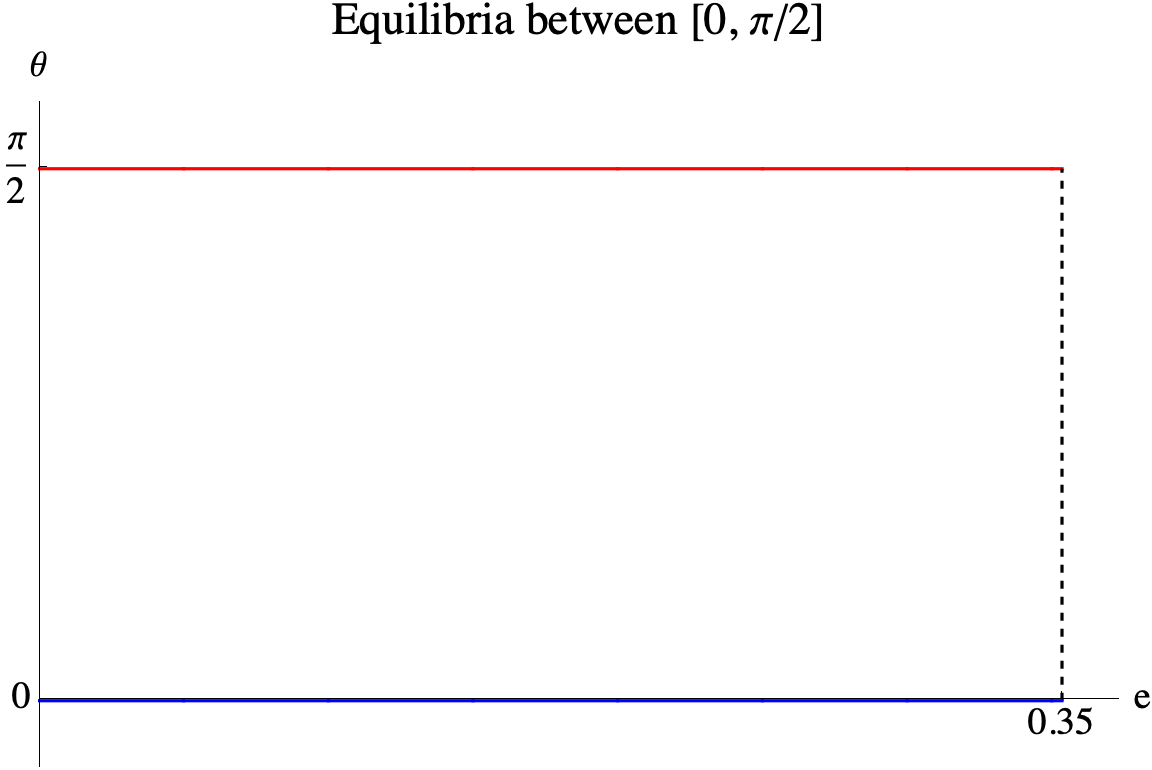}
	\includegraphics[scale=0.4]{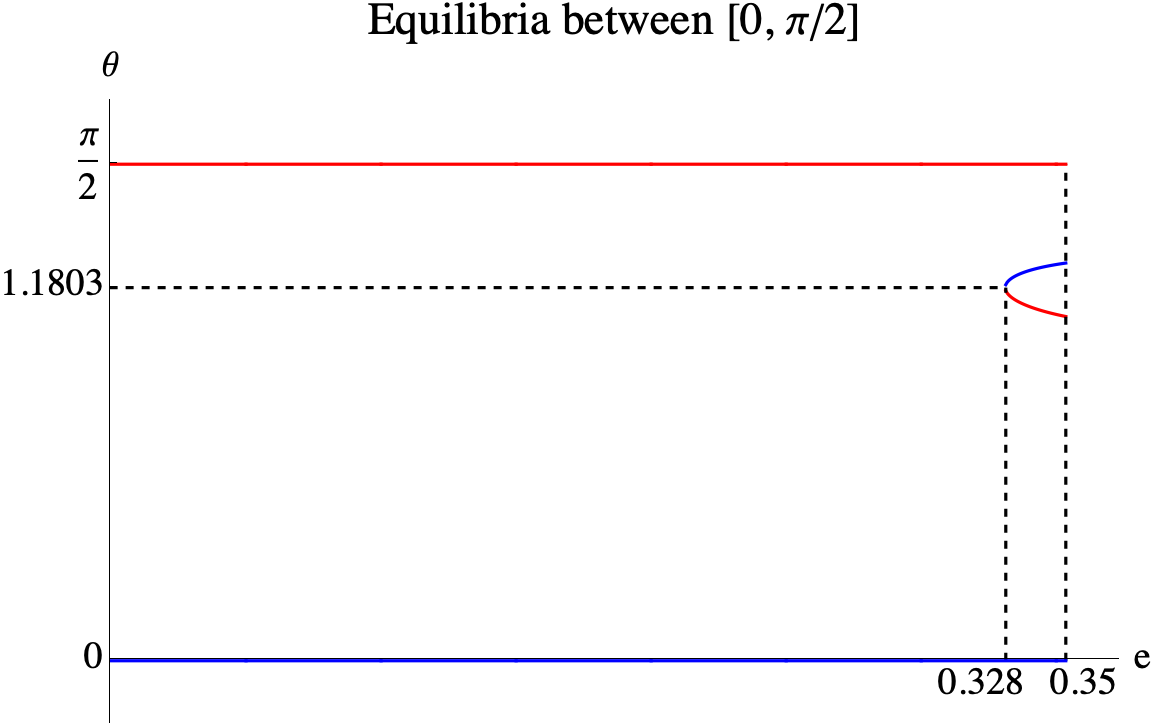}
	\caption{Corotation resonance. Left: almost spherical. Right: highly aspherical. Bifurcation diagrams in the $(e,\theta)$--plane as the eccentricity varies and for $\theta \in[0,\pi/2]$. Red stands for saddle equilibria while blue stands for centre ones.} 
	\label{fig:cor_bif}
\end{figure}

\subsubsection{Resonance $1:2$}

We consider the Hamiltonian~\eqref{eq:Ham_res} taking into account the normal part computed in $r_*= r_{12}$ plus the part depending on the angles associated to the $1:2$ resonance, which corresponds to take in Eq.~\eqref{eq:HIJ} only the angle $\theta+\varphi$ and its multiples, according to Section~\ref{Sec:resepi}; we notice that the Hamiltonian~\eqref{F_tilde_plus_F1_tilde} contains the resonant combination $2\theta+2\varphi$ (instead of $\theta+\varphi$) and its multiples. Hence, we obtain the Hamiltonian 
\beqa{res_ham_12_AS}
\H_{res}^{12}(J,I,\varphi,\theta) 
& = &   |\kappa_*| \, J + (n_*-\Omega_P) \, I + \frac{I^2}{2 \, r_*} +f_1(J) \, I + f_2(J) \,  I^2 \nonumber \\
&&  + \sum_{i=1}^5 \, f_{i+2}(J)\, \cos(2\,i\, \theta + 2\,i \, \varphi)  \, ,
\eeqa
where $\kappa_* = \kappa(r_{12})$, $n_* = n(r_{12})$ with $r_{12}$ computed through Eq.~\eqref{linres}. The terms $f_j(J)$, $j=1,\ldots,7$, are polynomials of order 8 in the variable $J$. To isolate the resonant angle, let us introduce the canonical transformation 
\beq{trans_12_bis}
\left\{
	\begin{aligned}
\psi & =  2 \, \varphi + 2 \, \theta    \\
\mu & =  2 \, \varphi +  \theta   \\
G & = -\frac{1}{2} \, J +  I \\
L &= J - \, I 
\end{aligned}	
	\right. \qquad .
\eeq
In these new variables, $\mu$ is cyclic, thus the action $L$ is a constant, say $L=L_0$, and the resonant Hamiltonian becomes
\beq{res_ham_12_new}
\H_{res}^{12}(G,\psi; L_0)  =  \sum_{i=1}^{10}\gamma_i(L_0) \, G^i + \sum_{i=1}^5 \gamma_{i+10}(G;L_0)\, \cos(i\, \psi)\, ,
\eeq
where the constants $\gamma_1,...,\gamma_{10}$ and the polynomials $\gamma_{11},...,\gamma_{15}$ can be computed through Eqs.~\equ{res_ham_12_AS}  and \equ{trans_12_bis}. 
Among the possible choices of the initial condition, we adopt the following procedure to define an integral level $L=L_0$: let us fix the constant $L_0 = J_0- I_0$ from the initial values $J_0$, $I_0$. As in the corotation case, we can compute $J_0$ by setting suitable values of $\rho$ and $p_{\rho}$; fixing, for example, $e= 10^{-3}$, we get $\rho = 10^{-3}\ r_{12}$ and hence, from Eq.~\eqref{J_0} (substituting the subscript {\it cor} by 12), we obtain 
$$
J_{0,AS} = 3.4527\cdot 10^{-4}\ , \qquad J_{0,HA} = 3.5465\cdot 10^{-4}\, .
$$
Moreover, by definition, we have $I = p_{\theta}-p_*$, where $p_{\theta}= n(r) \,r^2$ and $p_* = n_* \, r_*^2$ with $r_*=r_{12}$, $n_* = n(r_*)$. After  computing $p_{\theta}$ in $r = r_{12}+\rho$, we get;
$$
I_{0,AS} = 1.7064\cdot 10^{-1} \, , \qquad I_{0,HA} = 1.5451\cdot 10^{-1} \, ,
$$
from which we obtain 
$$
L_{0,AS} = -1.7029\cdot 10^{-1}\, , \qquad L_{0,HA} = -1.5416\cdot 10^{-1}\ , 
$$
that, substituted in Eq.~\eqref{res_ham_12_new} and omitting constant terms, give the 1 DOF Hamiltonian
\beqa{res_ham_12_Gpsi}
\H_{res}^{12}(G,\psi)  =  \sum_{i=1}^{10}\alpha_i \, G^i + \sum_{i=1}^5 \delta_{i}(G)\, \cos(i\, \psi)\, ,
\eeqa
with constants $\alpha_1,...,\alpha_{10}$ given in Table~\ref{tab:coeffs_12} and $\delta_{1},...,\delta_{5}$ polynomials of order 8 in $G$ given in Appendix~\ref{app:appB}. We stress that the coefficients of the trigonometric terms in Eq.~\equ{res_ham_12_Gpsi} depend on the action $G$. 

\begin{table}[h!]
	\begin{tabular}{|c|c|c|}
		\hline
Coefficients &  $1:2$ AS & $1:2$ HA  \\
		\hline
		$\alpha_1$ & $-2.0\cdot 10^{-6}$ & $-1.8\cdot 10^{-6}$  \\
		\hline
		$\alpha_2$  &$7.2\cdot 10^{-6}$ & $7.3\cdot10^{-6}$  \\
		\hline
		$\alpha_3$ &  $1.1 \cdot 10^{-7}$ & $1.1\cdot10^{-7}$ \\
		\hline
		$\alpha_4$  & $1.5\cdot 10^{-9}$ & $1.4\cdot10^{-9}$  \\
		\hline
		$\alpha_5$  & $2.0\cdot 10^{-11}$ & $1.5\cdot10^{-11}$\\
		\hline
		$\alpha_6$  & $2.6\cdot 10^{-13}$ & $1.1\cdot10^{-13}$\\
		\hline
		$\alpha_7$  & $3.3\cdot 10^{-15}$ & $-9.6\cdot10^{-16}$\\
		\hline
		$\alpha_8$  & $4.2\cdot 10^{-17}$ & $-7.2\cdot10^{-17}$\\
		\hline
		$\alpha_9$  & $3.0\cdot10^{-19}$ & $4.1\cdot10^{-19}$\\
		\hline
		$\alpha_{10}$  & $7.2\cdot10^{-22}$ & $9.7\cdot10^{-22}$\\
		\hline
	\end{tabular}
	
	\vskip.1in 
	
\caption{Coefficients appearing in the $1:2$ resonant Hamiltonian~\eqref{res_ham_12_Gpsi} for the sample cases AS and HA at $e = 10^{-3}$.}\label{tab:coeffs_12} 	
\end{table}

The equilibrium points associated to Eq.~\equ{res_ham_12_Gpsi}, depending on the eccentricity, are given by the solution of the system:
\beqno
	\left\{
	\begin{aligned}
\dot \psi & =  \frac{\partial \H_{res}^{12}(G,\psi)}{\partial G} = 0 \\
\dot G & =  - \frac{\partial \H_{res}^{12}(G,\psi)}{\partial \psi} = 0 
	\end{aligned}	
	\right. \qquad .
\eeqno
For $e=10^{-3}$, the equilibria of AS are:
\beqano
&& (\psi_1 = 0, \,G_{1} =  0.0549)\,,\quad  (\psi_2 = \pi , \, G_{2} = 0.2268)\, ,\\ 
&& (\psi_3 = 1.9231 , \,G_{3} = 0.1703) \,,\quad (\psi_4 = 4.3600 , \, G_{4} = 0.1703) \ 
\eeqano
and for HA they are given by  
\beqano
&& (\psi_1 = 0, \,G_{1} = -1.3113)\,,\quad  (\psi_2 = \pi , \, G_{2} = 1.8061)\, ,\\ 
&& (\psi_3 = 1.5891 , \,G_{3} = 0.1541) \,,\quad (\psi_4 = 4.6940 , \, G_{4} = 0.1541) \ .
\eeqano
The equilibria $(\psi_1, G_1)$, $(\psi_2, G_2)$ are elliptic points, while $(\psi_3, G_3)$ and $(\psi_4, G_4)$ are hyperbolic points. Computing the amplitude of the resonant islands (obtained by evaluating the distance between the separatrices), one gets
\beqano
\Delta G_{AS}^{(1)}&=& 2.3101 \cdot 10^{-1}  \ , \qquad 
\Delta G_{AS}^{(2)}= 1.1308\cdot 10^{-1} \, ,\nonumber\\
\Delta G_{HA}^{(1)}&=&2.9657 \ , \qquad \qquad \ \ 
\Delta G_{HA}^{(2)}= 3.2692 \, .
\eeqano
These values fully agree with the results of the numerical integration of Hamilton's equations associated to Eq.~\equ{res_ham_12_Gpsi} as shown in the phase portraits (top panels) of Figure~\ref{fig:12_ph_por}, on the left for the AS case and on the right for the HA case. We remark that in this case the computed values are not the semi--amplitudes as in the corotation case, but the amplitudes of the whole libration islands which are asymmetric with respect to the elliptic equilibria.

\begin{figure}
	\centering
	\includegraphics[scale=0.25]{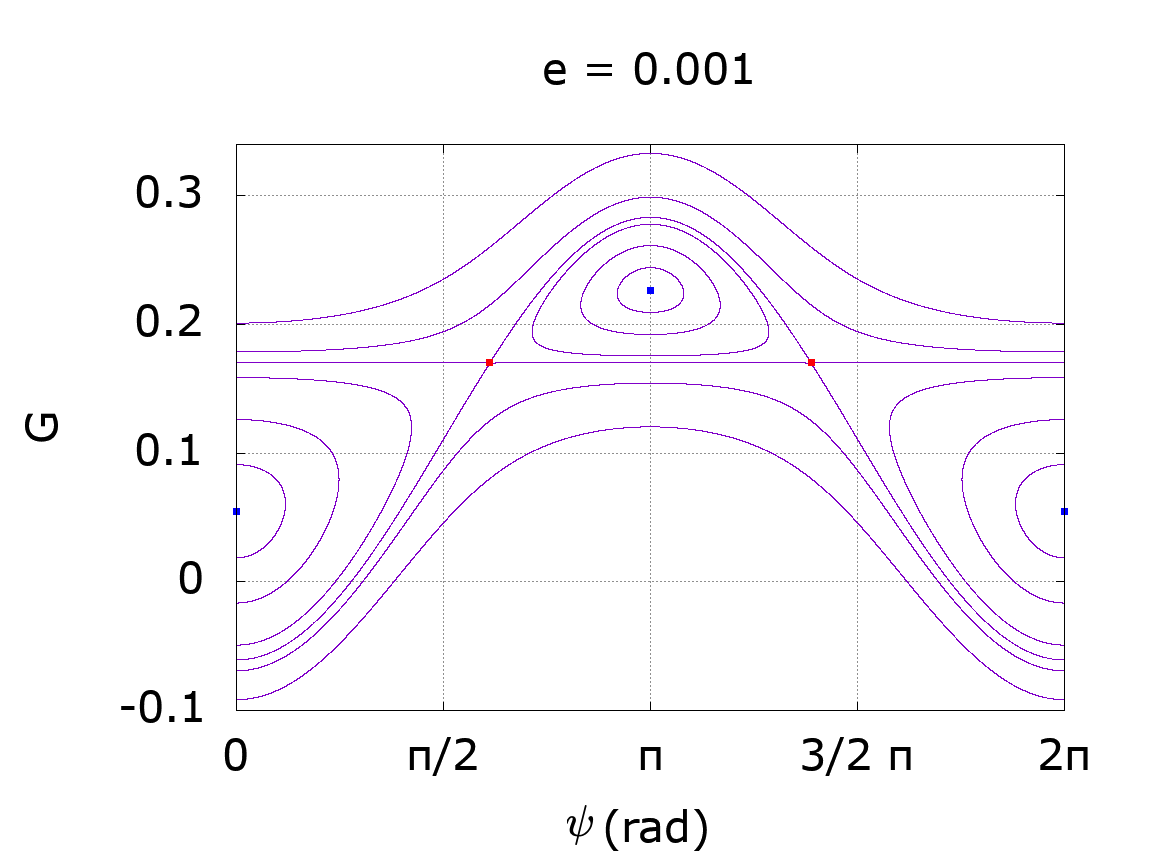}
	\includegraphics[scale=0.25]{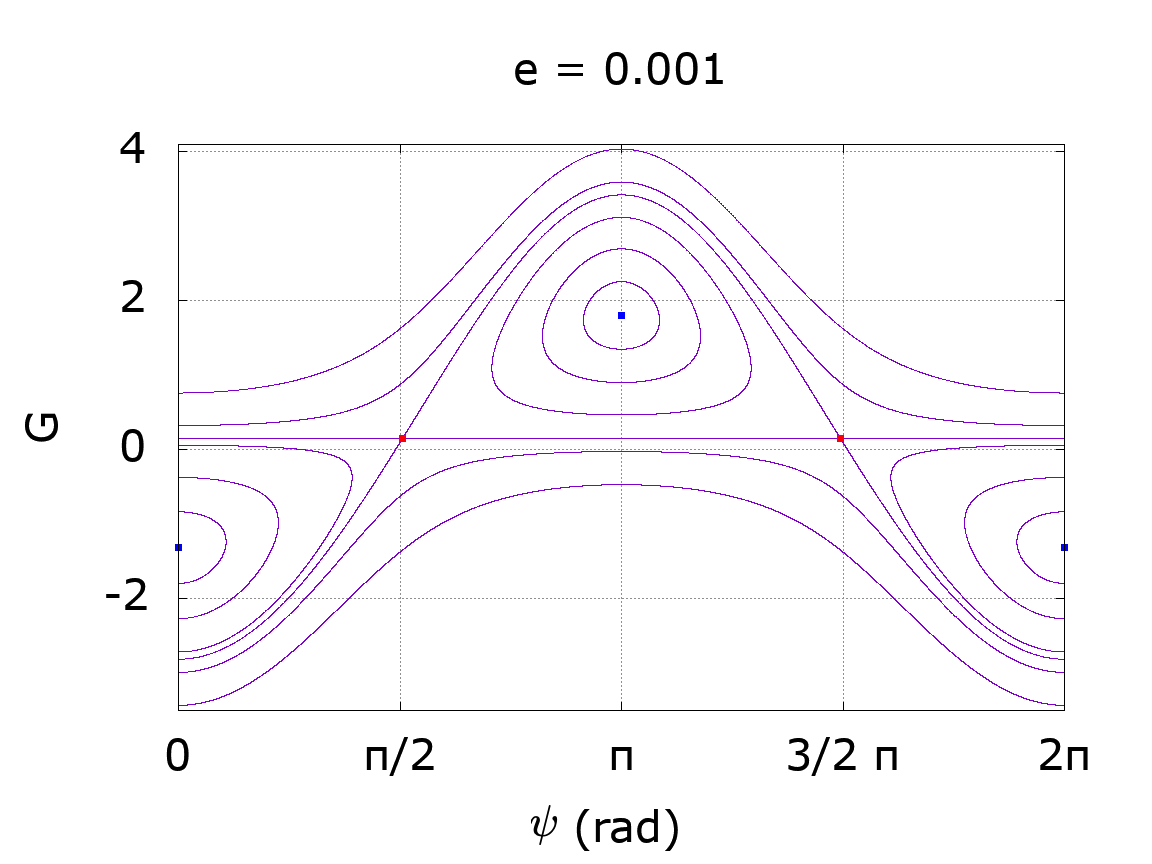}\\
	\includegraphics[scale=0.25]{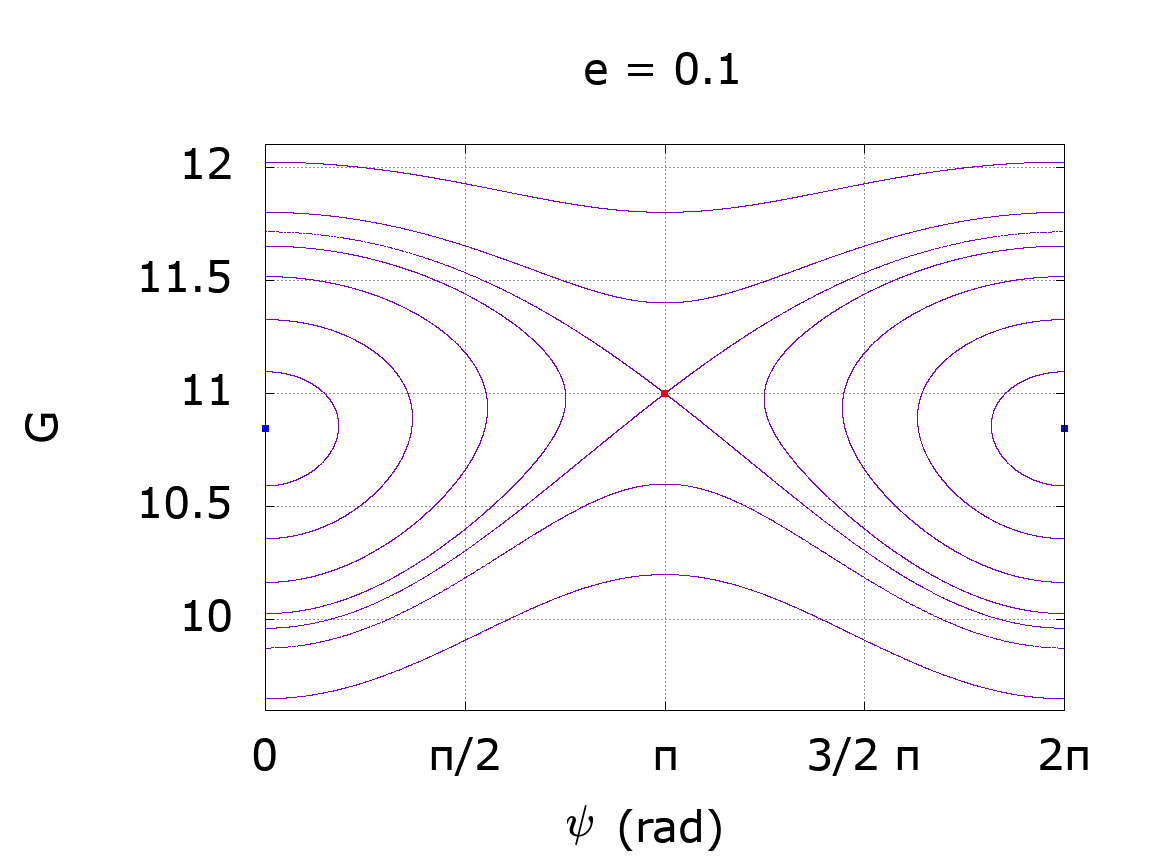}
	\includegraphics[scale=0.25]{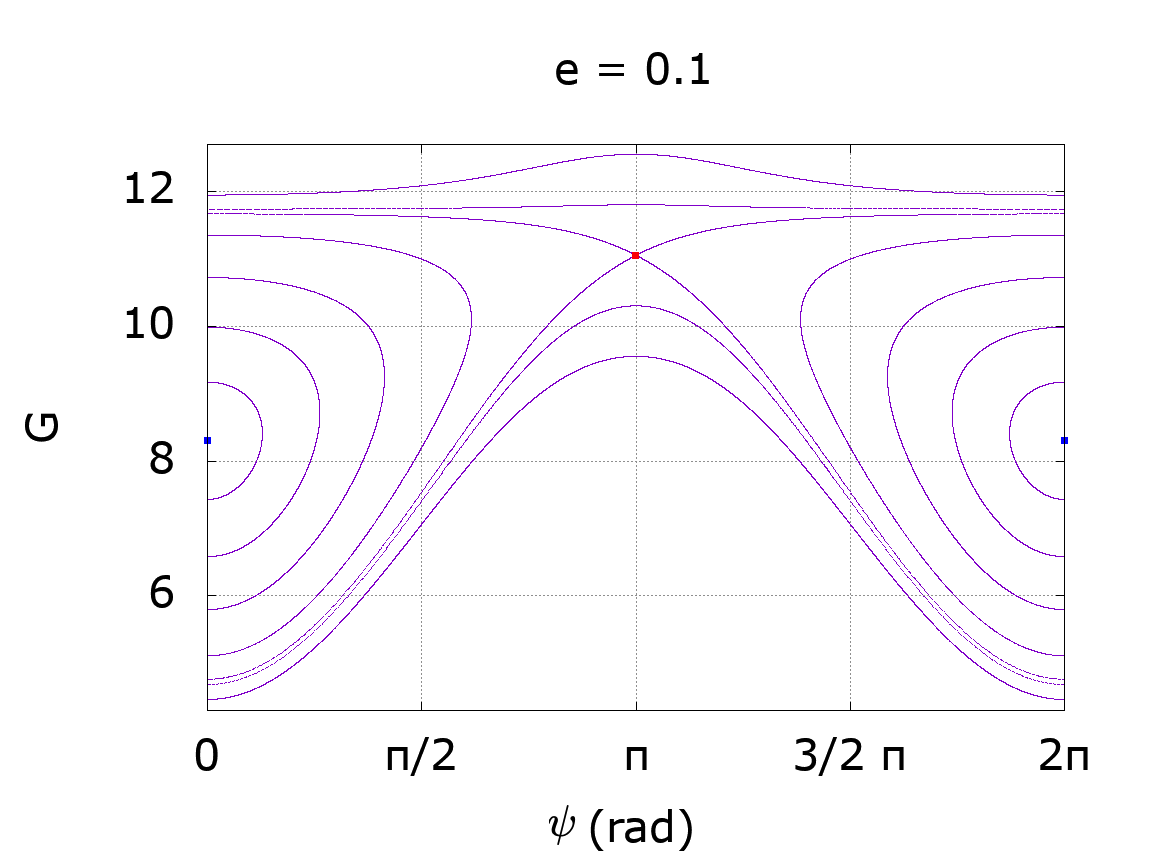}\\
	\includegraphics[scale=0.25]{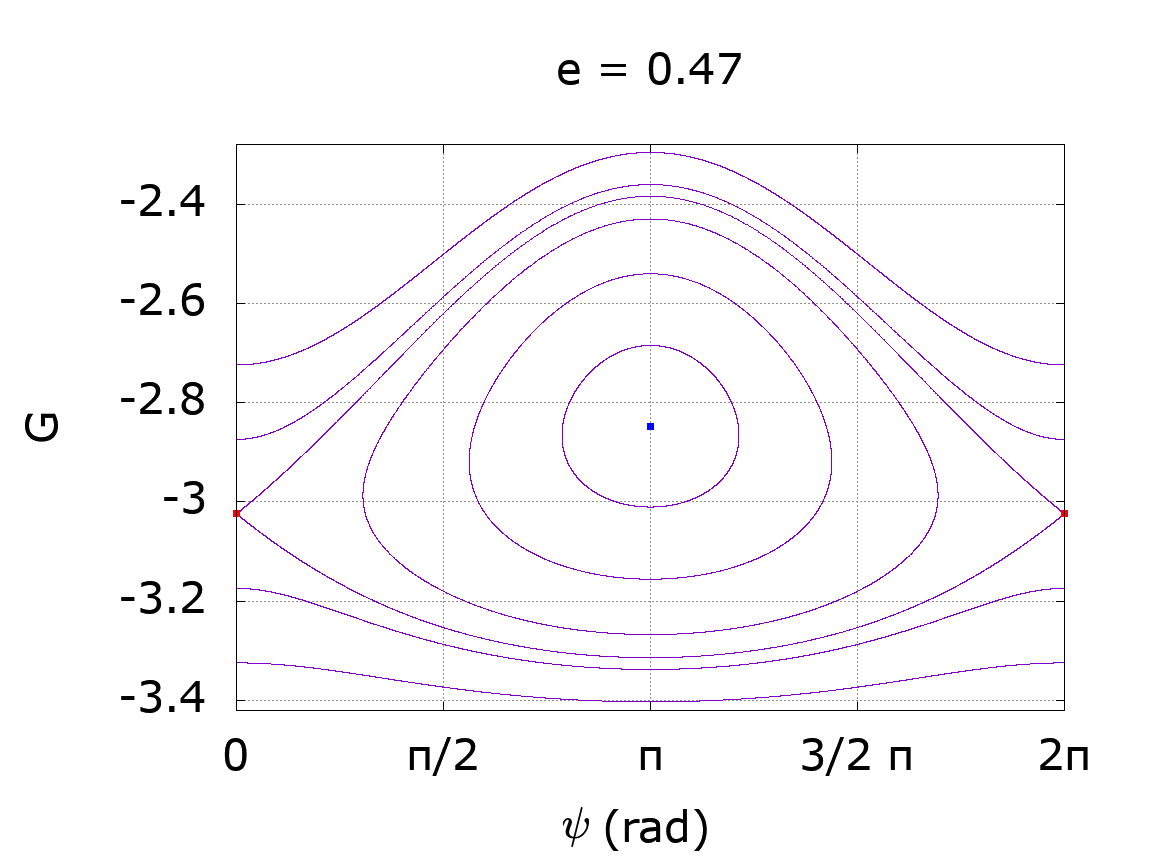}
	\includegraphics[scale=0.25]{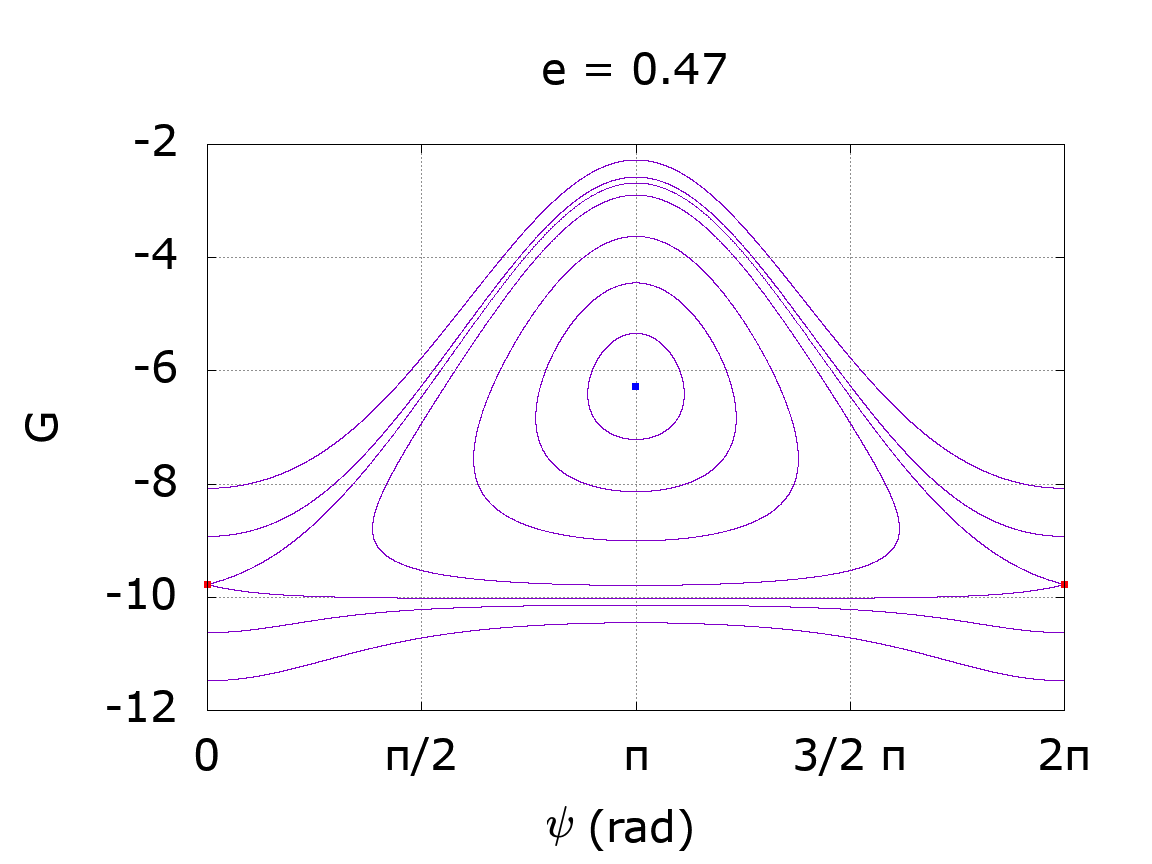}\\
	\caption{Phase portraits of the $1:2$ resonance for the almost spherical case on the left panels and for the highly aspherical case on the right panels, for three different values of the eccentricity. Red points stand for saddle equilibria, blue stand for centre ones.}
	\label{fig:12_ph_por}
\end{figure}

As already evident in Figure~\ref{fig:12_ph_por} where the phase portraits are represented for 3 different values of the eccentricity, the $1:2$ resonance presents a great variety of different situations as the eccentricity increases. Concerning the bifurcations, we observe that Hamilton's equations associated to Eq.~\eqref{res_ham_12_Gpsi} admit trivial equilibria for $\psi=0$ and $\psi=\pi$; Figure~\ref{fig:12_G1} summarises their location and stability in the AS and HA cases. In both cases, the equilibria change their stability as $e$ varies from 0 to 0.5. The equilibrium in $\psi =0$ is stable at $e=0$, and becomes unstable at $e = 0.4526$ for the AS case and $e = 0.4625$ in the HA case. For $\psi =\pi$, a double change of stability is observed. The equilibrium is stable at $e=0$, it becomes unstable at $e=0.0029$ ($e=0.0638$) for the AS case (HA case) and then, it becomes again stable at $e= 0.4467$ ($e=0.3491$) for the AS case (HA case).

\begin{figure}[h]
		\centering
		\includegraphics[scale=0.4]{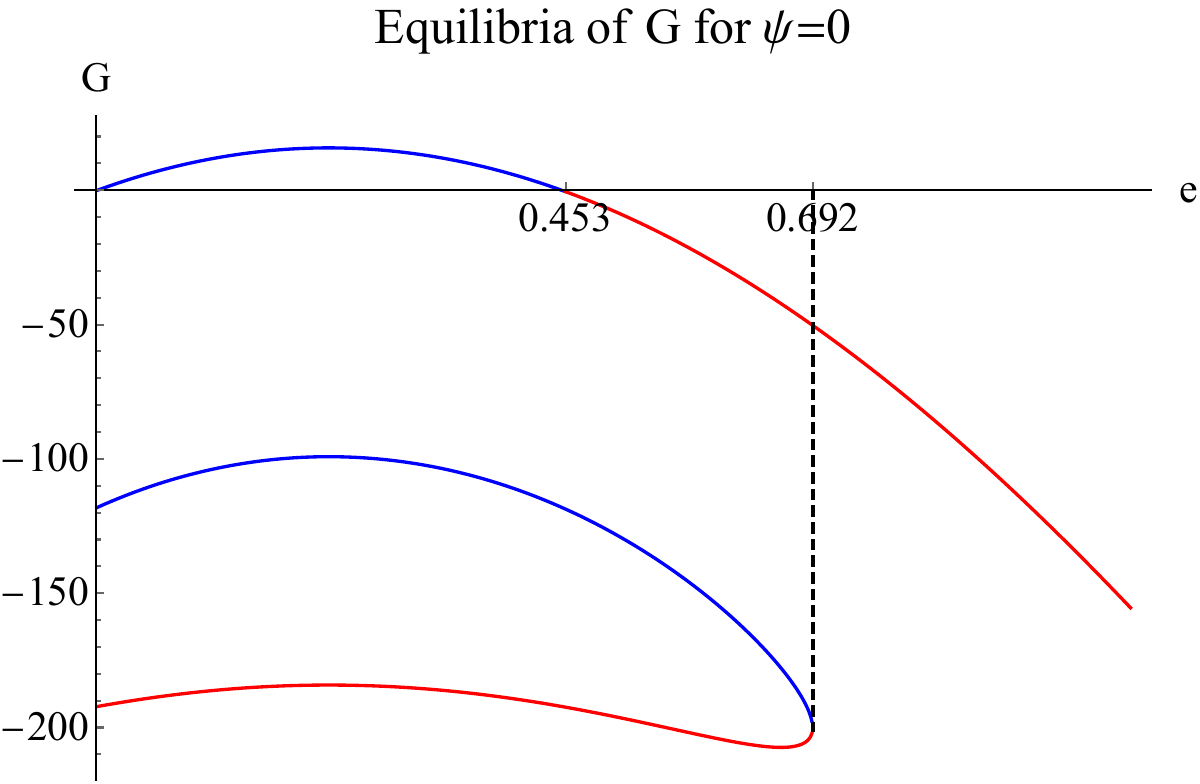}\hglue0.2cm
		\includegraphics[scale=0.4]{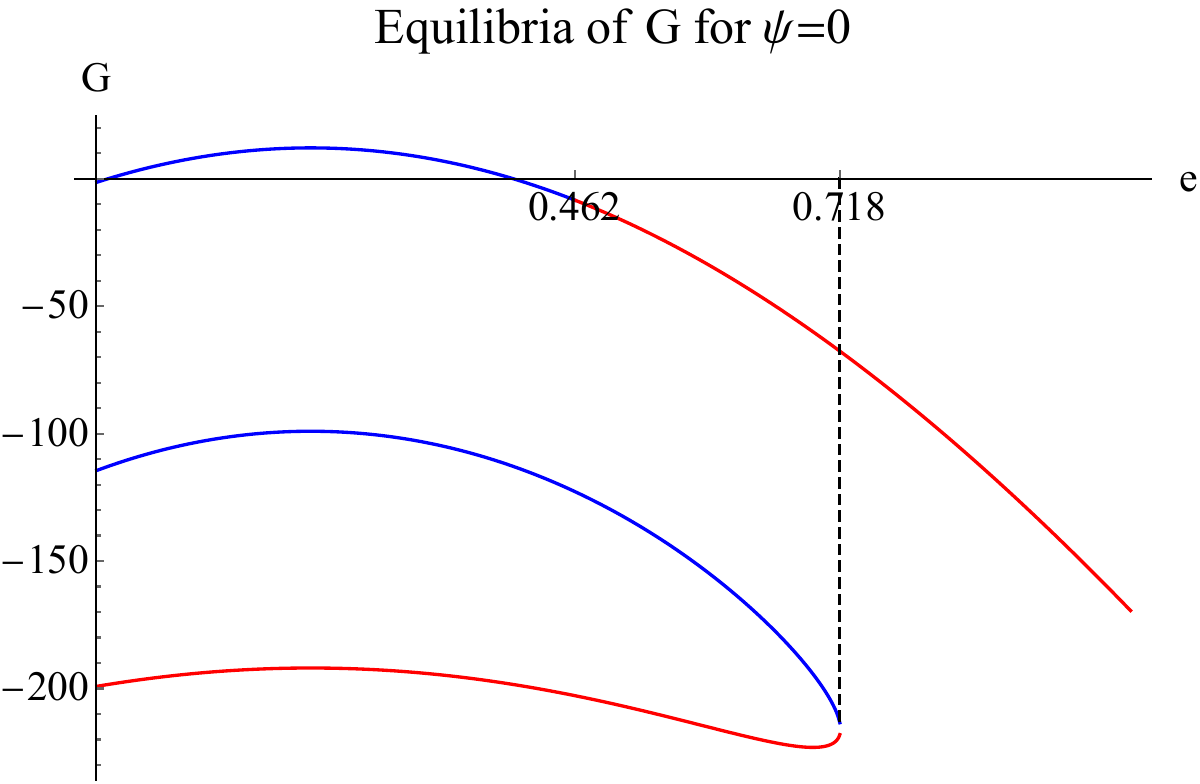}\\
		\includegraphics[scale=0.4]{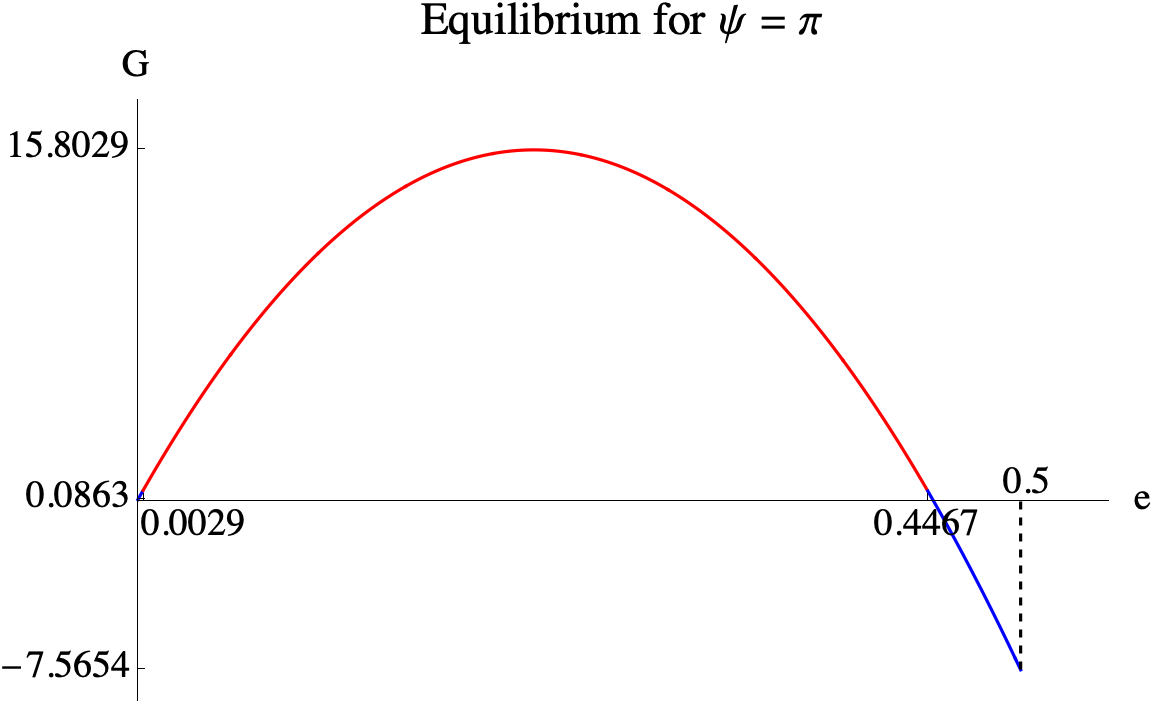}\hglue0.2cm
		\includegraphics[scale=0.4]{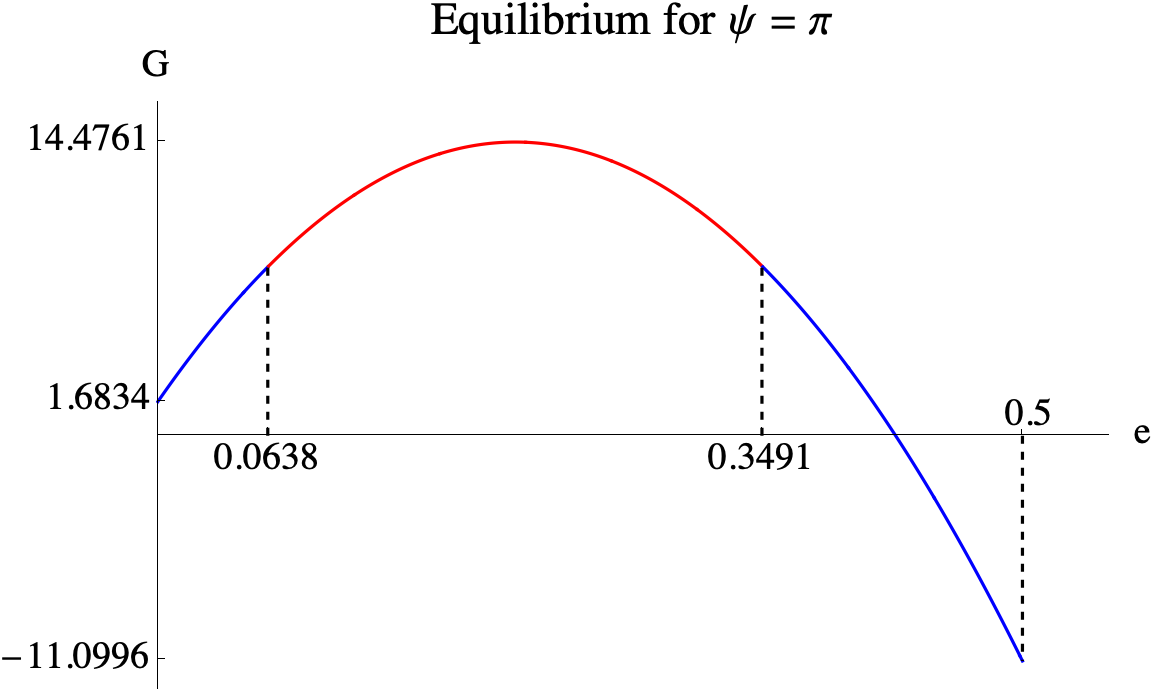}\\
		\caption{Resonance $1:2$. Left: almost spherical. Right: highly aspherical. Location and stability in the $(e, G)$--plane of the equilibria for $\psi =0$ (top) and $\psi=\pi$ (bottom) as the eccentricity varies. Blue stands for centre, red stands for saddle.  }
		\label{fig:12_G1}
	\end{figure}

Figure~\ref{fig:12_bif_G0} provides a finer analysis, showing the occurrence of pitchfork bifurcations for the $1:2$ resonance. For $\psi =0$ (mod. $2\pi$) and for $\psi =\pi$, the plots provide the changes in stability described in Figure~\ref{fig:12_G1}. In particular, for $\psi = \pi$ we have two pitchfork bifurcations at $e=0.0029$ and $e=0.4467$ for the AS case and at $e=0.0638$ and $e=0.3491$ for the HA case. For very small values of $e$, there are two secondary equilibria absorbed by $\psi=\pi$ as $e$ increases. For higher values of $e$, the secondary equilibria appear again, branching from $\psi=\pi$, and being absorbed by $\psi=0$ afterwards. \\

\begin{figure}[h]
	\centering
	\includegraphics[scale=0.4]{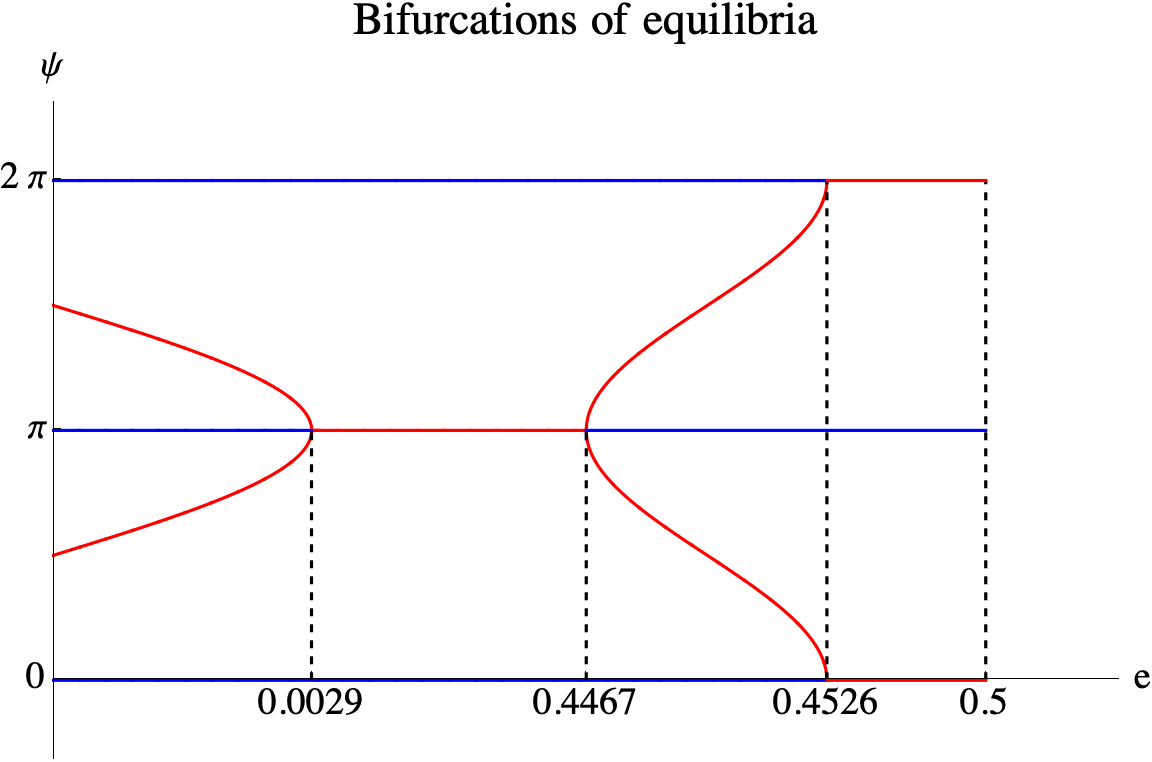}\hglue0.2cm
	\includegraphics[scale=0.4]{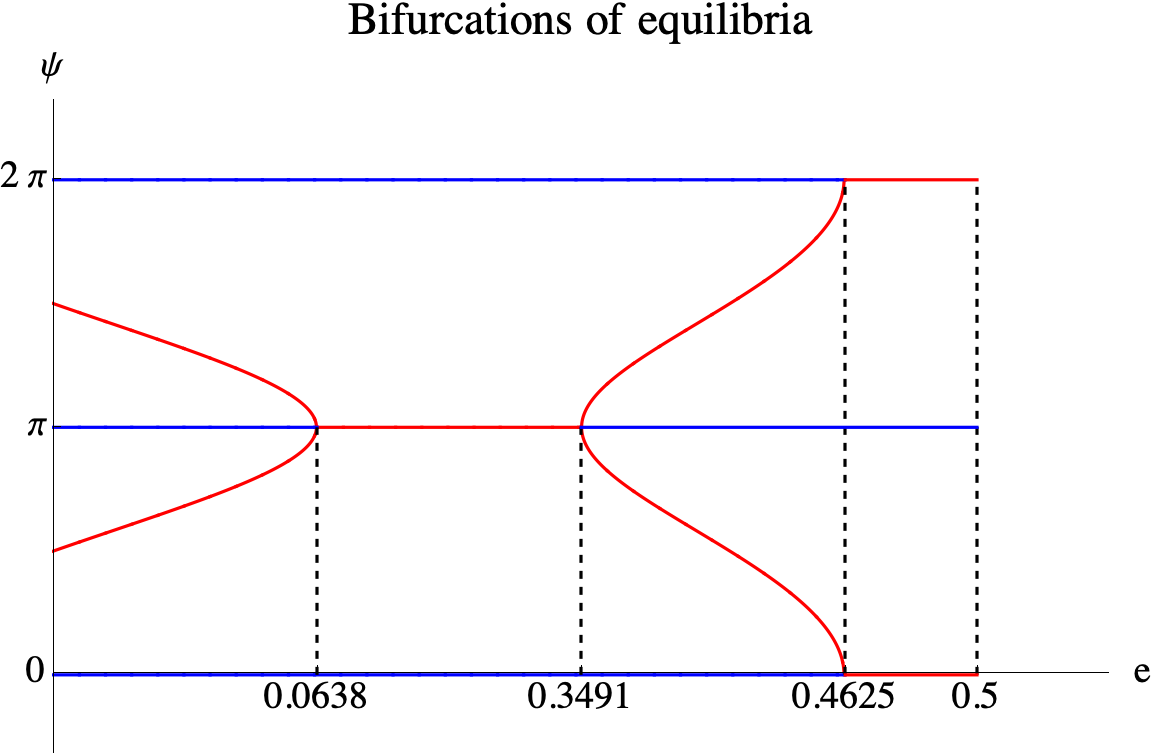}
	\caption{Resonance $1:2$. Left: almost spherical. Right: highly aspherical. Bifurcation diagrams in the $(e,\psi)$--plane of the equilibria in $\psi=0,\, \pi$. Blue stands for centre, red stands for saddle.}
	\label{fig:12_bif_G0}
\end{figure}

\subsubsection{Resonance $1:3$}

Let us consider the Hamiltonian~\eqref{eq:Ham_res} taking into account the normal part in $r_*= r_{13}$ and the part depending on the angles associated to the $1:3$ resonance, corresponding to take in Eq.~\eqref{eq:HIJ} the resonant combination $\theta+2\varphi$ and its multiples (see Section~\ref{Sec:resepi}): 
\beqa{Hres13_AS}
\H_{res}^{13}(J,I,\varphi,\theta) 
& = &   |\kappa_*| \, J + (n_*-\Omega_P) \, I + \frac{I^2}{2 \, r_*}   + f_1(J) \, I +f_2(J) \, I^2 \nonumber \\
&&  + \, \sum_{i=1}^4 f_{i+2}(J) \cos(2\,i \, \theta + 4\, i\, \varphi)    \, ,
\eeqa
where the terms $f_j(J)$, $j=1,\ldots,7$, are polynomials of order 8 in the variable $J$.

As in the case of the $1:2$ resonance, let us consider the canonical transformation 
\beq{trans_13}
	\left\{
	\begin{aligned}
\psi & =  4 \, \varphi + 2 \, \theta \\
\mu & =  \varphi +  \theta    \\
G & =\frac{1}{2} \, J - \frac{1}{2} \,  I \\
L & = - J + 2\, I
\end{aligned}	
	\right. \qquad .
\eeq
In the new variables, $\mu$ is cyclic, hence $L=L_0$ is constant, and the resonant Hamiltonian takes the form 
\beq{res_ham_13_new}
\H_{res}^{13}(G,\psi;L_0)  = \sum_{i=1}^{10} \gamma_i(L_0) \, G^i + \, \sum_{i=1}^4 \gamma_{i+10}(G;L_0)\, \cos(i\,\psi) 
\eeq
where the constants $\gamma_1,...,\gamma_{10}$ and the polynomials $\gamma_{11},...,\gamma_{14}$ can be computed through Eqs.~\equ{Hres13_AS} and \equ{trans_13}. Let us choose initial conditions to fix $L_0$ through $L_0 = -J_0 +2\,I_0$; we can compute $J_0$ by setting suitable values for $\rho$ and $p_{\rho}$, say $\rho = 10^{-3} \, r_{13}$ (corresponding to $e=10^{-3}$). Hence, from Eq.~\eqref{J_0} we obtain 
$$
J_{0,AS} = 3.9470\cdot 10^{-4}\ , \qquad J_{0,HA} = 3.9985\cdot 10^{-4}\, .
$$
Using the same procedure as in the $1:2$ resonance, we get
$$
I_{0,AS} = 1.9600\cdot 10^{-1} \, , \qquad I_{0,HA} = 1.8614\cdot 10^{-1}\, ,
$$
from which we obtain 
$$
L_{0,AS} = 3.9162\cdot 10^{-1}\, , \qquad L_{0,HA} = 3.7189\cdot 10^{-1} \, .
$$
Substituting $L_0$ in Eq.~\eqref{res_ham_13_new}, we arrive at the Hamiltonian
\beq{res_ham_13_Gpsi}
\H_{res}^{13}(G,\psi) =\sum_{i=1}^{10} \alpha_i \, G^i + \, \sum_{i=1}^4 \delta_{i}(G)\, \cos(i  \, \psi)
\eeq
with constants $\alpha_1,...,\alpha_{10}$ given in Table~\ref{tab:coeffs_13} and the polynomials $\delta_{1},\ldots, \delta_{4}$ of order 8 in $G$ given in Appendix~\ref{app:appB}; also in this case, the coefficients of the trigonometric terms depend on the action $G$.

\begin{table}[h!]
	\begin{tabular}{|c|c|c|}
		\hline
Coefficients & $1:3$ AS & $1:3$ HA \\
		\hline
		$\alpha_1$  &$2.7\cdot 10^{-6}$  & $2.6\cdot 10^{-6}$ \\
		\hline
		$\alpha_2$  &$1.1\cdot 10^{-5}$ & $1.1\cdot 10^{-5}$ \\
		\hline
		$\alpha_3$ &$3.1\cdot 10^{-7}$ & $3.1\cdot 10^{-7}$\\
		\hline
		$\alpha_4$ &$7.7\cdot 10^{-9}$ & $7.5\cdot 10^{-9}$ \\
		\hline
		$\alpha_5$ &$1.7\cdot 10^{-10}$ & $1.6\cdot 10^{-10}$\\
		\hline
		$\alpha_6$ &$4.0\cdot 10^{-12}$ & $3.5\cdot 10^{-12}$\\
		\hline
		$\alpha_7$ &$9.0\cdot 10^{-14}$ & $6.5\cdot 10^{-14}$\\
		\hline
		$\alpha_8$ &$1.9\cdot 10^{-15}$ & $9.9\cdot 10^{-16}$\\
		\hline
		$\alpha_9$ &$1.5\cdot 10^{-17}$ & $1.8\cdot 10^{-17}$\\
		\hline
		$\alpha_{10}$ &$3.6\cdot 10^{-20}$ & $4.2\cdot 10^{-20}$\\
		\hline
	\end{tabular}
	
	\vskip.1in 
	
\caption{Coefficients appearing in the $1:3$ resonant Hamiltonian~\eqref{res_ham_13_Gpsi} for the sample cases AS and HA, at $e=10^{-3}$.}\label{tab:coeffs_13} 	
\end{table}

In Figure~\ref{fig:13_ph_por}, we plot the phase portrait of the Hamiltonian~\eqref{res_ham_13_Gpsi} for three different values of the eccentricity. The equilibrium points are given by the solution of the system
\beqno
	\left\{
	\begin{aligned}
\dot \psi & =  \frac{\partial \H_{res}^{13}(G,\psi)}{\partial G} = 0 \\
\dot G & =  - \frac{\partial \H_{res}^{13}(G,\psi)}{\partial \psi} = 0\, 
	\end{aligned}	
	\right. \qquad .
\eeqno
The equilibrium positions, for $e=10^{-3}$, are:
\beqno
 (\psi_1 = 0, \,G_{1} = -0.1201)\,,\quad  (\psi_2 = \pi , \, G_{2} = -0.1200)
\eeqno
for AS and 
\beqno
 (\psi_1 = 0, \,G_{1} = -0.1151 )\,,\quad  (\psi_2 = \pi , \, G_{2} = -0.1143)
\eeqno
for HA. In both cases the first one is an elliptic point, while the second one is an hyperbolic point. For $e=10^{-3}$, we compute the amplitude of the resonant islands around the elliptic solution $(\psi_1,G_{1})$ which is given by 
\beqno
\Delta G_{AS}= 2.0656\cdot 10^{-3}\ ,\qquad \Delta G_{HA} = 8.7360
\cdot 10^{-3} 
\eeqno
(compare with the numerical integrations shown in Figure~\ref{fig:13_ph_por}, top panels).

\begin{figure}[h!]
	\centering
	\includegraphics[scale=0.25]{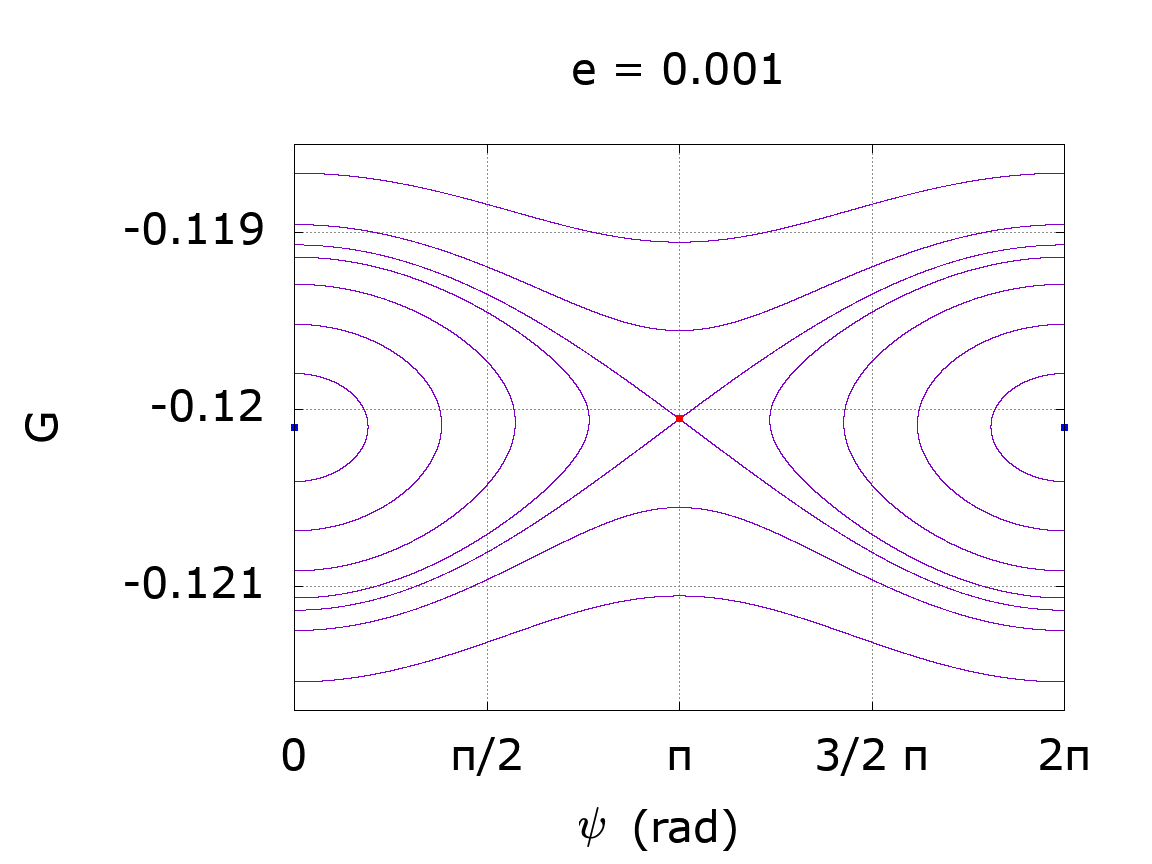}
	\includegraphics[scale=0.25]{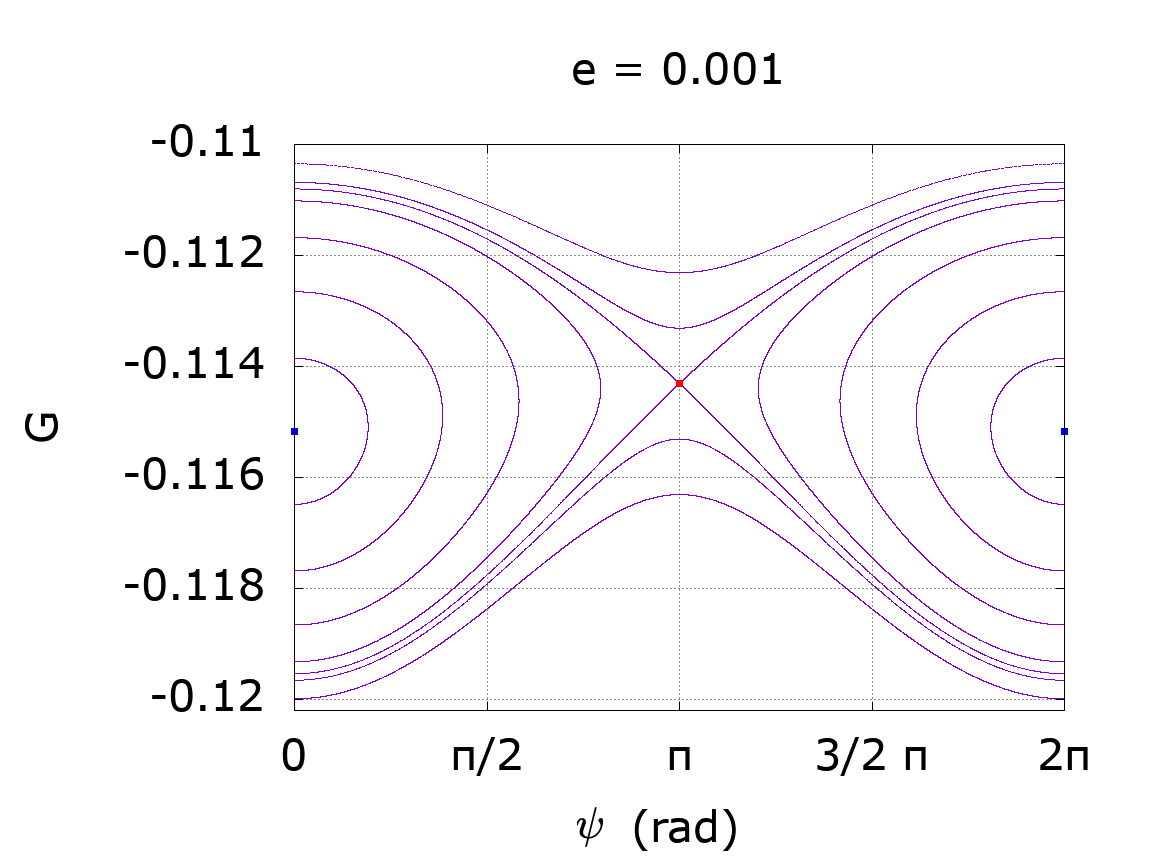}\\
	\includegraphics[scale=0.25]{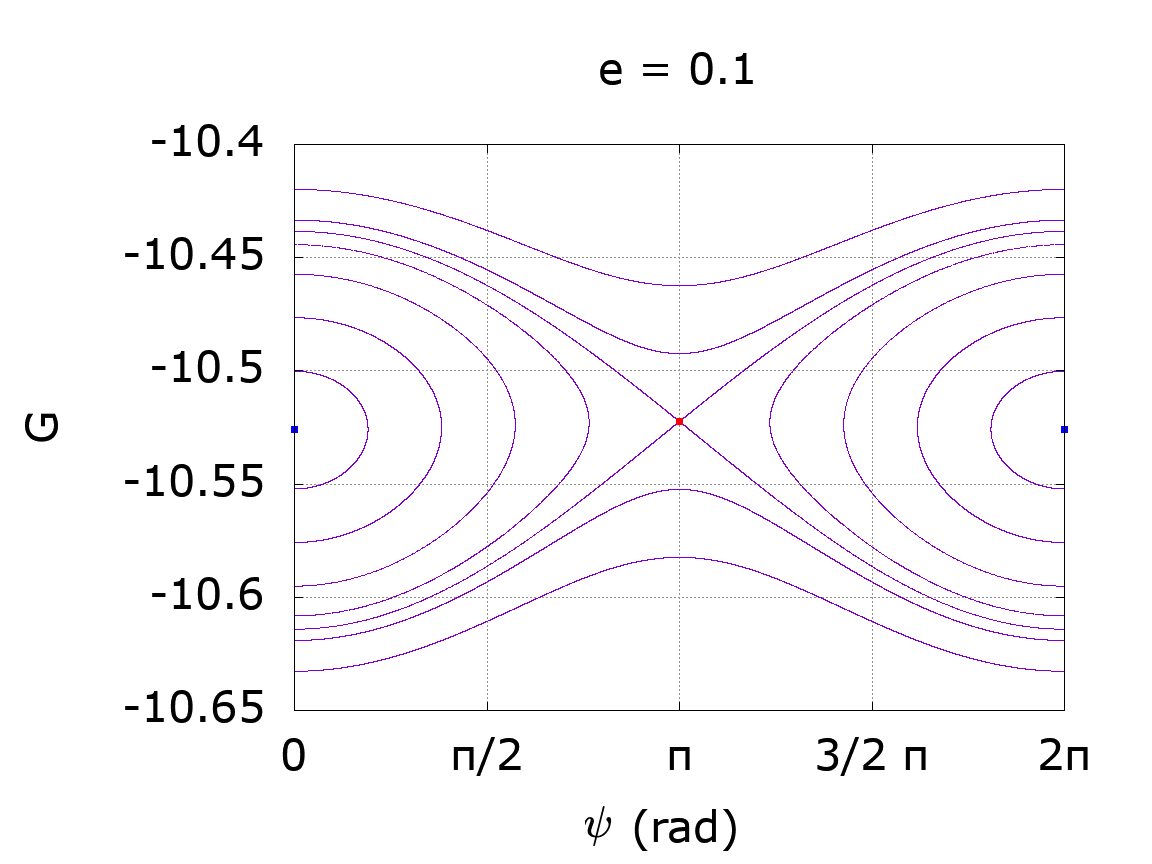}
	\includegraphics[scale=0.25]{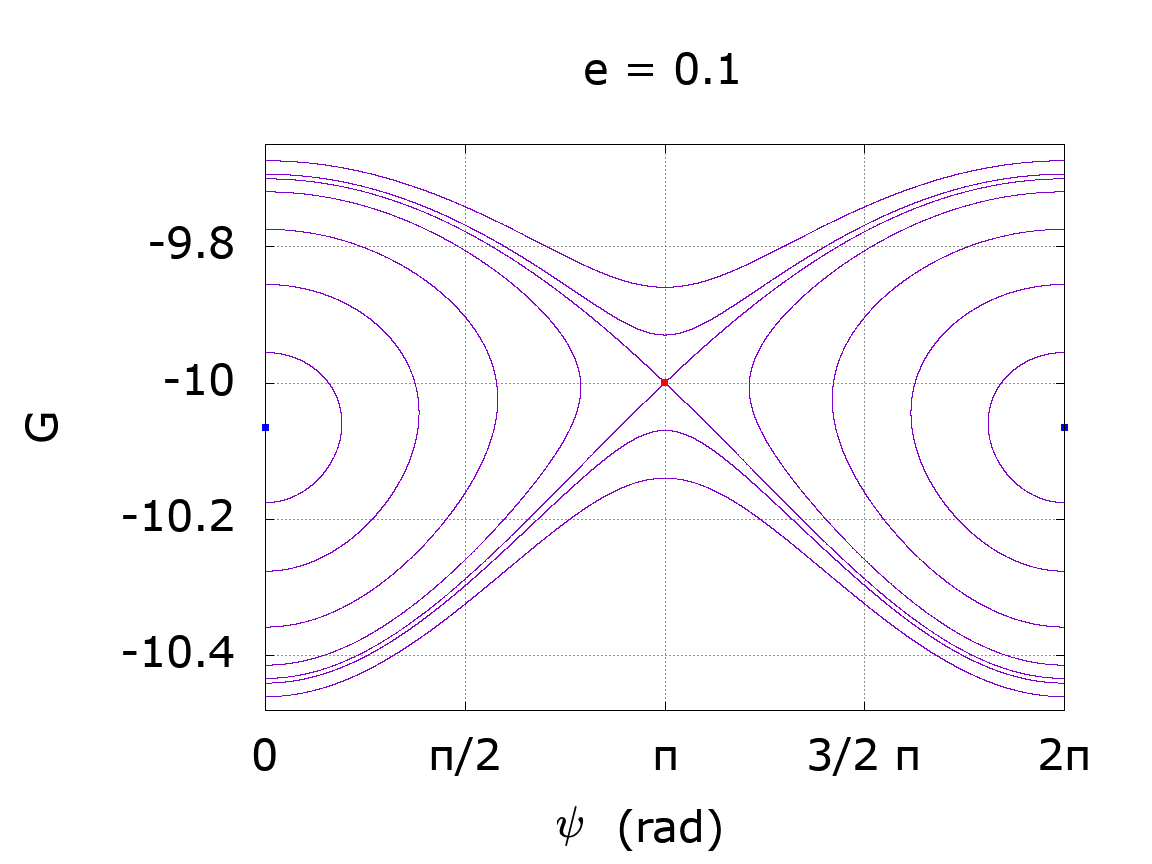}\\
	\includegraphics[scale=0.25]{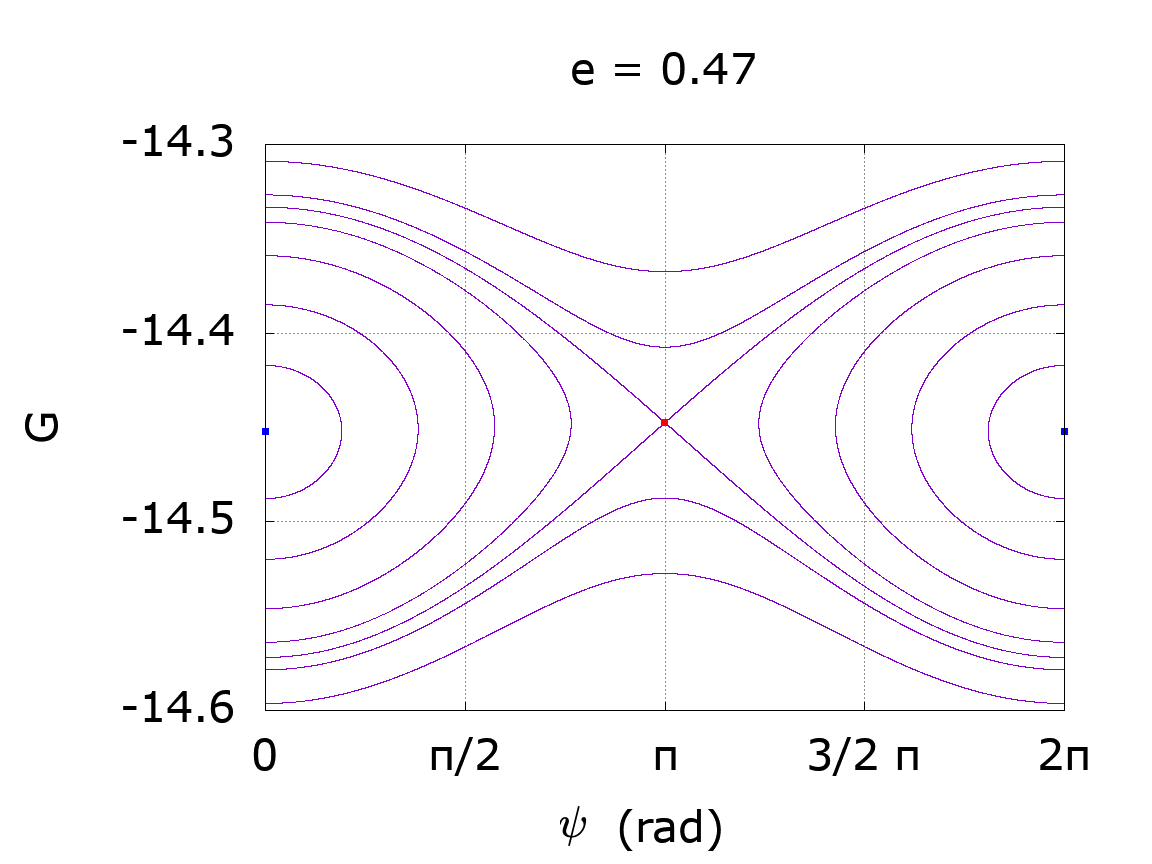}
	\includegraphics[scale=0.25]{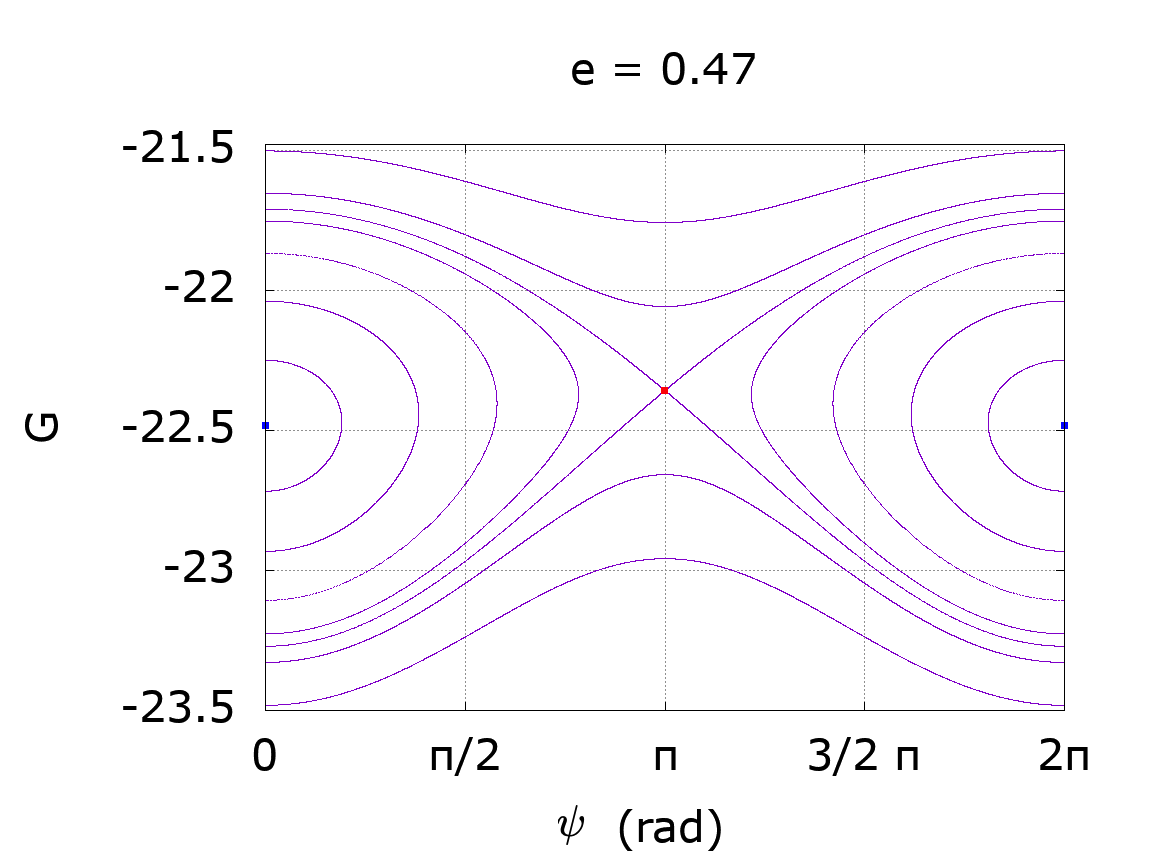}
	\caption{Phase portraits of the $1:3$ resonance for three different values of the eccentricity. Left: almost spherical. Right: highly aspherical. Red points stand for saddle equilibria, blue points stand for centre equilibria.}
	\label{fig:13_ph_por}
\end{figure}

As shown in Figure~\ref{fig:13_G}, for all values of the eccentricity, we find two equilibria associated to Eq.~\eqref{res_ham_13_new}, corresponding to $\psi=0$ and $\psi=\pi$, respectively centre and saddle. Remarkably, contrary to the corotation and $1:2$ resonance, the $1:3$ resonance does not experience bifurcations in the considered range of values of the eccentricity up to 0.5.

\begin{figure}[h!]
		\centering
		\includegraphics[scale=0.38]{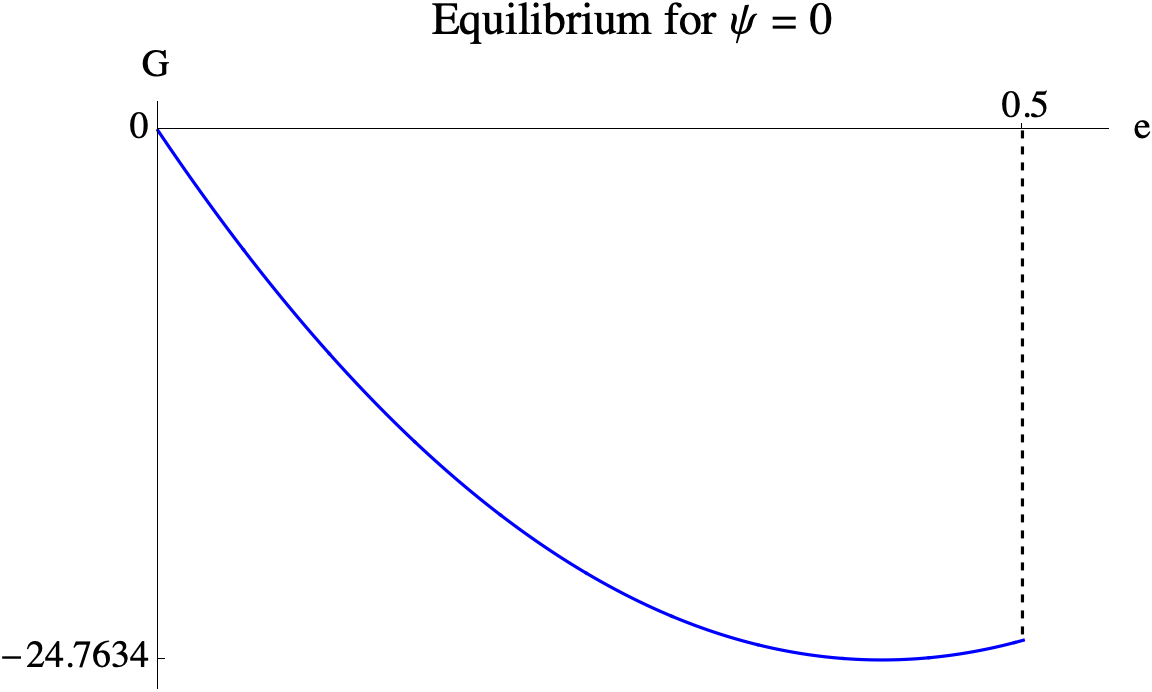}
		\includegraphics[scale=0.38]{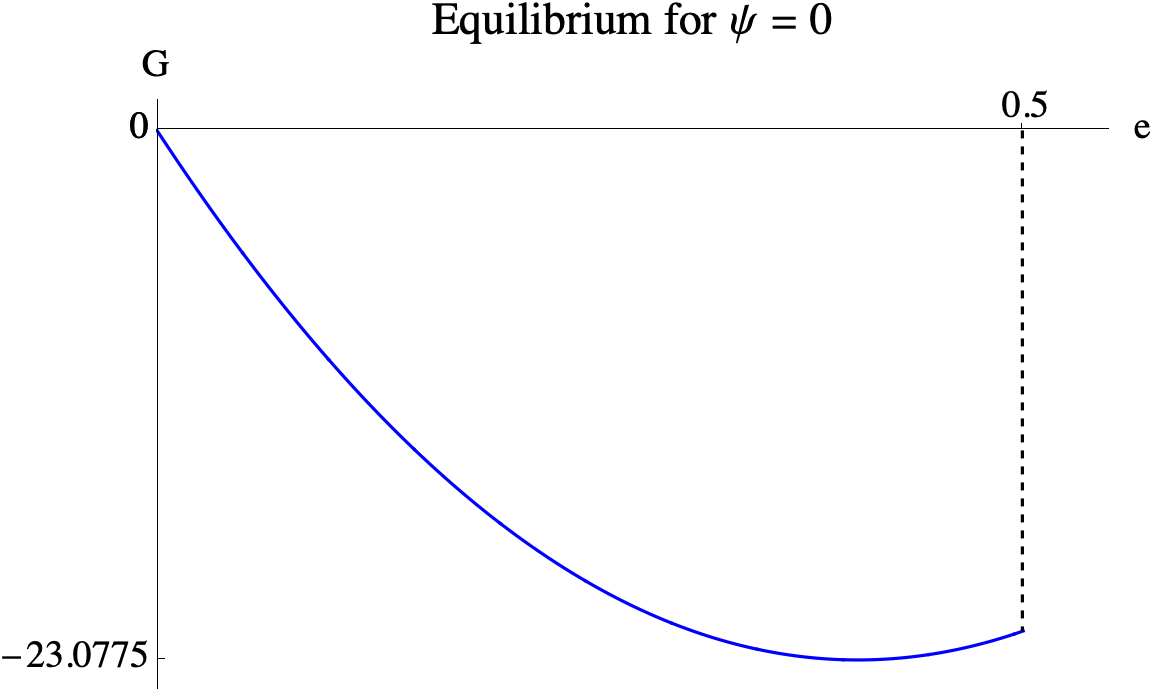}\\
		\includegraphics[scale=0.38]{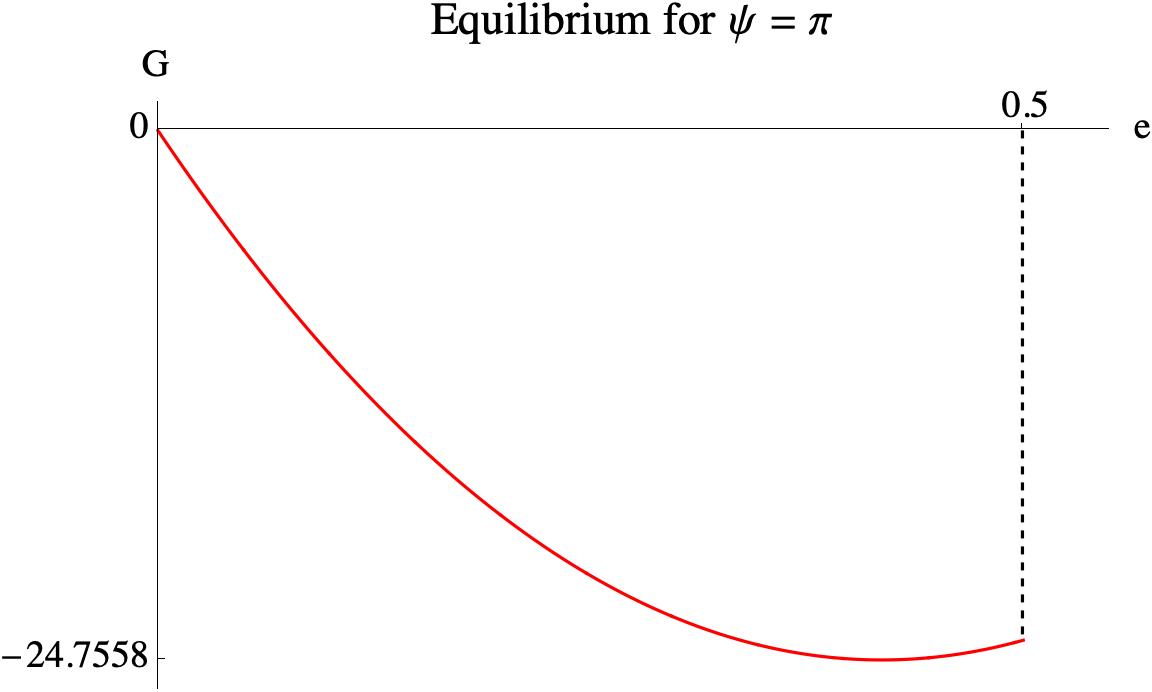}
		\includegraphics[scale=0.38]{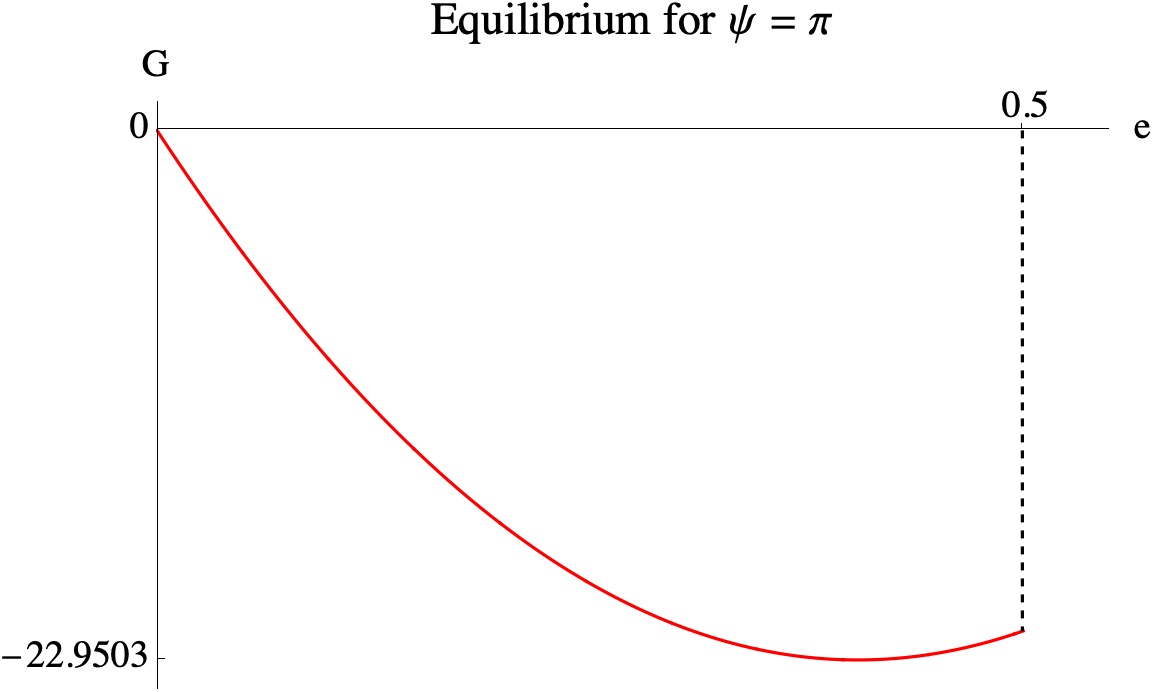}
		\caption{Resonance $1:3$. Left: almost spherical. Right: highly aspherical. Location and stability in the $(e, G)$--plane of the equilibrium for $\psi =0$ (top) and $\psi=\pi$ (bottom) as the eccentricity varies. Blue stands for centre, while red stands for saddle.  }
		\label{fig:13_G}
	\end{figure}

\section{Conclusions}
Once considered a classical topic, ring dynamics has experienced renewed interest due to the identification of ring systems surrounding minor celestial bodies. From a dynamical perspective, two aspects make the problem particularly interesting: first, the irregular shapes of minor bodies need a precise characterization of the resulting gravitational potential. Secondly, the frequent occurrence of ring systems in a $1:3$ resonance demands a thorough exploration of the underlying reasons for the selection of this specific resonance.
In this study, we present components essential for the examination of both inquiries. We introduce the gravitational potential, expressed in the epicyclic variables; we average the resulting Hamiltonian over the fast variables to retain only the terms corresponding to specific resonances. In this procedure, we control the error induced by the expansion in polynomial series of the action variables, bearing in mind that we retain only a finite number of terms. This analysis is performed for the principal resonances, specifically the corotation, as well as the $1:2$ and $1:3$ resonances.
Upon constructing the model, we undertake an analysis of the dynamical features of the resonances. Precisely, we determine the linear stability of the equilibria and we identify bifurcations as the parameters are varied. Specifically, bifurcation values are determined as the eccentricity of the particle within the ring varies over a reasonable range of values. This analysis enables us to draw conclusions regarding the different behavior of the resonances. The corotation resonance is situated in close proximity to the central body, thereby experiencing substantial influence from the gravitational potential, particularly in instances where the body has an irregular shape. Notably, the corotation resonance presents bifurcations at high values of the eccentricity.
The $1:2$ resonance is situated at a greater distance from the central body, and presents  numerous bifurcations, even at low eccentricity values. This could potentially be a reason for the non--selection of the $1:2$ resonance, as it experiences dynamical changes as the eccentricity varies, a scenario that may occur throughout the evolutionary history of the ring system. 
The $1:3$ resonance is located at a greater distance from the central body, thereby resulting in a diminished interaction with the gravitational potential of the irregular body. Furthermore, the $1:3$ resonance does not exhibit any bifurcation as the eccentricity varies, thereby providing enhanced stability relative to other resonances. Consequently, it is plausible to expect that the $1:3$ resonance has been predominantly adopted by ring systems around minor bodies. Specific case studies, precisely Chariklo, Haumea and Quaoar, are investigated in \cite{CDD25} on the basis of the findings of the present work.

\section*{Aknowledgement}
I.D.B. is supported by the Italian Research Center on High Performance Computing Big Data and Quantum Computing (ICSC), project funded by European Union - NextGenerationE and National Recovery and Resilience Plan (NRRP) - Mission
4 Component 2. Spoke 3, Astrophysics and Cosmos Observations.\\
I.D.B. and S.D.R. are partially supported, respectively, by INdAM groups G.N.A.M.P.A and G.N.F.M.



\appendix

\clearpage

\section{Explicit expression of the coefficients $c_{i,j}$, $V_{j,i}$, $\alpha_{i,j}$, $\beta_{i,j}$}\label{app:appA}

In this Appendix, we provide the expressions of the coefficients defined in Section~\ref{sec:epicyclic}. We start with the coefficients $c_{i,j}$ for $i = 0,1,2$ and for $j$ up to 4; we recall that the index $i$ corresponds to the order of the variable $I$ and it can assume only the values 0, 1, 2; the index $j$ corresponds to the order of the variable $\rho$ and it can assume any value from $3-i$ to $\infty$. The maximum value of $j$ appearing in the Hamiltonian~\eqref{HE} corresponds to the order of truncation of the expansion in $\rho$. The expressions of the coefficients $c_{i,j}$ are the following: 
\beqano
c_{0,3} & = & -\,\frac{P_{0,0}\, C_{0,0}\, \G M_P}{r_*^{4}} + \frac{4 \, P_{2,0}\, C_{2,0}\, \G M_P  R^2}{r_*^{6}}+\frac{25 \, P_{4,0}\, C_{4,0}\, \G M_P  R^4}{r_*^{8}}  \nonumber \\
 && + \, \frac{70 \, P_{6,0}\, C_{6,0}\, \G M_P  R^6}{r_*^{10}} \nonumber \\
c_{0,4} & = &  \frac{3 \,P_{0,0}\, C_{0,0}\, \G M_P}{2 \, r_*^{5}} - \frac{15 \, P_{2,0}\, C_{2,0}\, \G M_P  R^2}{2 \,r_*^{7}}-\frac{115 \, P_{4,0}\, C_{4,0}\, \G M_P  R^4}{2 \, r_*^{9}}  \nonumber \\
 && - \,\frac{385 \, P_{6,0}\, C_{6,0}\, \G M_P  R^6}{ 2\, r_*^{11}} \nonumber \\
c_{1,2} & = & \frac{3 \, n_*}{r_*^2}  \qquad 
c_{1,3}  =  - \,\frac{4 \, n_*}{r_*^3} \qquad 
c_{1,4} = \,\frac{5 \, n_*}{r_*^4} \nonumber \\
c_{2,1} & = & - \,\frac{1}{r_*^3} \qquad c_{2,2} =  \,\frac{3}{2 \, r_*^4}  \qquad \ \ 
c_{2,3} = - \,\frac{2}{r_*^5}  \qquad 
c_{2,4} =  \,\frac{5}{2 \, r_*^6}\ . 
\eeqano

Next, we provide the expressions of the coefficients $V_{j,i}$
appearing in Eq.~\eqref{Vji}, for $j$ up to 4 and for $i$ up to 3. We recall that the index $j$ represents the order of the variable $\rho$ thus, in principle, it varies from 0 to $\infty$. The index $i$ is half of the coefficient of the angle $\theta$ appearing in the expansion of the non--axisymmetric part $U_{ns}$. The expressions of the coefficients $V_{j,i}$ are the following: 
\beqano
V_{0,1} & = & -\, \frac{P_{2,2}\, C_{2,2}\, \G M_P  R^2}{r_*^{3}}-\frac{P_{4,2}\, C_{4,2}\, \G M_P  R^4}{r_*^{5}} - \, \frac{P_{6,2}\, C_{6,2}\, \G M_P  R^6}{r_*^{7}} \nonumber \\
 V_{0,2} & = & - \, \frac{P_{4,4}\, C_{4,4}\, \G M_P  R^4}{r_*^{5}} -\frac{P_{6,4}\, C_{6,4}\, \G M_P  R^6}{r_*^{7}} \nonumber \\
 V_{0,3} & = & - \, \frac{P_{6,6}\, C_{6,6}\, \G M_P  R^6}{r_*^{7}}
\eeqano
 
\beqano
 V_{1,1} & = &  \frac{3\,P_{2,2}\, C_{2,2}\, \G M_P  R^2}{r_*^{4}}+\frac{5\,P_{4,2}\, C_{4,2}\, \G M_P  R^4}{r_*^{6}} + \, \frac{7\, P_{6,2}\, C_{6,2}\, \G M_P  R^6}{r_*^{8}} \nonumber \\
 V_{1,2} & = & \frac{5\,P_{4,4}\, C_{4,4}\, \G M_P  R^4}{r_*^{6}} +\frac{7 \,P_{6,4}\, C_{6,4}\, \G M_P  R^6}{r_*^{8}} \nonumber \\
 V_{1,3} & = &  \frac{7\,P_{6,6}\, C_{6,6}\, \G M_P  R^6}{r_*^{8}} \nonumber \\ 
 V_{2,1} & = & -\, \frac{6\,P_{2,2}\, C_{2,2}\, \G M_P  R^2}{r_*^{5}}-\frac{15\,P_{4,2}\, C_{4,2}\, \G M_P  R^4}{r_*^{7}}  - \, \frac{28\,P_{6,2}\, C_{6,2}\, \G M_P  R^6}{r_*^{9}} \nonumber \\
 V_{2,2} & = & - \, \frac{15\,P_{4,4}\, C_{4,4}\, \G M_P  R^4}{r_*^{7}} -\frac{28 \,P_{6,4}\, C_{6,4}\, \G M_P  R^6}{r_*^{9}}\nonumber \\
  V_{2,3} & = & - \, \frac{28\,P_{6,6}\, C_{6,6}\, \G M_P  R^6}{r_*^{9}} 
  \nonumber \\
  V_{3,1} & = & \frac{10\,P_{2,2}\, C_{2,2}\, \G M_P  R^2}{r_*^{6}}+\frac{35\,P_{4,2}\, C_{4,2}\, \G M_P  R^4}{r_*^{8}} + \, \frac{84\,P_{6,2}\, C_{6,2}\, \G M_P  R^6}{r_*^{10}} \nonumber \\
  V_{3,2} & = &  \frac{35\,P_{4,4}\, C_{4,4}\, \G M_P  R^4}{r_*^{8}} +\frac{84 \,P_{6,4}\, C_{6,4}\, \G M_P  R^6}{r_*^{10}} \nonumber \\
  V_{3,3} & = & \frac{84\,P_{6,6}\, C_{6,6}\, \G M_P  R^6}{r_*^{10}} \nonumber \\
  V_{4,1} & = & -\, \frac{15\,P_{2,2}\, C_{2,2}\, \G M_P  R^2}{r_*^{7}}-\frac{70\,P_{4,2}\, C_{4,2}\, \G M_P  R^4}{r_*^{9}} - \, \frac{210\,P_{6,2}\, C_{6,2}\, \G M_P  R^6}{r_*^{11}} \nonumber \\
  V_{4,2} & = & -\, \frac{70\,P_{4,4}\, C_{4,4}\, \G M_P  R^4}{r_*^{9}} -\frac{210 \,P_{6,4}\, C_{6,4}\, \G M_P  R^6}{r_*^{11}} \nonumber \\
  V_{4,3} & = & -\,\frac{210\,P_{6,6}\, C_{6,6}\, \G M_P  R^6}{r_*^{11}} 
\eeqano

Let us now write explicitly the coefficient $\alpha_{2i,j},\beta_{2i,j}$ appearing in Eq.~\eqref{F_tilde_plus_F1_tilde}. Note that the index $2i$ is the coefficient of the angle $\theta$, while the index $j$ is the coefficient of the angle $\varphi$. In this Appendix, $|i|$ varies from 0 to 3 and $j$ from 0 to 4.
We report the expressions only for the non--zero coefficients. 
\beqano
\alpha_{-6,0} & = & -\, \frac{315 \, P_{6,6}\, C_{6,6} \, \G  M_P  R^6 J^{2}}{\kappa_*^{2} \,r_*^{11}} - \frac{28 \, P_{6,6}\, C_{6,6} \, \G  M_P  R^6 J}{ |\kappa_*| \,r_*^9} - \frac{P_{6,6}\, C_{6,6}\, \G  M_P  R^6 }{ r_*^7}  \nonumber \\
\beta_{-6,1} & = &  \frac{63 \sqrt{2}\, P_{6,6}\, C_{6,6} \, \G  M_P  R^6 J^{3/2}}{|\kappa_*|^{3/2} \,r_*^{10}} + \frac{7 \, P_{6,6}\, C_{6,6} \, \G  M_P  R^6 J^{1/2}}{\sqrt{2} \, |\kappa_*|^{1/2} \,r_*^8} \nonumber \\
\eeqano 

\beqano
\alpha_{-6,2} & = & \frac{210 \, P_{6,6}\, C_{6,6} \, \G  M_P  R^6 J^{2}}{\kappa_*^{2} \,r_*^{11}} + \frac{14 \, P_{6,6}\, C_{6,6} \, \G  M_P  R^6 J}{ |\kappa_*| \,r_*^9}  \nonumber \\
\beta_{-6,3} & = & -\, \frac{21 \sqrt{2}\, P_{6,6}\, C_{6,6} \, \G  M_P  R^6 J^{3/2}}{|\kappa_*|^{3/2} \,r_*^{10}}  \nonumber \\
\alpha_{-6,4} & = & -\, \frac{105 \, P_{6,6}\, C_{6,6} \, \G  M_P  R^6 J^{2}}{2\,\kappa_*^{2} \,r_*^{11}} \nonumber \\
\beta_{-6,4} & = &  0\nonumber \\
\alpha_{-4,0} & = & -\,\frac{315\, P_{6,4}\, C_{6,4} \, \G  M_P  R^6 J^2}{\kappa_*^2 \,r_*^{11}} - \frac{105 \, P_{4,4}\, C_{4,4} \, \G  M_P  R^4 J^2}{\kappa_*^2 \,r_*^9} - \frac{28 \, P_{6,4}\, C_{6,4} \, \G  M_P  R^6 J}{|\kappa_*|\, r_*^9} \nonumber \\ 
&&  - \, \frac{15\, P_{4,4}\, C_{4,4}\, \G  M_P  R^4 J}{|\kappa_*|\, r_*^7} - \frac{P_{6,4} \, C_{6,4} \,\G  M_P R^6 }{r_*^7} - \frac{P_{4,4} \, C_{4,4} \,  \G  M_P  R^4}{ r_*^5}  \nonumber \\
\beta_{-4,1} & = & \frac{63 \sqrt{2}\, P_{6,4}\, C_{6,4} \, \G  M_P  R^6 J^{3/2}}{|\kappa_*|^{3/2} \,r_*^{10}} + \frac{105 \, P_{4,4}\, C_{4,4} \, \G  M_P  R^4 J^{3/2}}{2 \sqrt{2} \, |\kappa_*|^{3/2} \,r_*^8} + \frac{7\, P_{6,4}\, C_{6,4}\, \G  M_P  R^6 J^{1/2}}{\sqrt{2} \, |\kappa_*|^{1/2}\,  r_*^8}  \nonumber \\
&& + \, \frac{5 \, P_{4,4}\, C_{4,4} \, \G  M_P  R^4 J^{1/2}}{ \sqrt{2} \, |\kappa_*|^{1/2} \,r_*^6}    \nonumber \\
\alpha_{-4,2} & = & \frac{210\, P_{6,4}\, C_{6,4} \, \G  M_P  R^6 J^2}{\kappa_*^2 \,r_*^{11}} + \frac{70 \, P_{4,4}\, C_{4,4} \, \G  M_P  R^4 J^2}{\kappa_*^2 \,r_*^9} + \frac{14 \, P_{6,4}\, C_{6,4} \, \G  M_P  R^6 J}{|\kappa_*|\, r_*^9} \nonumber \\ 
&&  + \, \frac{15\, P_{4,4}\, C_{4,4}\, \G  M_P  R^4 J}{2\,|\kappa_*|\, r_*^7}   \nonumber \\
\beta_{-4,3} & = &   -\, \frac{21 \sqrt{2}\, P_{6,4}\, C_{6,4} \, \G  M_P  R^6 J^{3/2}}{|\kappa_*|^{3/2} \,r_*^{10}}  - \frac{35 \, P_{4,4}\, C_{4,4} \, \G  M_P  R^4 J^{3/2}}{2 \sqrt{2}|\kappa_*|^{3/2} \,r_*^{8}}  \nonumber \\
\alpha_{-4,4} & = & -\, \frac{105\, P_{6,4}\, C_{6,4} \, \G  M_P  R^6 J^2}{2\, \kappa_*^2 \,r_*^{11}} - \frac{35 \, P_{4,4}\, C_{4,4} \, \G  M_P  R^4 J^2}{2 \, \kappa_*^2 \,r_*^9} \nonumber \\
\alpha_{-2,0} & = & -\,\frac{315\, P_{6,2}\, C_{6,2} \, \G  M_P  R^6 J^2}{\kappa_*^2 \,r_*^{11}} - \frac{105 \, P_{4,2}\, C_{4,2} \, \G  M_P  R^4 J^2}{\kappa_*^2 \,r_*^9} - \frac{45\, P_{2,2}\, C_{2,2}\, \G  M_P  R^2 J^2}{2 \, \kappa_*^2 \, r_*^7}\nonumber \\ 
&&  - \, \frac{28 \, P_{6,2}\, C_{6,2} \, \G  M_P  R^6 J}{|\kappa_*|\, r_*^9}  - \frac{15 P_{4,2} \, C_{4,2} \,\G  M_P R^4 J }{ |\kappa_*|\, r_*^7} - \frac{6 \, P_{2,2} \, C_{2,2} \,\G  M_P R^2 J }{ |\kappa_*|\, r_*^5} \nonumber \\
&&   - \,  \frac{P_{6,2} \, C_{6,2} \,  \G  M_P  R^6}{ r_*^7}  -  \frac{P_{4,2} \, C_{4,2} \,  \G  M_P  R^4}{ r_*^5}  -  \frac{P_{2,2} \, C_{2,2} \,  \G  M_P  R^2}{ r_*^3} \nonumber \\
\beta_{-2,0} & = & 0 \nonumber \\
\alpha_{-2,1} & = & 0 \nonumber \\
\beta_{-2,1} & = &  \frac{63 \sqrt{2}\, P_{6,2}\, C_{6,2} \, \G  M_P  R^6 J^{3/2}}{|\kappa_*|^{3/2} \,r_*^{10}} + \frac{105 \, P_{4,2}\, C_{4,2} \, \G  M_P  R^4 J^{3/2}}{2 \sqrt{2} \, |\kappa_*|^{3/2} \,r_*^8} + \frac{7\, P_{6,2}\, C_{6,2}\, \G  M_P  R^6 J^{1/2}}{\sqrt{2} \, |\kappa_*|^{1/2} r_*^8}  \nonumber \\
&& + \,  \frac{15 \, P_{2,2}\, C_{2,2} \, \G  M_P  R^2 J^{3/2}}{\sqrt{2}\, |\kappa_*|^{3/2} \,r_*^{6}} + \frac{5 \, P_{4,2}\, C_{4,2} \, \G  M_P  R^4 J^{1/2}}{ \sqrt{2} \, |\kappa_*|^{1/2} \,r_*^6} + \frac{3\, P_{2,2}\, C_{2,2}\, \G  M_P  R^2 J^{1/2}}{\sqrt{2} \, |\kappa_*|^{1/2} r_*^4}  \nonumber \\
\eeqano 

\beqano
\alpha_{-2,2} & = & \frac{210\, P_{6,2}\, C_{6,2} \, \G  M_P  R^6 J^2}{\kappa_*^2 \,r_*^{11}} + \frac{70 \, P_{4,2}\, C_{4,2} \, \G  M_P  R^4 J^2}{\kappa_*^2 \,r_*^9} + \frac{14 \, P_{6,2}\, C_{6,2} \, \G  M_P  R^6 J}{|\kappa_*| r_*^9} \nonumber \\ 
&&  + \, \frac{15\, P_{2,2}\, C_{2,2}\, \G  M_P  R^2 J^2}{\kappa_*^2 r_*^7} + \frac{15 \, P_{4,2} \, C_{4,2} \,\G  M_P R^4 J}{2 |\kappa_*| r_*^7} + \frac{3  \,P_{2,2} \, C_{2,2} \,  \G  M_P  R^2 J}{|\kappa_*| r_*^5} \nonumber \\
\beta_{-2,3} & = & -\, \frac{21 \sqrt{2}\, P_{6,2}\, C_{6,2} \, \G  M_P  R^6 J^{3/2}}{|\kappa_*|^{3/2} \,r_*^{10}} - \frac{35 \, P_{4,2}\, C_{4,2} \, \G  M_P  R^4 J^{3/2}}{2 \sqrt{2} \, |\kappa_*|^{3/2} \,r_*^8} - \frac{5\, P_{2,2}\, C_{2,2}\, \G  M_P  R^2 J^{3/2}}{\sqrt{2} \, |\kappa_*|^{3/2} r_*^6}  \nonumber \\
\alpha_{-2,4} & = &  -\,\frac{105\, P_{6,2}\, C_{6,2} \, \G  M_P  R^6 J^2}{2 \,\kappa_*^2 \,r_*^{11}} - \frac{35 \, P_{4,2}\, C_{4,2} \, \G  M_P  R^4 J^2}{2 \, \kappa_*^2 \,r_*^9} - \frac{15\, P_{2,2}\, C_{2,2}\, \G  M_P  R^2 J^2}{4 \, \kappa_*^2 r_*^7}  \nonumber \\
\beta_{0,1} & = & \frac{105 \sqrt{2}\, P_{6,0}\, C_{6,0} \, \G  M_P  R^6 J^{3/2}}{|\kappa_*|^{3/2} \,r_*^{10}} + \frac{75 \, P_{4,0}\, C_{4,0} \, \G  M_P  R^4 J^{3/2}}{\sqrt{2} \, |\kappa_*|^{3/2} \,r_*^8}  + \frac{6\, \sqrt{2}\, P_{2,0}\, C_{2,0}\, \G  M_P  R^2 J^{3/2}}{  |\kappa_*|^{3/2}\, r_*^6} \nonumber \\
&& - \, \frac{3  \, P_{0,0}\, C_{0,0} \, \G  M_P  J^{3/2}}{\sqrt{2}\, |\kappa_*|^{3/2} \,r_*^{4}} - \frac{3 \sqrt{2}\,  I^2 J^{3/2}}{ |\kappa_*|^{3/2} \,r_*^5} -  \frac{6 \sqrt{2}\,  n_* \, I \, J^{3/2}}{ |\kappa_*|^{3/2} \,r_*^3}- \frac{ \sqrt{2}\,  I^2 J^{1/2}}{ |\kappa_*|^{1/2} \,r_*^3} - \frac{2 \sqrt{2}\, n_* I J^{1/2}}{ |\kappa_*|^{1/2} \,r_*}\nonumber \\
\alpha_{0,2} & = &  \frac{385\, P_{6,0}\, C_{6,0} \, \G  M_P  R^6 J^2}{\kappa_*^2 \,r_*^{11}} + \frac{115 \, P_{4,0}\, C_{4,0} \, \G  M_P  R^4 J^2}{\kappa_*^2 \,r_*^9}    +  \frac{15\, P_{2,0}\, C_{2,0}\, \G  M_P  R^2 J^2}{\kappa_*^2\, r_*^7} \nonumber \\ 
&& - \, \frac{3 \, P_{0.0} \, C_{0,0} \,\G  M_P  J^2}{\kappa_*^2\, r_*^5} - \frac{5  \, I^2 J^2 }{\kappa_*^2\, r_*^6}- \frac{10\, n_*  \, I J^2 }{\kappa_*^2\, r_*^4}-\frac{3  \, I^2 J }{2 |\kappa_*|\, r_*^4}-\frac{3\,n_*  \, I J }{|\kappa_*|\, r_*^2}\nonumber \\
\beta_{0,3} & = & -\, \frac{35 \sqrt{2}\, P_{6,0}\, C_{6,0} \, \G  M_P  R^6 J^{3/2}}{|\kappa_*|^{3/2} \,r_*^{10}} - \frac{25 \, P_{4,0}\, C_{4,0} \, \G  M_P  R^4 J^{3/2}}{\sqrt{2} \, |\kappa_*|^{3/2} \,r_*^8}  - \frac{2\, \sqrt{2}\, P_{2,0}\, C_{2,0}\, \G  M_P  R^2 J^{3/2}}{  |\kappa_*|^{3/2}\, r_*^6} \nonumber \\
&& + \, \frac{  P_{0,0}\, C_{0,0} \, \G  M_P  J^{3/2}}{\sqrt{2}\, |\kappa_*|^{3/2} \,r_*^{4}} + \frac{ \sqrt{2}\,  I^2 J^{3/2}}{ |\kappa_*|^{3/2} \,r_*^5} +  \frac{2 \sqrt{2}\,  n_* \, I \, J^{3/2}}{ |\kappa_*|^{3/2} \,r_*^3} \nonumber \\
\alpha_{0,4} & = & -\,\frac{385\, P_{6,0}\, C_{6,0} \, \G  M_P  R^6 J^2}{4\, \kappa_*^2 \,r_*^{11}} - \frac{115 \, P_{4,0}\, C_{4,0} \, \G  M_P  R^4 J^2}{4 \, \kappa_*^2 \,r_*^9}    -  \frac{15\, P_{2,0}\, C_{2,0}\, \G  M_P  R^2 J^2}{4 \, \kappa_*^2\, r_*^7} \nonumber \\ 
&& + \, \frac{3 \, P_{0.0} \, C_{0,0} \,\G  M_P  J^2}{4\, \kappa_*^2\, r_*^5} + \frac{5  \, I^2 J^2 }{4\, \kappa_*^2\, r_*^6}+ \frac{5\, n_*  \, I J^2 }{2\, \kappa_*^2\, r_*^4}\nonumber\\
\alpha_{2,0} & = & -\,\frac{315\, P_{6,2}\, C_{6,2} \, \G  M_P  R^6 J^2}{\kappa_*^2 \,r_*^{11}} - \frac{105 \, P_{4,2}\, C_{4,2} \, \G  M_P  R^4 J^2}{\kappa_*^2 \,r_*^9} - \frac{45\, P_{2,2}\, C_{2,2}\, \G  M_P  R^2 J^2}{2 \, \kappa_*^2 \, r_*^7}\nonumber \\ 
&&  - \, \frac{28 \, P_{6,2}\, C_{6,2} \, \G  M_P  R^6 J}{|\kappa_*|\, r_*^9}  - \frac{15 P_{4,2} \, C_{4,2} \,\G  M_P R^4 J }{ |\kappa_*|\, r_*^7} - \frac{6 \, P_{2,2} \, C_{2,2} \,\G  M_P R^2 J }{ |\kappa_*|\, r_*^5} \nonumber \\
&&   - \,  \frac{P_{6,2} \, C_{6,2} \,  \G  M_P  R^6}{ r_*^7}  -  \frac{P_{4,2} \, C_{4,2} \,  \G  M_P  R^4}{ r_*^5}  -  \frac{P_{2,2} \, C_{2,2} \,  \G  M_P  R^2}{ r_*^3} \nonumber \\
\beta_{2,1} & = &  \frac{63 \sqrt{2}\, P_{6,2}\, C_{6,2} \, \G  M_P  R^6 J^{3/2}}{|\kappa_*|^{3/2} \,r_*^{10}} + \frac{105 \, P_{4,2}\, C_{4,2} \, \G  M_P  R^4 J^{3/2}}{2 \sqrt{2} \, |\kappa_*|^{3/2} \,r_*^8} + \frac{7\, P_{6,2}\, C_{6,2}\, \G  M_P  R^6 J^{1/2}}{\sqrt{2} \, |\kappa_*|^{1/2} r_*^8}  \nonumber \\
&& + \,  \frac{15 \, P_{2,2}\, C_{2,2} \, \G  M_P  R^2 J^{3/2}}{\sqrt{2}\, |\kappa_*|^{3/2} \,r_*^{6}} + \frac{5 \, P_{4,2}\, C_{4,2} \, \G  M_P  R^4 J^{1/2}}{ \sqrt{2} \, |\kappa_*|^{1/2} \,r_*^6} + \frac{3\, P_{2,2}\, C_{2,2}\, \G  M_P  R^2 J^{1/2}}{\sqrt{2} \, |\kappa_*|^{1/2} r_*^4}  \nonumber \\
\eeqano

\beqano
\alpha_{2,2} & = & \frac{210\, P_{6,2}\, C_{6,2} \, \G  M_P  R^6 J^2}{\kappa_*^2 \,r_*^{11}} + \frac{70 \, P_{4,2}\, C_{4,2} \, \G  M_P  R^4 J^2}{\kappa_*^2 \,r_*^9} + \frac{14 \, P_{6,2}\, C_{6,2} \, \G  M_P  R^6 J}{|\kappa_*| r_*^9} \nonumber \\ 
&&  + \, \frac{15\, P_{2,2}\, C_{2,2}\, \G  M_P  R^2 J^2}{\kappa_*^2 r_*^7} + \frac{15 \, P_{4,2} \, C_{4,2} \,\G  M_P R^4 J}{2 |\kappa_*| r_*^7} + \frac{3  \,P_{2,2} \, C_{2,2} \,  \G  M_P  R^2 J}{|\kappa_*| r_*^5} \nonumber \\
\beta_{2,3} & = &  -\, \frac{21 \sqrt{2}\, P_{6,2}\, C_{6,2} \, \G  M_P  R^6 J^{3/2}}{|\kappa_*|^{3/2} \,r_*^{10}} - \frac{35 \, P_{4,2}\, C_{4,2} \, \G  M_P  R^4 J^{3/2}}{2 \sqrt{2} \, |\kappa_*|^{3/2} \,r_*^8} - \frac{5\, P_{2,2}\, C_{2,2}\, \G  M_P  R^2 J^{3/2}}{\sqrt{2} \, |\kappa_*|^{3/2} r_*^6}  \nonumber \\
\alpha_{2,4} & = &  -\,\frac{105\, P_{6,2}\, C_{6,2} \, \G  M_P  R^6 J^2}{2 \,\kappa_*^2 \,r_*^{11}} - \frac{35 \, P_{4,2}\, C_{4,2} \, \G  M_P  R^4 J^2}{2 \, \kappa_*^2 \,r_*^9} - \frac{15\, P_{2,2}\, C_{2,2}\, \G  M_P  R^2 J^2}{4 \, \kappa_*^2 r_*^7}  \nonumber \\
\alpha_{4,0} & = & -\,\frac{315\, P_{6,4}\, C_{6,4} \, \G  M_P  R^6 J^2}{\kappa_*^2 \,r_*^{11}} - \frac{105 \, P_{4,4}\, C_{4,4} \, \G  M_P  R^4 J^2}{\kappa_*^2 \,r_*^9} - \frac{28 \, P_{6,4}\, C_{6,4} \, \G  M_P  R^6 J}{|\kappa_*| \, r_*^9} \nonumber \\ 
&&  - \, \frac{15\, P_{4,4}\, C_{4,4}\, \G  M_P  R^4 J}{|\kappa_*|\, r_*^7} - \frac{P_{6,4} \, C_{6,4} \,\G  M_P R^6 }{r_*^7} - \frac{P_{4,4} \, C_{4,4} \,  \G  M_P  R^4}{ r_*^5} \nonumber \\
\beta_{4,1} & = & \frac{63 \sqrt{2}\, P_{6,4}\, C_{6,4} \, \G  M_P  R^6 J^{3/2}}{|\kappa_*|^{3/2} \,r_*^{10}} + \frac{105 \, P_{4,4}\, C_{4,4} \, \G  M_P  R^4 J^{3/2}}{2 \sqrt{2} \, |\kappa_*|^{3/2} \,r_*^8} + \frac{7\, P_{6,4}\, C_{6,4}\, \G  M_P  R^6 J^{1/2}}{\sqrt{2} \, |\kappa_*|^{1/2}\,  r_*^8}  \nonumber \\
&& + \, \frac{5 \, P_{4,4}\, C_{4,4} \, \G  M_P  R^4 J^{1/2}}{ \sqrt{2} \, |\kappa_*|^{1/2} \,r_*^6}     \nonumber \\
\alpha_{4,2} & = & \frac{210\, P_{6,4}\, C_{6,4} \, \G  M_P  R^6 J^2}{\kappa_*^2 \,r_*^{11}} + \frac{70 \, P_{4,4}\, C_{4,4} \, \G  M_P  R^4 J^2}{\kappa_*^2 \,r_*^9} + \frac{14 \, P_{6,4}\, C_{6,4} \, \G  M_P  R^6 J}{|\kappa_*|\, r_*^9} \nonumber \\ 
&&  + \, \frac{15\, P_{4,4}\, C_{4,4}\, \G  M_P  R^4 J}{2\,|\kappa_*|\, r_*^7}\nonumber \\
\beta_{4,3} & = &  -\, \frac{21 \sqrt{2}\, P_{6,4}\, C_{6,4} \, \G  M_P  R^6 J^{3/2}}{|\kappa_*|^{3/2} \,r_*^{10}}  - \frac{35 \, P_{4,4}\, C_{4,4} \, \G  M_P  R^4 J^{3/2}}{2 \sqrt{2}|\kappa_*|^{3/2} \,r_*^{8}} \nonumber \\
\alpha_{4,4} & = & -\, \frac{105\, P_{6,4}\, C_{6,4} \, \G  M_P  R^6 J^2}{2\, \kappa_*^2 \,r_*^{11}} - \frac{35 \, P_{4,4}\, C_{4,4} \, \G  M_P  R^4 J^2}{2 \, \kappa_*^2 \,r_*^9} \nonumber \\
\alpha_{6,0} & = & -\, \frac{315 \, P_{6,6}\, C_{6,6} \, \G  M_P  R^6 J^{2}}{\kappa_*^{2} \,r_*^{11}} - \frac{28 \, P_{6,6}\, C_{6,6} \, \G  M_P  R^6 J}{ |\kappa_*| \,r_*^9} - \frac{P_{6,6}\, C_{6,6}\, \G  M_P  R^6 }{ r_*^7}  \nonumber \\
\beta_{6,1} & = &   \frac{63 \sqrt{2}\, P_{6,6}\, C_{6,6} \, \G  M_P  R^6 J^{3/2}}{|\kappa_*|^{3/2} \,r_*^{10}} + \frac{7 \, P_{6,6}\, C_{6,6} \, \G  M_P  R^6 J^{1/2}}{\sqrt{2} \, |\kappa_*|^{1/2} \,r_*^8} \nonumber \\
\alpha_{6,2} & = & \frac{210 \, P_{6,6}\, C_{6,6} \, \G  M_P  R^6 J^{2}}{\kappa_*^{2} \,r_*^{11}} + \frac{14 \, P_{6,6}\, C_{6,6} \, \G  M_P  R^6 J}{ |\kappa_*| \,r_*^9}  \nonumber \\\
\beta_{6,3} & = &  -\, \frac{21 \sqrt{2}\, P_{6,6}\, C_{6,6} \, \G  M_P  R^6 J^{3/2}}{|\kappa_*|^{3/2} \,r_*^{10}} \nonumber \\
\alpha_{6,4} & = & -\, \frac{105 \, P_{6,6}\, C_{6,6} \, \G  M_P  R^6 J^{2}}{2\,\kappa_*^{2} \,r_*^{11}}\ . 
\eeqano

\section{Values of the coefficients associated to the $1:2$ and $1:3$ resonances.}\label{app:appB}

The polynomials $\delta_{1},...,\delta_{5}$ for the $1:2$ resonance (see Eq.~\equ{res_ham_12_Gpsi}) for the AS case are given by the following expressions:

\beqano
\delta_1(G) & = & -\,2.1\cdot 10^{-7} + 1.2\cdot 10^{-6}\, G + 3.6\cdot 10^{-8}\, G^2  + 7.6\cdot 10^{-10}\, G^3 \nonumber \\ 
&& + \, 1.3\cdot 10^{-11}\, G^4+ 2.2\cdot 10^{-13}\, G^5 + 3.5\cdot 10^{-15}\, G^6  \nonumber \\ 
&& + \, 5.3\cdot 10^{-17}\, G^7 + 8.0\cdot 10^{-19} \, G^8 \, , \nonumber \\ 
\delta_2(G) & = & -\,4.8\cdot 10^{-12} + 5.6\cdot 10^{-11}\, G - 1.6\cdot 10^{-10}\, G^2-8.8\cdot 10^{-12}\, G^3\nonumber \\ 
&& - \, 2.8\cdot 10^{-13}\, G^4-7.3\cdot 10^{-15}\, G^5 -1.6\cdot 10^{-16}\, G^6  \nonumber \\ 
&& - \, 3.3\cdot 10^{-18}\, G^7 - 6.5\cdot 10^{-20} \, G^8 \, , \nonumber \\ 
\delta_3(G) & = & -\,1.8\cdot 10^{-16} + 3.2\cdot 10^{-15}\, G - 1.8\cdot 10^{-14}\, G^2+3.5\cdot 10^{-14}\, G^3\nonumber \\ 
&& + \, 2.7\cdot 10^{-15}\, G^4+1.2\cdot 10^{-16}\, G^5 +4.1\cdot 10^{-18}\, G^6  \nonumber \\ 
&& + \, 1.1\cdot 10^{-19}\, G^7 + 3.0\cdot 10^{-21} \, G^8 \, , \nonumber \\  
\delta_4(G) & = & -\,8.8\cdot 10^{-21} + 2.0\cdot 10^{-19}\, G - 1.8\cdot 10^{-18}\, G^2+6.9\cdot 10^{-18}\, G^3\nonumber \\ 
&& - \, 9.8\cdot 10^{-18}\, G^4-1.0\cdot 10^{-18}\, G^5 -5.6\cdot 10^{-20}\, G^6  \nonumber \\ 
&& - \, 2.3\cdot 10^{-21}\, G^7 - 8.3\cdot 10^{-23} \, G^8  \, ,\nonumber \\ 
\delta_5(G) & = & -\,4.6\cdot 10^{-25} + 1.3\cdot 10^{-23}\, G - 1.6\cdot 10^{-22}\, G^2+9.3\cdot 10^{-22}\, G^3\nonumber \\ 
&& - \, 2.6\cdot 10^{-21}\, G^4+2.9\cdot 10^{-21}\, G^5 +3.7\cdot 10^{-22}\, G^6  \nonumber \\ 
&& + \, 2.5\cdot 10^{-23}\, G^7 +1.3\cdot 10^{-24} \, G^8 \, , 
\eeqano
while for the HA case, they are given by 
\beqano
\delta_1(G) & = & -\,3.4\cdot 10^{-6} + 2.2\cdot 10^{-5}\, G + 8.0\cdot 10^{-7}\, G^2  + 2.1\cdot 10^{-8}\, G^3 \nonumber \\ 
&& + \, 5.4\cdot 10^{-10}\, G^4+ 1.3\cdot 10^{-11}\, G^5 + 3.3\cdot 10^{-13}\, G^6  \nonumber \\ 
&& + \, 8.0\cdot 10^{-15}\, G^7 + 2.0\cdot 10^{-16} \, G^8 \, , \nonumber \\ 
\delta_2(G) & = & -\,1.2\cdot 10^{-9} + 1.6\cdot 10^{-8}\, G - 5.1\cdot 10^{-8}\, G^2-3.3\cdot 10^{-9}\, G^3\nonumber \\ 
&& - \, 1.3\cdot 10^{-10}\, G^4-4.5\cdot 10^{-12}\, G^5 -1.3\cdot 10^{-13}\, G^6  \nonumber \\ 
&& - \, 3.7\cdot 10^{-15}\, G^7 - 1.0\cdot 10^{-16} \, G^8 \, , \nonumber \\ 
\eeqano

\beqano
\delta_3(G) & = & -\,6.8\cdot 10^{-13} + 1.3\cdot 10^{-11}\, G - 8.5\cdot 10^{-11}\, G^2+1.8\cdot 10^{-10}\, G^3\nonumber \\ 
&& + \, 1.6\cdot 10^{-11}\, G^4+8.6\cdot 10^{-13}\, G^5 +3.4\cdot 10^{-14}\, G^6  \nonumber \\ 
&& + \, 1.1\cdot 10^{-15}\, G^7 + 3.6\cdot 10^{-17} \, G^8 \, , \nonumber \\  
\delta_4(G) & = & -\,4.2\cdot 10^{-16} + 1.1\cdot 10^{-14}\, G - 1.0\cdot 10^{-13}\, G^2+4.5\cdot 10^{-13}\, G^3\nonumber \\ 
&& - \, 7.0\cdot 10^{-13}\, G^4-7.9\cdot 10^{-14}\, G^5 -4.9\cdot 10^{-15}\, G^6  \nonumber \\ 
&& - \, 2.2\cdot 10^{-16}\, G^7 - 8.8\cdot 10^{-18} \, G^8  \, ,\nonumber \\ 
\delta_5(G) & = & -\,2.3\cdot 10^{-19} + 7.5\cdot 10^{-18}\, G - 9.7\cdot 10^{-17}\, G^2+6.2\cdot 10^{-16}\, G^3\nonumber \\ 
&& - \, 1.9\cdot 10^{-15}\, G^4+2.4\cdot 10^{-15}\, G^5 +3.2\cdot 10^{-16}\, G^6  \nonumber \\ 
&& + \, 2.2\cdot 10^{-17}\, G^7 +1.2\cdot 10^{-18} \, G^8 \, . 
\eeqano

\vskip.1in 

The polynomials $\delta_1,...,\delta_5$ for the $1:3$ resonance (see Eq.~\equ{res_ham_13_Gpsi}) for the AS case are given by 
\beqano
\delta_1(G) & = & -\,1.1\cdot 10^{-10} - 2.4\cdot 10^{-9}\, G - 1.2\cdot 10^{-8}\, G^2  - 7.2\cdot 10^{-10}\, G^3 \nonumber \\ 
&& - \, 2.8\cdot 10^{-11}\, G^4-9.1\cdot 10^{-13}\, G^5 -2.6\cdot 10^{-14}\, G^6  \nonumber \\ 
&& - \, 7.4\cdot 10^{-16}\, G^7 -1.9\cdot 10^{-17} \, G^8 \, , \nonumber \\ 
\delta_2(G) & = & -\,2.2\cdot 10^{-18} -9.3\cdot 10^{-17}\, G - 1.4\cdot 10^{-15}\, G^2-9.8\cdot 10^{-15}\, G^3\nonumber \\ 
&& - \, 2.5\cdot 10^{-14}\, G^4-2.6\cdot 10^{-15}\, G^5 -1.6\cdot 10^{-16}\, G^6  \nonumber \\ 
&& - \, 7.7\cdot 10^{-18}\, G^7 - 3.0\cdot 10^{-19} \, G^8 \, , \nonumber \\ 
\delta_3(G) & = & -\,7.2\cdot 10^{-26} -4.4\cdot 10^{-24}\, G - 1.1\cdot 10^{-22}\, G^2-1.5\cdot 10^{-21}\, G^3\nonumber \\ 
&& - \, 1.2\cdot 10^{-20}\, G^4-4.9\cdot 10^{-20}\, G^5 -8.9\cdot 10^{-20}\, G^6  \nonumber \\ 
&& - \, 1.2\cdot 10^{-20}\, G^7 -1.0\cdot 10^{-21} \, G^8 \, , \nonumber \\  
\delta_4(G) & = & -\,2.8\cdot 10^{-33} -2.3\cdot 10^{-31}\, G - 8.2\cdot 10^{-30}\, G^2-1.6\cdot 10^{-28}\, G^3\nonumber \\ 
&& - \, 2.1\cdot 10^{-27}\, G^4-1.7\cdot 10^{-26}\, G^5 -9.0\cdot 10^{-26}\, G^6  \nonumber \\ 
&& - \, 2.6\cdot 10^{-25}\, G^7 - 3.3\cdot 10^{-25} \, G^8  \, ,
\eeqano
and for the HA case, they take the following expression
\beqano
\delta_1(G) & = & -\,1.9\cdot 10^{-9} - 4.2\cdot 10^{-8}\, G - 2.3\cdot 10^{-7}\, G^2  -1.5\cdot 10^{-8}\, G^3 \nonumber \\ 
&& - \, 7.0\cdot 10^{-10}\, G^4-2.8\cdot 10^{-11}\, G^5 -1.0\cdot 10^{-12}\, G^6  \nonumber \\ 
&& - \, 3.9\cdot 10^{-14}\, G^7 -1.4\cdot 10^{-15} \, G^8 \, , \nonumber \\ 
\delta_2(G) & = & -\,6.1\cdot 10^{-16} -2.6\cdot 10^{-14}\, G - 4.2\cdot 10^{-13}\, G^2-3.1\cdot 10^{-12}\, G^3\nonumber \\ 
&& - \, 8.5\cdot 10^{-12}\, G^4-1.0\cdot 10^{-12}\, G^5 -7.3\cdot 10^{-14}\, G^6  \nonumber \\ 
&& - \, 4.2\cdot 10^{-15}\, G^7 - 2.0\cdot 10^{-16} \, G^8 \, , \nonumber \\ 
\delta_3(G) & = & -\,2.9\cdot 10^{-22} -1.8\cdot 10^{-20}\, G - 5.1\cdot 10^{-19}\, G^2-7.3\cdot 10^{-18}\, G^3\nonumber \\ 
&& - \, 6.0\cdot 10^{-17}\, G^4-2.6\cdot 10^{-16}\, G^5 -4.9\cdot 10^{-16}\, G^6  \nonumber \\ 
&& - \, 8.0\cdot 10^{-17}\, G^7 -7.1\cdot 10^{-18} \, G^8 \, , \nonumber \\  
\delta_4(G) & = & -\,1.4\cdot 10^{-28} -1.2\cdot 10^{-26}\, G - 4.7\cdot 10^{-25}\, G^2-1.0\cdot 10^{-23}\, G^3\nonumber \\ 
&& - \, 1.3\cdot 10^{-22}\, G^4-1.1\cdot 10^{-21}\, G^5 -6.3\cdot 10^{-21}\, G^6  \nonumber \\ 
&& - \, 1.9\cdot 10^{-20}\, G^7 - 2.6\cdot 10^{-20} \, G^8  \, .
\eeqano


\printbibliography

\end{document}